\documentclass[11pt]{article}

\newif\ifjrss
\jrsstrue

\usepackage{amssymb,amsbsy,amsfonts,amsmath,amsthm,xspace}
\usepackage{lscape}
\usepackage[round]{natbib}
\usepackage{graphicx}
\usepackage{subcaption}

\usepackage{color}
\usepackage[hidelinks]{hyperref}
\usepackage{bm}
\usepackage{comment}
\usepackage[ruled,linesnumbered]{algorithm2e}
\usepackage[doublespacing]{setspace}
\usepackage{booktabs}
\usepackage{authblk}

\newcommand{\bc}{\bm{c}}

\newcommand{\bu}{\bm{u}}

\newcommand{\bbeta}{\bm{\beta}}

\newcommand{\bbE}{\mathbb{E}}

\newcommand{\bbP}{\mathbb{P}}

\newcommand{\bbR}{\mathbb{R}}

\newcommand{\cC}{\mathcal{C}}

\newcommand{\cN}{\mathcal{N}}

\theoremstyle{definition}
\newtheorem{theorem}{Theorem}
\newtheorem{definition}{Definition}
\newtheorem{proposition}{Proposition}
\newtheorem{corollary}{Corollary}
\newtheorem{lemma}{Lemma}
\newtheorem{remark}{Remark}

\newtheorem{assumption}{Assumption}

\newcommand{\convd}{\stackrel{d}{\longrightarrow}}
\newcommand{\convp}{\stackrel{P}{\longrightarrow}}

\newcommand{\sfvar}{\mathsf{Var}}

\newcommand{\cov}{\mathsf{Cov}}

\newcommand{\inprob}{\overset{P}{\rightarrow}}
\newcommand{\tabby}{\hspace{10pt}}

\newcommand{\tp}{\intercal}

\newcommand{\prob}{\mathbb{P}}

\newcommand{\hui}[1]{\textcolor{red}{#1}}

\begin{document}


\title{Minority representation and fairness in network ranking: An application to school contact diary data}

\author[1]{Hui Shen}
\author[2]{Peter W. MacDonald}
\author[3]{Eric D. Kolaczyk}
\affil[1]{McGill University, \url{hui.shen2@mail.mcgill.ca}}
\affil[2]{University of Waterloo, \url{pwmacdonald@uwaterloo.ca}}
\affil[3]{McGill University, \url{eric.kolaczyk@mcgill.ca}}
\date{}

\maketitle

\begin{abstract}
    Considerations of bias, fairness and representation are a prerequisite of responsible modern statistics.
    In statistical network analysis, observed networks are often incomplete or systematically biased, which can lead to systematic underrepresentation of protected groups, and affect any downstream ranking or decision based on the observed network. 
    In this paper, we study a high school contact network constructed from self-reported contact diaries and introduce a formal measure of minority representation, defined as the proportion of minority nodes among the top-ranked individuals.  
    We model systematic bias through group-dependent missing edge mechanisms and develop statistical methods to estimate and test for such bias. 
    When bias is detected, we propose a re-ranking procedure based on an asymptotic approximation that improves group representation.  
    Applying the framework to the high school contact network reveals systematic underreporting of cross-group contacts consistent with recall bias.  
    These findings highlight the importance of modeling and correcting systematic bias in social networks with heterogeneous groups.

    \vspace{0.1in}
\end{abstract}


\section{Introduction} \label{sec:intro}

\subsection{Background and Motivation}
Statistical decision making is more prevalent than ever in our modern world.
When decisions impact the health and livelihood of individuals, issues of bias, fairness and representation must be considered to ensure they are made responsibly.
This is particularly true in statistical network analysis; observed networks are often incomplete or systematically biased due to self-reporting or partial observation.

Networks are a powerful tool for analyzing complex data, where nodes represent entities and edges capture their interactions. Network analysis plays a crucial role in modern data science, enabling insights into both global network structures and individual node properties. One fundamental problem in network analysis is node ranking, a topic which has been studied extensively. A well-known node ranking method is the PageRank algorithm \citep{bianchini2005inside}, originally developed at Google in the 1990s to score nodes according to their importance in information diffusion. These scores can be used to algorithmically rank nodes—originally to select top hits for search queries. 
More generally, network-based ranking algorithms take an input network of $n$ nodes and produce a ranked ordering of these nodes.  
Such rankings typically rely on node centrality metrics \citep{liao2017ranking}, such as degree, closeness, betweenness, or eigenvector centrality \citep{rodrigues2019network}, which capture different aspects of influence. 

In practical applications, algorithmic ranking can have major implications for opportunity and selection of nodes, for instance hiring of job applicants, or acceptance of students into university.
In network analyses, there can also be downstream effects on tasks like influence maximization: node rankings may be used as a simple proxy to select seeds for information spread, thus increasing their access to knowledge \citep{li2018influence,mesner2023fair}.

Network data often includes node attributes, such as gender or age in social networks, and cell types in biological networks. 
Systematic bias occurs when the distributions of node centrality measures depend on protected group membership, resulting in node rankings which are not a fair reflection of the true network structure \citep{zehlike2022fairness}.

To illustrate the importance of such rankings, consider contact networks, which capture patterns of person-to-person interactions and play a central role in understanding social behavior as well as modeling the spread of infectious diseases, information, and behaviors. In such networks, individuals
are represented as nodes and reported contacts form the edges. Identifying the most connected individuals with the highest degree can highlight key spreaders or control points, with direct implications for epidemic control and targeted interventions \citep{stehle2011high, fournet2014contact, mastrandrea2015contact, kucharski2018structure}.

\subsection{High School Contact Networks}

We illustrate this phenomenon using contact diary data reported by French high school students, alongside wearable sensor data recorded on the same day.
Wearable sensors are embedded in badges and tuned to detect face-to-face contacts when two individual are within approximately 1 to 1.5 metres for 20 or more seconds \citep{mastrandrea2015contact}.
Figure~\ref{fig:intro_bias_overview}, panel (A) summarizes the overall interaction pattern, showing an asymmetry in reported cross-gender contacts through the estimated connection probability matrix.
Panel (B) compares the resulting degree-based rankings from the two networks. The diary-based network, based on students’ self-reported perceptions of their social contacts, exhibits a striking gender imbalance: although females comprise about 45\% of the population, nine of the top ten nodes are male (female proportion = 0.10). In contrast, the sensor-based network provides a closer approximation to the true contact structure, with a more balanced gender composition among the top ten (female proportion = 0.30).

\begin{figure}[tbp]
\centering
\begin{subfigure}[t]{0.35\linewidth}
    \centering
    \includegraphics[width=\linewidth]{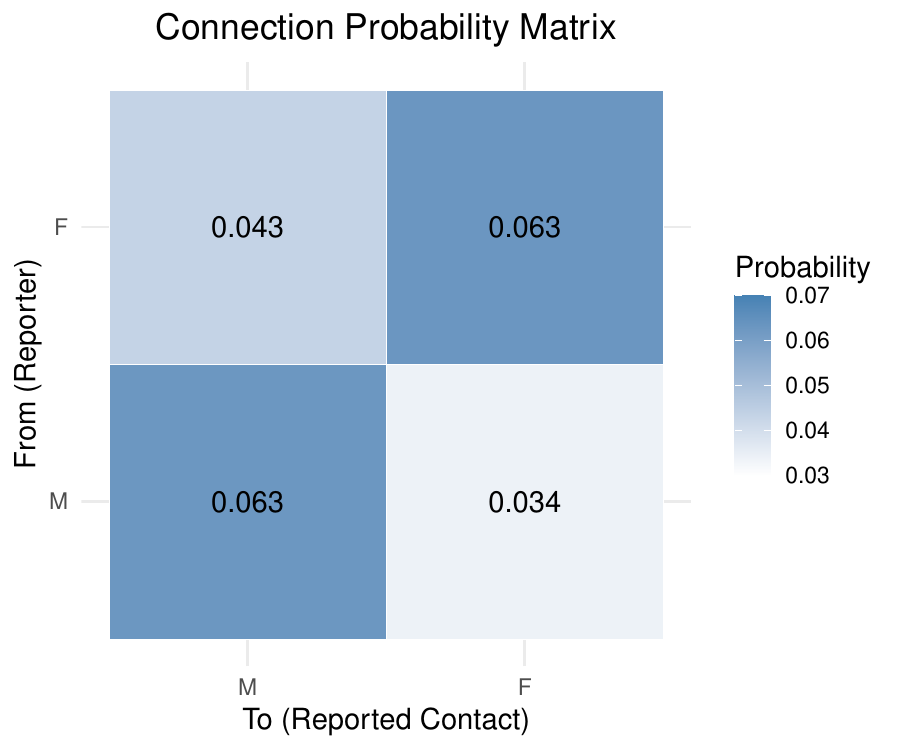}
    \caption{Connection probability matrix, with nodes subsetted by gender.}
    \label{fig:intro_conn_matrix}
\end{subfigure}
\hfill
\begin{subfigure}[t]{0.63\linewidth}
    \centering
    \includegraphics[width=\linewidth]{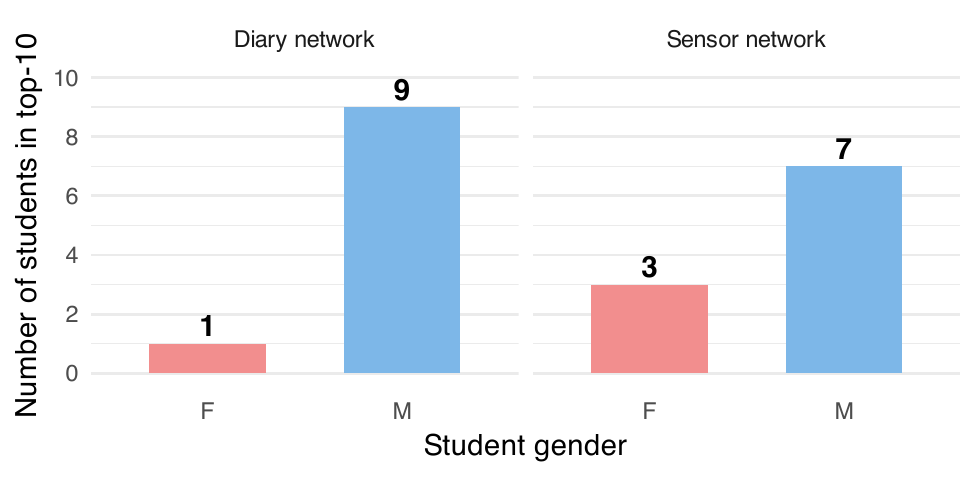}
    \caption{Gender composition of the top-10 students under naive in-degree ranking.}
    \label{fig:intro_top10_bar}
\end{subfigure}
\caption{Bias patterns in the high-school contact network data. (A) Estimated connection probabilities between genders in the diary-based network. (B) Gender composition among the top-10 ranked students in the diary-based (left) and sensor-based (right) networks.}
\label{fig:intro_bias_overview}
\end{figure}

This imbalance raises an important question: does it reflect genuine differences in connectivity, or might it be partly driven by systematic reporting errors such as recall bias? Diary-based networks are known to contain both error and bias \citep{smieszek2012collecting}, and there is evidence that such bias may depend on traits such as gender. For instance, studies suggest that females tend to recall social connections more accurately than males, who are more likely to underreport certain ties \citep{quintane2025gender}.  

Contact diary data nonetheless remains an important tool, particularly in settings where wearable sensors are infeasible, or as a complementary perspective on social network structure. 
However, if recall bias systematically varies across groups, it can distort centrality measures, propagate into rankings, and ultimately undermine both scientific conclusions and fairness in decisions based on these networks. 

\subsection{Fairness Framework and Contributions}

Our contact diary data illustrates a broader problem: systematic bias in network data can misrepresent centrality and distort rankings, with direct consequences for fairness. To study this problem more generally, we now introduce a framework for bias and fairness in network-based rankings.

Systematic bias and mathematical fairness have been defined in various ways across statistics and machine learning. We adopt the framework of \citet{friedler2021possibility} to formalize these concepts for network rankings. Each node $i$ in the network has an {\em oracle score $S_i^*$}, representing true importance. The set of oracle scores resides in a {\em construct space}, where a fair ranking orders nodes strictly by $S_i^*$. In practice, $S_i^*$ is unobserved; instead, we rely on an {\em observed score} $S_i$, which may be biased or noisy, existing in the {\em observation space}.

{\em Systematic bias} arises when the mapping from oracle scores to observed scores differs across groups, a phenomenon known as group skew. For a  protected attribute $Z_i$ (e.g., gender), if the conditional distributions of $S_i$ given $S^*_i$ vary across groups, i.e.,
\begin{equation}\label{model:systematic_bias}
       p(S_i \mid S^*_i, Z_i = 0) \neq p(S_i \mid S^*_i, Z_i = 1),
\end{equation}
then group-level disparities can emerge. In our contact diary data, for example, if fewer contacts are recorded with female students ($Z_i = 1$) than male students ($Z_i = 0$), with similar $S^*_i$, then $S_i$ underrepresents their true importance. Rankings based only on $S_i$ can therefore misrepresent the centrality of female students in the social network.

\citet{friedler2021possibility} highlight the issue of \textit{non-identifiability} in the mappings between construct and observation spaces. This non-identifiability leads to different assumptions being made in practice, contributing to various interpretations and disagreements about what constitutes a fair ranking. They identify two primary worldviews or assumptions: “what you see is what you get” (WYSIWYG) and “we’re all equal” (WAE).

In fair algorithmic ranking, WYSIWYG assumes no group skew: observed scores accurately reflect true node importance. Under this assumption, directly ranking nodes based on their observed scores is considered fair because the observed scores are viewed as unbiased representations of the oracle scores. In contrast, WAE assumes that fairness in ranking is related to {\em proportional representation} among protected groups. Specifically, it requires that the proportion of nodes from each group in the top-$K$ ranked positions is approximately equal to their proportion in the overall population of nodes. If the ranking fails to meet this condition, it is considered unfair, and correction is required.

Under WYSIWYG, a non-proportional ranking may still be fair if it reflects true importance. Under WAE, any systematic deviation from proportional representation constitutes unfairness. An intermediate position, which does not fully adhere WYSIWYG or WAE, makes the question of fairness less clear-cut. In such cases, proportionality alone does not imply fairness, and estimation of the imbalance between oracle and observed scores can be formally unidentifiable.
\ifjrss
\else
Related concepts of non-identifiability have been explored in the context of noisily observed networks \citep{chang2022estimation}, and we build on this work to show that, under certain bias models, unfairness in degree-based rankings can be quantified and corrected—provided we have access to additional information, such as independent network replicates or knowledge of the bias process.
\fi

The question of fairness in ranking has received growing attention across statistical and algorithmic settings. We briefly review several recent studies addressing fairness and proportional representation in network-based rankings. \citet{karimi2018homophily} focused on the effects of homophily on proportional representation, under a group-structured preferential attachment model. 
They note that homophily or heterophily in the preferential attachment mechanism can lead to over- or underrepresentation of the majority group in the top degree-ranked nodes. Rather than addressing fairness directly, their empirical investigation focuses on the proportional representation of protected groups within node rankings.
\citet{zhang2021chasm} empirically studied a ``chasm effect'' in minority and majority representation in degree-based node rankings; they observed underrepresentation of minority group nodes among both the top and bottom ranked sets.
\citet{neuhauser2021simulating} empirically studied how group-skewed edge errors affect networks with two unbalanced groups. They found that missing within-group edges  disadvantage majority group nodes, while missing between-group edges disadvantage minority group nodes. While not addressing fairness directly, their analysis lies between the WYSIWYG and WAE worldviews: they simulate networks lacking proportional rankings (violating WAE), then add group-dependent edge observation errors (violating WYSIWYG).
Due to this intrinsic non-identifiability of their model, the analysis in this paper is limited to empirical investigation.
\ifjrss
\else
As the authors note, ``unless specific assumptions are made about the generative process creating the observed network, \ldots it is generally impossible to extract (and correct) errors in network data based on a single network observation.''
\fi
This is consistent with formal statistical treatment of noisily observed networks in \citep{chang2022estimation}. Relatedly, \citet{antunes2024representation} analyzed how sampling schemes and network homophily shape the representation and ranking of minority nodes in attributed networks.

A separate line of research in fairness-aware graph learning adopts a different notion of fairness. Here, fairness assumes nodes from different groups have similar within- and between-group connectivity. To estimate the graph, a fairness penalty term is incorporated into the graph estimation procedure. Some examples include \citet{navarro2024mitigating, zhou2024fairness, navarro2024fair} under various network models, including the Gaussian graphical model and Gaussian covariance graph model.

\ifjrss
In this manuscript, in addition to our analysis of high school contact network data, we make several broadly applicable theoretical and methodological contributions. 
We provide a formal statistical framework for assessing systematic bias in network rankings, including precise definitions of the statistics used to quantify minority group representation. 
For degree-based rankings under some canonical network models, we derive the asymptotic behavior of minority representation statistics.
Finally, extending the empirical insights of \citet{neuhauser2021simulating}, we develop methods to formally detect and correct systematic bias in network models with observation error,  producing node rankings which are fair on average.
\else
In addition to a detailed, fairness-aware analysis of this high school contact network data, we make the following broadly applicable theoretical and methodological contributions:
\begin{enumerate}
\item [1.] We provide a formal statistical framework for assessing systematic bias in network rankings, including precise definitions of the statistics used to quantify minority group representation. 
\item [2.] For degree-based rankings, we derive asymptotic behavior of minority representation under a broad class of labeled graphon network models. 
\item [3.] 
We develop methods to formally detect and correct systematic bias in network models observed with error, with access to independent network replicates,  producing node rankings which are fair on average.
\end{enumerate}
\fi
To our knowledge, this is the first fairness-aware theoretical investigation of such node ranking statistics, and the first methodology for detecting and correcting systematic bias and mathematical unfairness in degree-based rankings.
Our principled approach to fairness and systematic bias leads to a novel and detailed analysis of the contact diary data; we apply our bias detection and re-ranking methods to verifiably improve group fairness in degree-based rankings, when validated against the corresponding wearable sensor data.
\ifjrss
\else

The rest of the article is organized as follows. In Section \ref{sec:model}, we establish the mathematical definitions and introduce the network models. Section \ref{sec:theory} presents asymptotic results on the representation of minority node groups under different network models. Section \ref{sec:correction} introduces various methods which leverage our asymptotic results to test for systematic bias and apply ranking corrections based on estimated model parameters. Section~\ref{sec:realdata} describes our main analysis of the high school contact network data, and Section~\ref{sec:conclusion} concludes with a discussion of key insights and directions for future research.
\fi

\section{Minority representation profiles and models} \label{sec:model}


Throughout the paper, we primarily consider a network represented as $G = (V, E)$, where $V$ is the set of $n$ nodes and $E$ is the set of edges, represented by a binary adjacency matrix $A^{(n)} \in \{0,1\}^{n \times n}$ and group labels $\bm{c}^{(n)} \in \{1,2\}^n$, jointly drawn from a labelled network model 
\ifjrss
denoted by $\mathcal{P}_n$.
\else
$$
(\bm{c}^{(n)}, A^{(n)}) \sim \mathcal{P}_n.
$$
\fi
Each node is independently assigned to the minority group (group 1) with probability $\kappa < 1/2$ and to the majority group (group 2) with probability $1-\kappa$. For simplicity, we omit the superscript $n$ in the notation for $\bm{c}^{(n)}$, $A^{(n)}$ throughout, and all limits should be interpreted as $n \to \infty$. 

Nodes are ranked using a general {\em ranking criterion}, which maps from the adjacency matrix to a vector of node scores. We focus on \textit{degree-based ranking} for its tractability and widespread use. Let $ d_i = \sum_{j=1}^n A_{ij}$ denote the degree of node $i$ in an undirected network, and more specifically the in-degree in a directed network. The top-$K$ set, denoted $D_K(A)$, consists of the $K$ nodes with the highest degrees, with ties broken at random.

\subsection{Minority representation profiles}
\label{subsec:minority_represent}

Our primary interest is in the sequence of statistics known as the {\em minority representation profile}, which tracks the proportion of minority (group 1) nodes in the top-$K$ according to the degree ranking. 
Specifically, define the random variable
$$
    R_K(\bc,A) = \frac{1}{K} \left\lvert \left\{ i \in D_K(A) : c_i = 1 \right\} \right\rvert \in [0,1], 
$$
for $K \in \{1,\ldots,n\}$.
To describe minority representation of a given network model, we define an expected analog of $R_K(\bm{c},A)$, as well as a limiting profile as $K$ scales with $n$.
\begin{definition}
    For a labelled network model $\mathcal{P}_n$, define the {\em expected minority representation profile} as 
    $
        \rho_K = \mathbb{E}_{\mathcal{P}_n} \left\{ R_{K}(\bm{c},A) \right\},
    $
    for $K \in \{1,\ldots,n\}$.
\end{definition}
\begin{definition}
    For a sequence of labelled network models $\{\mathcal{P}_n\}_{n \geq 1}$, define the {\em asymptotic minority representation profile} as 
    $
        \rho^*(z) = \operatorname{plim}_{n \rightarrow \infty} \left\{ R_{\lfloor nz \rfloor}(\bm{c},A) \right\}, 
    $
    for $z \in (0,1]$ if the limit exists, where $\operatorname{plim}$ denotes the limit in probability of a sequence of random variables.
\end{definition}
\ifjrss
\else
To facilitate analysis, we propose several classes of parametric network models under which the asymptotic minority representation profile can be computed analytically. 
For the sequence of labelled network models $\{\mathcal{P}_n\}_{n \geq 1}$, $\rho^*(z)$ describes the limiting behavior of minority representation in degree-based rankings.
\fi
As we have noted in Section~\ref{sec:intro}, empirical assessments often rely on some form of proportionality. In both finite and limiting regimes, we formalize this notion as follows. 
\begin{definition}\label{def:proportionality}
A labelled network model $\mathcal{P}_n$ has {\em proportional representation} with respect to degree ranking if
$
    \rho_K = \kappa,
$
for all $K \in \{1,\ldots,n\}$.
\end{definition}
\begin{definition}\label{def:asy_proportionality}
A sequence of labelled network models  $\mathcal{P}_n$ has {\em asymptotic proportional representation} with respect to degree ranking if
$
    \rho^*(z) = \kappa,
$
for all $z\in (0,1]$. 
\end{definition}

\subsection{Canonical network models}\label{subsec:network_model}

In this section, we introduce two canonical generative models for labelled, undirected  networks: the general labelled graphon model and, as a special case, the stochastic block model. We use these frameworks to analyze minority representation in degree-based rankings.

\subsubsection{Labelled graphon} \label{subsubsec:graphon_model}

A general graphon model on $n$ nodes is parameterized by a \textit{latent position vector} $ \bu = (u_1, u_2, \dots, u_n) $, where $ u_i \in [0,1] $ represents the latent feature of node $ i $, and a \textit{graphon function} $ \omega: [0,1] \times [0,1] \to [0,1] $, where $ \omega(u_i, u_j) $ gives the probability of an edge between nodes $ i $ and $ j $. The adjacency matrix $ A $, representing the connections in the network, is generated as follows:
\begin{equation} \label{eqn:graphon}
\prob(A_{ij} = 1 \mid u_i, u_j) = \omega(u_i, u_j), \quad \forall i < j.
\end{equation}
To assess minority group representation, we introduce a labelled variant of the graphon. 
Assume nodes are independently assigned to the minority group with probability $\kappa < 1/2$, and the majority group otherwise. Conditional on group labels, we sample latent positions:
\begin{equation*}
    U_i \vert c_i=1 \sim \operatorname{Uniform}(0,\kappa), \quad 
    U_i \vert c_i=2 \sim \operatorname{Uniform}(\kappa,1), 
\end{equation*}
independently. Then generate the adjacency matrix $A$ as in \eqref{eqn:graphon}.
This construction ensures the marginal distribution of $U_i \sim \operatorname{Uniform}[0,1]$, so the resulting network remains a valid graphon model.
\ifjrss
This is similar to the so-called {\em stochastic block smooth graphon model} of \cite{sischka2024stochastic} (treating community labels as observed), and the graphon model with node features proposed in \citet{su2020network}.
\else
This is similar to the so-called {\em stochastic block smooth graphon model} of \cite{sischka2024stochastic}, although here we treat the community labels as observed.

An equivalent formulation arises when labels are defined in terms of the latent positions: 
$c_i = \mathbb{I}(U_i \in M)$
for a measurable subset $M \subset [0,1]$ with Lebesgue measure $\kappa$, and $\mathbb{I}(\cdot)$ denotes the indicator function. 
This is similar to the graphon model with node features proposed in \citet{su2020network}. Due to the invariance property of graphon functions to measure-preserving transformations, this leads to the same distribution over adjacency matrices.
\fi 

In the graphon model, node degree is determined by the graphon function $\omega$ and the latent features of other nodes. 
\ifjrss
Assortativity in $\omega$ can lead to systematic underrepresentation of minority nodes; however this underrepresentation can arise in a broader range of network structures, and is not restricted to block-exchangeable patterns.
\else
A node with latent feature $u_0$ will tend to have higher degree if $\omega(u_0,v)$ is greater for $v$ varying between 0 and 1, reflecting greater influence. When $\omega$ is denser within the majority group (i.e.~for $u_0 \in [\kappa,1]$) than within the minority group ($u_0 \in [0,\kappa]$), or if it is denser for within-group edges than between-group connections, minority nodes tend to be underrepresented in degree-based rankings. 
\fi

\subsubsection{Stochastic block model} \label{subsubsec:sbm}

The stochastic block model (SBM) is a widely used generative network model that can be viewed as a special case of the labelled graphon framework. Originally proposed to capture community structure in social networks \citep{holland1983stochastic, abbe2018community}, the SBM partitions nodes into distinct communities, with edges formed independently according to the community memberships of their endpoints.

We study a 2-group SBM, where each node is independently assigned to the minority group with probability $\kappa < 1/2$, and to the majority group otherwise. Given the group label vector 
\ifjrss
$\bm{c}$,
\else
$\bc = (c_1, c_2, \dots, c_n)\in \{1,2\}^n$, 
\fi
edges are formed independently with probabilities specified by a connectivity matrix $B \in [0,1]^{2 \times 2}$, according to
$
\prob(A_{ij} = 1 \mid c_i, c_j) = B_{c_ic_j}
$ for all $i < j$.
In an unbalanced SBM with $\kappa < 1/2$, minority node degrees are influenced more by between-group connections due to group size disparity. When the model is assortative (within-group connections are more likely), this can lead to systematic underrepresentation of minority nodes in top-$K$ degree rankings. Numerically, we will have $\rho_K < \kappa$ for small values of $K$.

\subsection{Network models with systematic bias} \label{subsec:sys_bias}

In network analyses, particularly those focused on minority representation in degree-based rankings, systematic inaccuracies can significantly distort conclusions \citep{neuhauser2021simulating}. These distortions often arise from biases embedded in the attributed structure of social networks, unfairly impacting minority node rankings. 
To model this bias, we suppose the true underlying network $A$ is unobserved, and instead we observe a biased version $Y$, where the bias depends on the group labels of the nodes. Formally, we assume 
\ifjrss
$\mathbb{P}(Y_{ij} = 0 \mid A_{ij} = 1) = \beta_{c_ic_j}$,
\else
\begin{equation}\label{eqn:systematic_bias}
\mathbb{P}(Y_{ij} = 0 \mid A_{ij} = 1) = \beta_{c_ic_j}, 
\end{equation}
\fi
\ifjrss
independently over all $i < j$, for type II error rates $\beta_{11}$, $\beta_{22}$, and $\beta_{21}=\beta_{12}$ (by symmetry).
\else
where the matrix of type II error rates is defined as 
\begin{equation}\label{eqn:SBM_beta}
    \bbeta = \begin{pmatrix}
        \beta_{11} & \beta_{12} \\ 
        \beta_{21} & \beta_{22}
    \end{pmatrix} \in [0,1]^{2 \times 2}.
\end{equation}
\fi
Following the terminology in \citet{friedler2021possibility}, we refer to the generative model for the underlying network $A$ as the {\em construct model}, and the generative model for $Y$ as the {\em observation model}.

These observation errors can systematically reduce the visibility of minority nodes in degree-based rankings. If the construct network $A$ has proportionally representation (for instance if it follows an Erdős–Rényi model), but the observation errors favors within-group edges (e.g., $\beta_{12} < \beta_{11} = \beta_{22}$),  minority nodes will be underrepresented in top-$K$.
If the majority group nodes are more likely to form within-group connections in the construct network $A$, additional observation bias with $\beta_{12} \leq \beta_{11} = \beta_{22} $ will further exacerbate the underrepresentation of minority nodes in top-$K$ rankings.


The addition of this observation error mechanism does not increase the class of models under study; it is easy to see that if the construct model comes from either the class of SBMs defined in Section~\ref{subsubsec:sbm}, or the class of the labelled graphons defined in Section~\ref{subsubsec:graphon_model}, the observation model will remain in the same class. For this reason, it is sufficient to study the probability limits of $\rho_K$ under the SBM and labelled graphon models in Section~\ref{sec:theory}, without explicitly incorporating the observation bias. 
We return to the systematic bias model in Section~\ref{sec:correction}, where we use it to develop a correction procedure when replicated network observations are available.

\section{Quantifying minority representation}\label{sec:theory}

In this section, we analyze the asymptotic behavior of the minority representation profile under the canonical network models introduced in Section \ref{subsec:network_model}.

\subsection{Labelled graphon}\label{sec:model_graphon}

In this section, we derive asymptotic limits for minority representation in degree-based rankings under labelled graphon models. For each $n$, define a graphon function $\omega_n: [0,1] \times [0,1] \rightarrow [0,1]$. 
To ensure nontrivial asymptotic behavior in degree-based rankings, we will make the following two assumptions on the sequence $\{\omega_n\}_{n \geq 1}$.

\begin{assumption} \label{assump:graphon_concentration}
There exists a constant $\epsilon > 0$, and sequence $\{ p_n \}_{n \geq 1} \in (\epsilon, 1 - \epsilon)$ such that
\begin{equation}\label{eqn:graphon_concentration}
    \sup_{u,v \in [0,1]} ~ \lvert \omega_n(u,v) - p_n \rvert = O\left(\frac{1}{\sqrt{n}}\right). 
\end{equation}
\end{assumption}
Without loss of generality, we may take $p_n$ to be the population edge density of 
\ifjrss
$\omega_n$.
\else
each graphon,
$$
   p_n = \int_0^1 \int_0^1 \omega_n(u,v)dudv.
$$
\fi
\begin{assumption} \label{assump:graphon_mean}
For $p_n$ in \eqref{eqn:graphon_concentration}, assume that for all $u \in [0,1]$ and some limiting function $\mu(u)$,
\begin{equation}\label{eqn:graphon_mean}
    \sqrt{\frac{n}{p_n(1-p_n)}}\int_0^1 \left\{ \omega_n(u,v) - p_n \right\}dv \rightarrow \mu(u). 
\end{equation}
\end{assumption}
Here, $\mu(u)$ describes the normalized expected degree of a node with latent feature $u$.
For more discussion of the choice of scaling in Assumption~\ref{assump:graphon_concentration}, see the discussion around Proposition~\ref{prop:SBM}.
\ifjrss
\else
The following result quantifies the (asymptotic) underrepresentation of the minority group in degree-based rankings based on a sequence of labelled graphon models.
\fi

\begin{theorem}\label{thm:graphon_limits}
    Let $A$ be a network generated from a labelled graphon $w_n$  satisfying Assumptions~\ref{assump:graphon_concentration}-\ref{assump:graphon_mean} and let $K/n \to z$ as $n \to \infty$. Then, 
    $$
    R_K(\bm{c},A)  \convp \frac{\kappa}{z}\left\{ 1 - F_1(c^*) \right\} ,
    $$
    where $c^*$ is the $(1-z)$-th quantile of the following mixture distribution: 
    $$
    \kappa F_1(\cdot) + (1-\kappa)F_2(\cdot) = \kappa \left\{ \mathcal{N}\left(0, 1\right) + \mu(U^{(1)})\right\} + (1-\kappa) \left\{ \mathcal{N}\left(0, 1\right) + \mu(U^{(2)})\right\}, 
    $$
    with $U^{(1)}\sim \operatorname{Uniform}(0,\kappa)$, $U^{(2)}\sim \operatorname{Uniform}(\kappa,1)$, and $\mu(\cdot)$ is defined in Assumption~\ref{assump:graphon_mean}.
\end{theorem}

\subsection{Stochastic block model}\label{sec:model_SBM}

To better understand the general result above, we illustrate with a simple two-group SBM consisting of a minority and a majority group. 
Consider the SBM with probability matrix
\begin{equation}\label{eqn:SBM_B}
    B = \begin{pmatrix}
        p_1 & q \\ 
        q & p_2
    \end{pmatrix},
\end{equation}
where $p_g$ ($g=1,2$) denotes the within-group connection probability for group $g$, and $q$ denotes the between-group connection probability. 
\ifjrss
As above, $\kappa < 1/2$ denotes the probability that a node is assigned to the minority group.
\else
Nodes are assigned independently to the two groups: group 1 (minority) with probability $\kappa < 1/2$, and group 2 (majority) with probability $1-\kappa$. 
\fi

We denote this model as SBM($n, B, \kappa$). In \eqref{eqn:SBM_B}, we allow $p_1,p_2$ and $q$ to depend on $n$. The following proposition illustrates a phase transition in the detectability of group differences. Let $\delta_g := p_g - q$ denote the signal strength for group $g =1,2$. 
When $\delta_g$ is too small, the group differences are not detectable in a degree-based ranking, and the limiting ranking is proportionally representative. Conversely, when $\delta_g$ is too large, the ranking becomes trivial: all majority group nodes dominate the top of the ranking in the limit. This idea is formalized in the following result. 
\begin{proposition}\label{prop:SBM}
Consider the SBM defined in \eqref{eqn:SBM_B} with an assortative structure, i.e., $\delta_g > 0$ for $g=1,2$. Then:
\begin{enumerate}
    \item[(a)] If $\delta_g \ll \frac{1}{\sqrt{n}}$ for $g=1,2$ and $K/n \to z$, then 
    $
    R_K(\bm{c}, A) \convp \kappa.
    $
    \item[(b)] If $(1-\kappa)\delta_2 - \kappa\delta_1 \gtrsim \sqrt{\frac{\log n}{n}}$, $\kappa < \frac{1}{2}$, and $K/n \to z$ with $z \leq 1 - \kappa$, then
    $
    R_K(\bm{c}, A) \convp 0.
    $

    \item[(c)] If $\kappa\delta_1 - (1-\kappa)\delta_2 \gtrsim \sqrt{\frac{\log n}{n}}$, $\kappa < \frac{1}{2}$, and $K/n \to z$ with $z \leq \kappa$, then
    $
    R_K(\bm{c}, A) \convp 1.
    $
\end{enumerate}
\end{proposition}
\ifjrss
\else 
\begin{remark}
    When $p_g - q \ll \frac{1}{\sqrt{n}} $ with balanced communities, \cite{banerjee2018contiguity} showed that the SBM is statistically contiguous to the Erdős–Rényi (ER) model. 
\end{remark}
\begin{remark}
    When $p_g - q \geq \frac{c}{\sqrt{n}} $ for a sufficiently large constant $c$, naive clustering via degree achieves perfect label recovery. This significant disparity between the groups results in an extreme ``unfairness'' that we would expected to be rarely observed in real-world networks.
\end{remark}
\fi
Therefore, we can focus on the parameter space where $\delta_g \sim \frac{1}{\sqrt{n}}$ for $g=1,2$, in order to obtain nontrivial asymptotic behavior. To this end, we define $ p_g = q + \frac{\mu_g}{\sqrt{n}} $ for $ \mu_g \in \mathbb{R}$, resulting in the following connection probability matrix:
\begin{equation}\label{eqn:SBM_B_def}
    B = \begin{pmatrix}
        q + \mu_1/\sqrt{n}  & q \\ 
        q & q + \mu_2/\sqrt{n}
    \end{pmatrix} 
\end{equation}
Note that this connection probability matrix is a piecewise constant graphon, which satisfies Assumptions~\ref{assump:graphon_concentration} and \ref{assump:graphon_mean} of Section~\ref{sec:model_graphon}.
We can characterize the behavior of limiting representation profile with the following result.
\begin{theorem}\label{thm:SBM_limits}
Consider a network $A$ generated from an SBM$(n, B, \kappa)$, where $B$ is defined in \eqref{eqn:SBM_B_def}.
Suppose $K/n \to z$ as $n \to \infty$. 
\ifjrss
Then we have 
$$
 R_K(\bm{c},A)  \convp \frac{\kappa}{z}\left\{ 1 - F_1(c^*) \right\} =: \rho^*(z;\kappa, q, \mu_1, \mu_2),
$$
\else
Then, the minority proportion among the top-$K$ nodes converges as follows:
$$
 R_K(\bm{c},A)  \convp \frac{\kappa}{z}\left[ 1 - F_1(c^*) \right] =: \rho^*(z;\kappa, q, \mu_1, \mu_2),
$$
\fi
where $c^*$ is the $(1-z)$-th quantile of the following mixture of Gaussian distributions:
$$
\kappa F_1(\cdot) + (1-\kappa)F_2(\cdot) := \kappa \mathcal{N}\left( \frac{\kappa \mu_1}{\sqrt{q(1-q)}}, 1 \right) + (1-\kappa) \mathcal{N}\left( \frac{(1-\kappa) \mu_2 }{\sqrt{q(1-q)}}, 1 \right).
$$
\end{theorem}

\begin{remark}
The conclusion remains valid if one or both signals vanish, i.e., $\mu_1 = o(1)$ and/or $\mu_2 = o(1)$. The limiting mixture distribution simplifies accordingly; for example, if $\mu_1 = o(1)$, the first component becomes $\mathcal{N}(0,1)$, and similarly for $\mu_2$. 
\end{remark}


In Section 2 of the Supplementary Material, we show that this asymptotic approximation can perform well for $n$ as small as $200$, and that it can be used to produce a consistent, parametric estimate of the minority representation profile which can be more accurate than the empirical class memberships, especially for small values of $K$.
In Figure~\ref{fig:sbmlimits}, we inspect the shape of some of the asymptotic minority representation profiles as we vary different SBM parameters. 
In panel (A), we show the effect of changing $\kappa$, the minority group proportion, for fixed $q=0.15$ and $\mu_1=\mu_2=1$.
In panel (B), we show the effect of changing $\mu_1$ and $\mu_2$ together, for fixed $\kappa=0.4$ and $q=0.15$, where we denote $\mu = \mu_1 = \mu_2$.
Finally, in panel (C), we show the effect of changing $q$, the between-group edge density, for fixed $\kappa=0.4$ and $\mu_1=\mu_2=2$.

\begin{figure}[!t]
        \centering
        \begin{subfigure}[b]{0.48\linewidth}
            \centering
            \includegraphics[width=\linewidth]{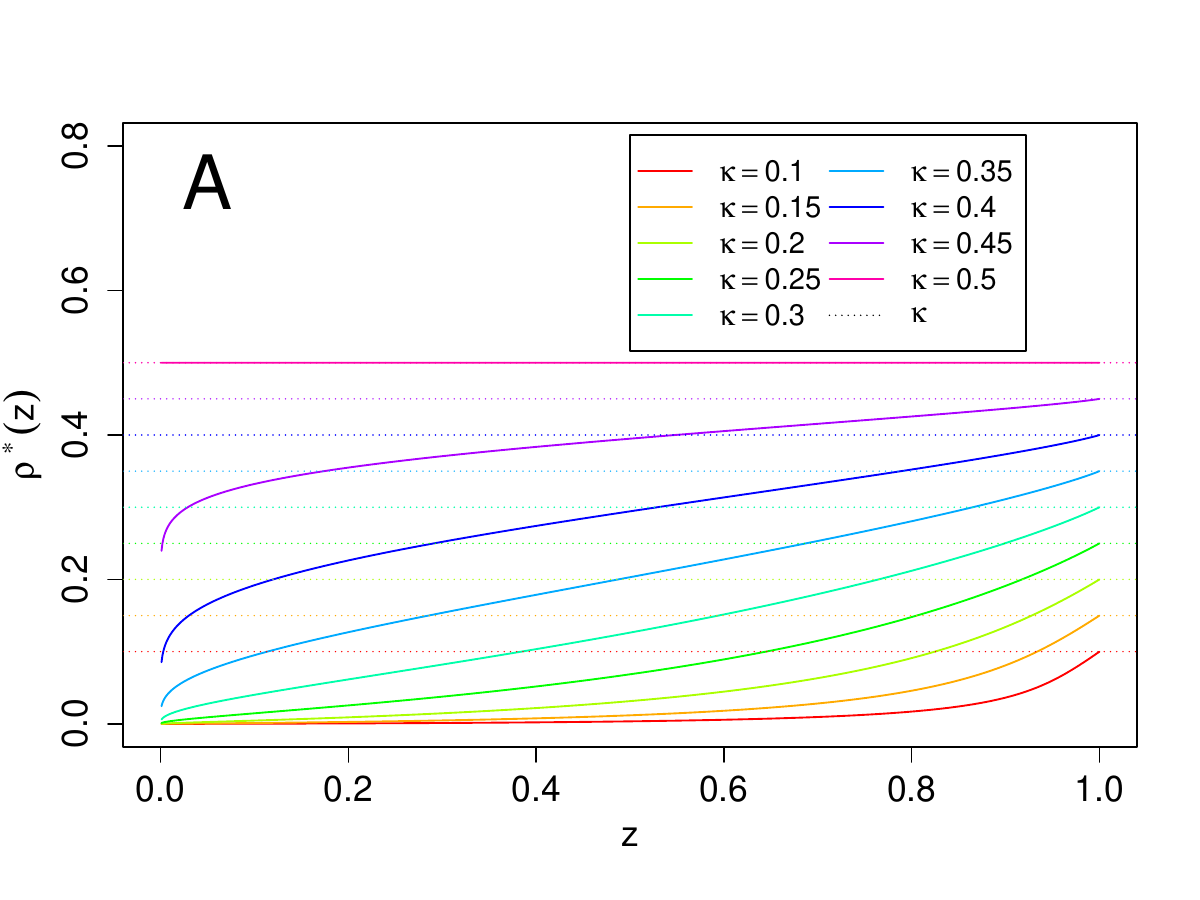}
            \caption{Varying $\kappa$, $\mu_1=\mu_2=1, q=0.15$}
        \end{subfigure}
        \hfill
        \begin{subfigure}[b]{0.48\linewidth}  
            \centering 
            \includegraphics[width=\linewidth]{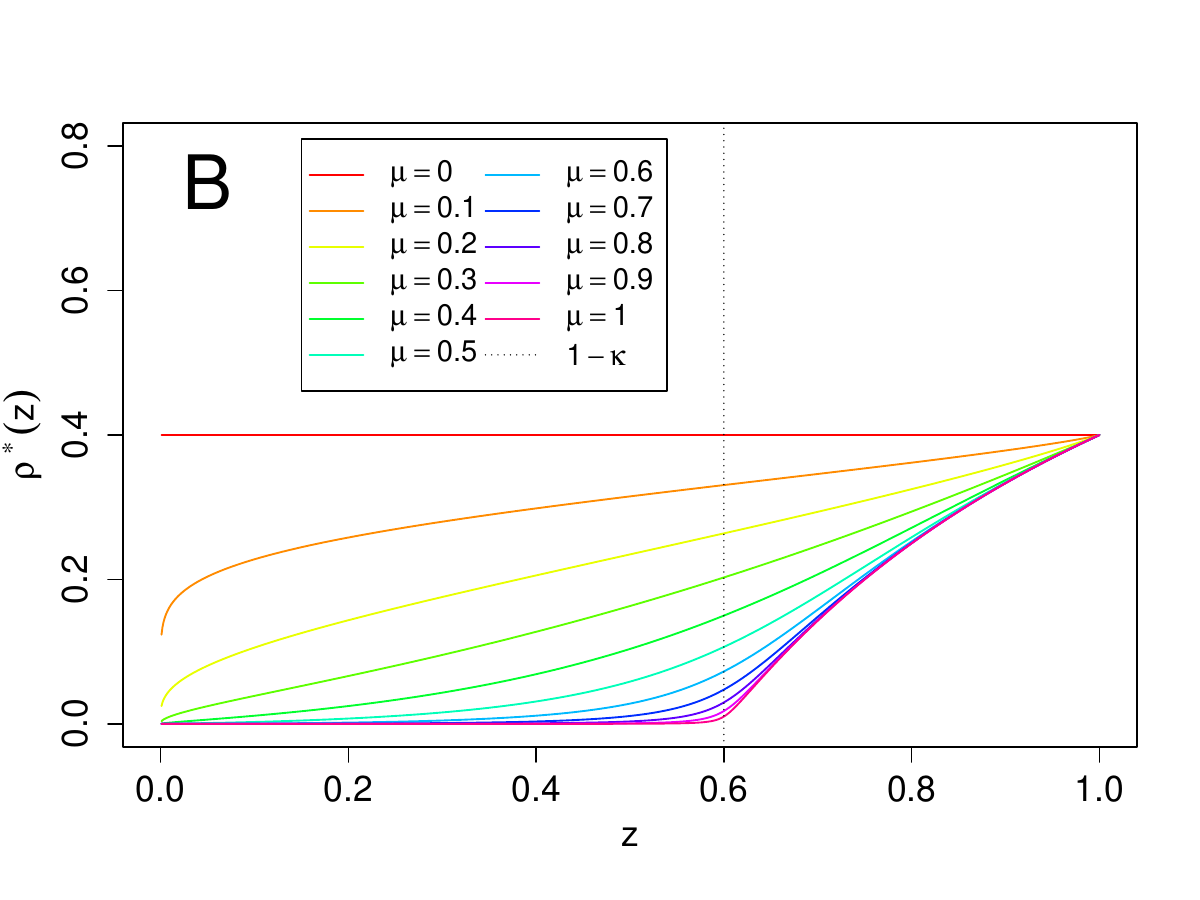}
            \caption{Varying $\mu$, $\kappa=0.4, \mu_1=\mu_2=\mu, q=0.15$}
        \end{subfigure}
        \begin{subfigure}[b]{0.5\linewidth}  
            \centering 
            \includegraphics[width=\linewidth]{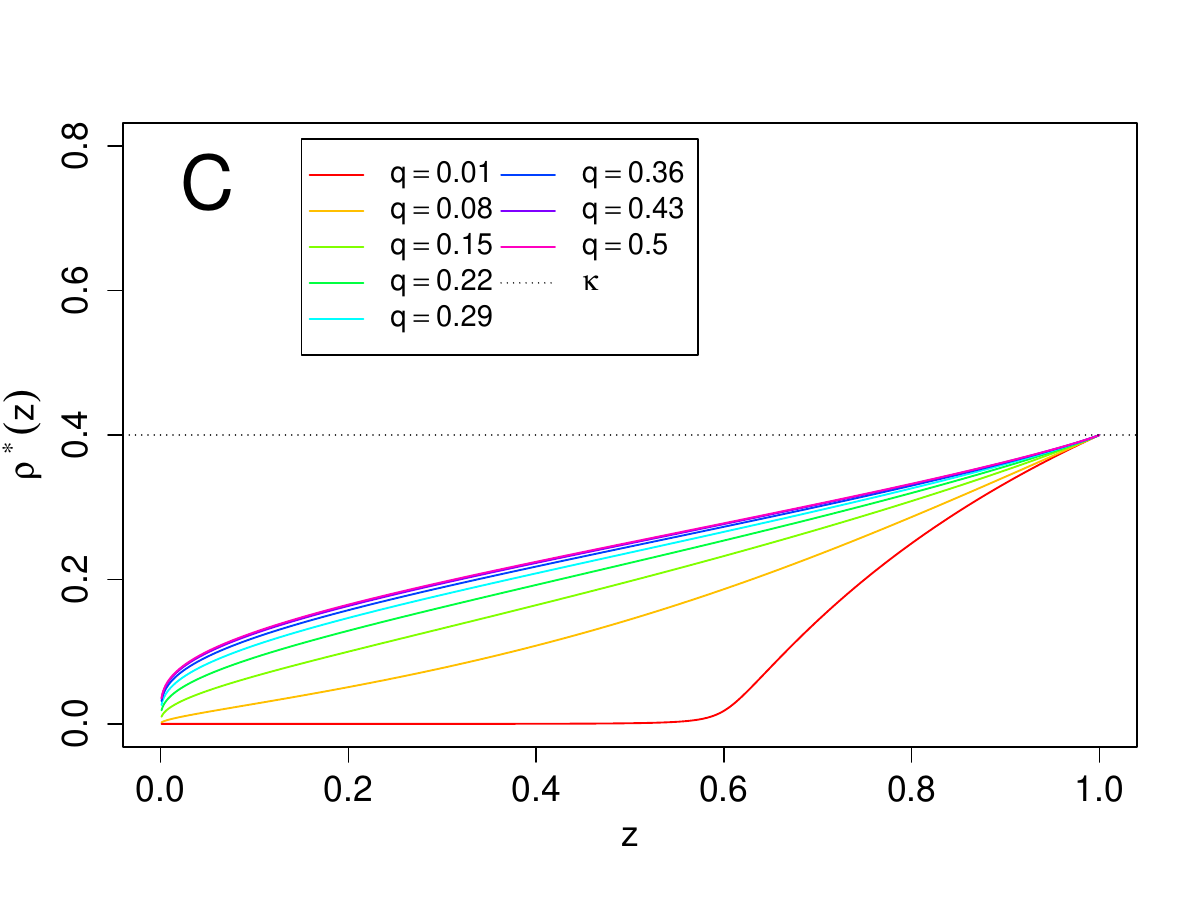}
            \caption{Varying $q$, $\kappa=0.4, \mu_1=\mu_2=2$}
        \end{subfigure}
        \caption{$\rho^*(z)$ for various SBM parameterizations.}
        \label{fig:sbmlimits}
\end{figure}

Panel (A) shows that as the groups become more unbalanced, so does the degree ranking, even with fixed SBM connection probabilities.
Panel (B) shows, as we expect, that as the SBM connection probabilities become more unbalanced in favor of in-group connections, the degree ranking becomes more unbalanced.
Finally, panel (C) shows a subtle effect of edge density on minority representation: for fixed minority proportion and signal strength, as the network becomes more sparse, the degree ranking becomes more unbalanced.
Note that $\rho^*(z)$ depends on $q$ only through $q(1-q)$, so the imbalance in the rankings will similarly increase for $q > 1/2$.

\section{Detecting and correcting for systematic bias}\label{sec:correction}

In this section, we develop novel methodology to detect and correct for systematic bias in rankings, under a 2-group SBM with observation errors, using the model defined in Section~\ref{subsec:sys_bias}.
\ifjrss
\else
This framework posits a construct model that we would like to use for ranking, but differential rates of (type II) error depending on group membership, which can disproportionately affect the rankings of minority group nodes under the observation model.
\fi
In general, the full model is unidentifiable from a single observed network, however the two fairness worldviews -- WAE and WYSIWYG -- can be interpreted formally as identifying assumptions.

Under the identifying assumptions of WYSIWYG, we assume that there are no observation errors, and thus there is no need to correct the ranking of observed degrees. 
Conversely, under the identifying assumptions of WAE, and assuming the construct model is ER, 
\ifjrss
\else
the observation model follows an SBM, as analyzed in Section~\ref{sec:theory}, and
\fi
any difference in the observation model edge densities is attributable to differences in the observation error rates.
Under WAE, the construct model is proportionally representative, making estimation of the construct model  $\rho_K$ a trivial task.
Moreover, there is a natural corrected degree ranking which takes the top nodes by degree in each group, so that the group proportions in any prefix of the ranking approximately match the overall minority proportion $\kappa$ \citep{zehlike2022fairness}.

With neither worldview as an identifying assumption, our goal is not necessarily to make the ranking proportionally representative, but instead to correct it to match an adaptive target representation profile: a plug-in estimate of the asymptotic minority representation profile for the construct model.
In practice, we compute corrected rankings sequentially. 
Initializing from the empty set, the next ranked node is either the top available minority group or top available majority group node, where ties are broken uniformly at random, and the group identity is chosen to minimize the discrepancy between the group proportion in the next prefix of the ranking and the corresponding value in the adaptive representation profile (see Algorithm~\ref{alg:correction}).
Any remaining discrepancy between the corrected and target representation profiles is due to the discreteness of the ranking problem.

\begin{algorithm} 
  \SetAlgoLined
  \caption{$\rho_K$-corrected degree ranking \label{alg:correction}}
  \KwResult{Returns corrected top-$K$ sets $\{\tilde{D}_K(A) \}_{K=1}^n$.}
  \KwIn{Adjacency matrix $A$, group labels $\bm{c}$, target representation profile $\{\rho_K\}_{K=1}^n$}
  $\tilde{R}_0(A) = 0$ ; $\tilde{D}_0(A) = \emptyset$ \;
  \For{$K \gets 1$ \KwTo $n$}{
    $g^*_K \gets \operatorname{argmin}_{g \in \{0,1\}} \left\{ \left\lvert \frac{(K-1)\tilde{R}_{K-1}(A) +  g}{K} - \rho_K \right\rvert \right\}$ \;
    Sample $j^*_K \in \operatorname{argmax}_j \{ d_j : j \notin \tilde{D}_{K-1}(A), c_j = 2 - g^*_K\}$
        uniformly at random. \;
    $\tilde{R}_K(A) \gets (\tilde{R}_{K-1}(A) + g^*_K)/K$ ; $\tilde{D}_K(A) \gets \tilde{D}_{K-1}(A) \cup \{ j^*_K \}$ \;
  }
\end{algorithm}

In order to develop methodology without an identifying worldview, we extend the work of \citet{chang2022estimation}, using independent network replicates to consistently estimate the group-dependent observation error rates and construct model parameters.
By \citet{chang2022estimation}, Theorem 1, replication is indeed a requirement here, it is impossible to consistently estimate the observation error rates without it.
Using the systematically biased model described above with only type II errors, two replicates are sufficient to estimate the blockwise observation error rates.

The availability of two independent replicates is common in various applications. For example, contact networks collected over different days \citep{fournet2014contact} and gene expression networks derived from repeated experiments under similar conditions \citep{faith2007large}. In other applications, such as for contact diary data, it may be possible to observe a single directed, noisy observation $Y$ of a construct network $A$, characterized by group-dependent directed type II (recall) errors with
$
    \mathbb{P}(Y_{ij}=0 \vert A_{ij}=1) = \beta_{c_ic_j}
$
for any $i,j\in [n]$; we may have $\beta_{12} \neq \beta_{21}$.


Directed noisy networks are often obtained via egocentric sampling, where an individual (the ``ego'') is asked to report their connections to a set of other individuals (the ``alters'') within social or communication networks. 
This is the case in our high school contact diary data.
As discussed in Section~\ref{sec:intro}, egocentric networks (including contact diary data) are affected by recall or reporting biases, which can introduce systematic errors in the observed ties. 
Beyond contact diary data, studies have also shown that male psychology researchers tend to systematically underreport or overlook ties with their female colleagues \citep{psychology_gender_bias}, illustrating another example of asymmetric, directed recall error. Similarly, research has found that recall discriminability in egocentric networks can be influenced by factors such as age and gender \citep{graves2017effects}.
One directed noisy network provides similar information to the two replicate case, and the error rate parameters can be estimated analogously. 
In the following, we focus on the two replicate, undirected case, with the corresponding methodology for a single directed noisy network provided in Section 1.4 of the Supplementary Material.




\subsection{Estimation}\label{subsec:estimation}

Under the (undirected) 2-group SBM with observation errors, we first generate $A$ according to a 2-group SBM with connection probability matrix parameterized by $(\kappa,q,\mu_1,\mu_2)$, then observe two conditionally independent and identically distributed replicates $Y$, $Y^*$ with type II observation errors, such that
$
    \mathbb{P}(Y_{ij}=0 \vert A_{ij}=1) = \beta_{c_ic_j},
$
independently over the node pairs with $i < j$; note $\beta_{12} = \beta_{21}$ by symmetry.
As in Section~\ref{sec:theory}, we reparameterize the model so that the marginal connection probabilities are within $O(n^{-1/2})$ by setting
$$
    \beta_{11} = \beta - \frac{\gamma_1}{\sqrt{n}}, \quad \beta_{22} = \beta - \frac{\gamma_2}{\sqrt{n}}, \quad \beta_{12} = \beta
$$
for constants $(\beta,\gamma_1,\gamma_2)$.
Per the discussion in Section~\ref{sec:theory}, this is the asymptotically interesting regime for analyzing representation profiles under the 2-group SBM. 
Marginally, both $Y$ and $Y^*$ follow a 2-group SBM with connection probability matrix
\begin{equation} \label{obs_matrix}
    \begin{pmatrix}
        (1-\beta + \gamma_1/\sqrt{n})(q + \mu_1/\sqrt{n}) & (1-\beta)q \\ (1-\beta)q & (1-\beta + \gamma_2/\sqrt{n})(q + \mu_2/\sqrt{n}) \end{pmatrix}.  
\end{equation}
To estimate the unknown parameters of the construct and observation models, we generalize the method of moments approach proposed in \citet{chang2022estimation}, by localizing according to node group memberships.
Define
\begin{align*}
    N_{g} &= \binom{n_g}{2}, \tabby &N_{\mathrm{b}} &= n_1n_2 \\
    \hat{u}_{1,g} &= \frac{1}{N_g} \sum_{i < j ~:~ c_i = c_j = g} Y_{ij}, \tabby &\hat{u}_{1,\mathrm{b}} &= \frac{1}{N_{\mathrm{b}}} \sum_{i < j ~:~ c_i \neq c_j} Y_{ij}, \\
    \hat{u}_{2,g} &= \frac{1}{2N_g} \sum_{i < j ~:~ c_i = c_j =g} \lvert Y_{ij} - Y_{ij}^* \rvert, \tabby &\hat{u}_{2,\mathrm{b}} &= \frac{1}{2N_{\mathrm{b}}} \sum_{i < j ~:~ c_i \neq c_j} \lvert Y_{ij} - Y_{ij}^* \rvert
\end{align*}
for $g=1,2$.
We estimate type II error rates by
\ifjrss
$\hat{\beta}_{\mathrm{x}} = \hat{u}_{2,\mathrm{x}} / \hat{u}_{1,\mathrm{x}}$, 
\else
\begin{equation*}
    \hat{\beta}_{\mathrm{x}} = \hat{u}_{2,\mathrm{x}} / \hat{u}_{1,\mathrm{x}},  
\end{equation*}
\fi
a consistent estimator for $\beta_{\mathrm{x}}$, replacing $\mathrm{x}$ by either ``$1$'', ``$2$'', or ``$\mathrm{b}$''.
In particular, we have
\begin{equation*}
    \lvert \hat{\beta}_{\mathrm{b}} - \beta \rvert = O_{\prob}(n^{-1}), \quad
    \left\lvert \hat{\beta}_{g} - \beta + \frac{\gamma_g}{\sqrt{n}} \right\rvert = O_{\prob}(n^{-1}),
\end{equation*}
for $g=1,2$,
which implies that
$\hat{\gamma}_g = \sqrt{n}(\hat{\beta}_{\mathrm{b}} - \hat{\beta}_{g})$ is a consistent estimator of $\gamma_g$ for $g=1,2$.
Finally, we define the following consistent estimators of the construct model parameters $(\kappa,q,\mu_1,\mu_2)$:
\begin{equation}
    \hat{\kappa} = \frac{n_1}{n}, \quad 
    \hat{q} = \frac{\hat{u}_{1,\mathrm{b}}}{1 - \hat{\beta}_{\mathrm{b}}}, \quad \hat{\mu}_g = \sqrt{n} \left( \frac{\hat{u}_{1,g}}{1 - \hat{\beta_g}} - \frac{\hat{u}_{1,\mathrm{b}}}{1 - \hat{\beta}_{\mathrm{b}}} \right)
    \ifjrss 
    \nonumber 
    \else 
    \label{construct_estimators}
    \fi
\end{equation}
for $g=1,2$.
The first estimator is simply the empirical proportion of minority group nodes.
\ifjrss
The other three estimators are found by adjusting the marginal blockwise connection probabilities under the observation model \eqref{obs_matrix} by the estimated type II error rates.
\else
For the other three, from \eqref{obs_matrix}, $\mathbb{E}\hat{u}_{1,\mathrm{b}} = (1-\beta)q$, and 
$$
    \mathbb{E}\hat{u}_{1,g} = \left\{ 1 - \left( \beta - \frac{\gamma_g}{\sqrt{n}} \right) \right\} \left( q + \frac{\mu_g}{\sqrt{n}} \right)
$$
for $g=1,2$. 
The estimators in \eqref{construct_estimators} are found by substituting the empirical $\hat{u}_{1,\mathrm{b}}$, $\hat{u}_{1,1}$ and $\hat{u}_{1,2}$ for their expectations; the error rate estimators $\hat{\beta}_{\mathrm{b}}$, $\hat{\beta}_1$, and $\hat{\beta}_2$ for $\beta$, $\beta - \gamma_1/\sqrt{n}$, and $\beta - \gamma_2/\sqrt{n}$, respectively; and solving for $q$, $\mu_1$, and $\mu_2$.
\fi
We will leverage these asymptotically normal estimators, and the general $\rho_K$-corrected ranking algorithm for detection (Section~\ref{subsec:detection}) and correction (Section~\ref{subsec:correction}) under the 2-group SBM with observation errors.

\subsection{Detection}\label{subsec:detection}

In this section, under our (undirected) 2-group SBM with observation errors, and assuming we have access to two independent network replicates, we develop an asymptotically valid test for the null hypothesis that no correction is needed for the degree ranking.
\ifjrss
\else
Recall that our 2-group SBM with observation errors has seven unknown parameters $(\kappa,p,\mu_1,\mu_2,\beta,\gamma_1,\gamma_2)$.
\fi

We develop a hypothesis test where the null hypothesis corresponds to a subset of the parameter space in which the corrected ranking is expected to be identical or nearly identical to the uncorrected ranking.
Formally, the null parameter space is made up of two subsets,
$
    H_{0,\bar{\beta}} = H_{0,\bar{\beta}}^{(\beta)} \cup H_0^{(\mu)},
$
where
$$
    H_{0,\bar{\beta}}^{(\beta)}: \gamma_1=\gamma_2=0, \beta \in [0,\bar{\beta}], \quad H_{0}^{(\mu)}: \gamma_1=\gamma_2=0, \mu_1=\mu_2=0
$$
for a tuning parameter $\bar{\beta} \in (0,1)$.
Under $H^{(\beta)}_{0,\bar{\beta}}$, the errors are balanced with $\gamma_1 = \gamma_2 = 0$, and the overall type II edge error rate is small, with $\beta  \leq \bar{\beta}$. This approximately corresponds to the WYSIWYG worldview.
We allow for small values of $\beta$ to fall in the null parameter space, since $\beta=0$ produces degenerate network replicates with no observation errors.
In practice, we suggest setting $\bar{\beta} = 0.1$; we see in Figure~\ref{fig:sbmlimits} (C) that the minority representation profiles in this range are nearly indistinguishable on the plot (compare the curves for $p=0.50$ and $p=0.43$).

To develop a hypothesis test for the first subset of the null parameter space, we first consider the degenerate case $\gamma_1=\gamma_2=0,\beta=0$. In this case we use a trivial test function $ \mathbb{I}(Y \neq Y^*)$ for any level $\alpha$, where $\mathbb{I}(\cdot)$ denotes the indicator function.
\ifjrss
\else
As $Y$ and $Y^*$ are exact observations of the same construct network $A$, this test will never reject, and defines a level $\alpha$ test for any $\alpha \geq 0$.
\fi
For $\gamma_1=\gamma_2=0$, $\beta \in (0,\bar{\beta}]$, the errors are balanced but non-zero. Assuming $Y \neq Y^*$, we find a $\chi_3^2$ asymptotic pivotal quantity.

Under $H_0^{(\mu)}$, the errors are not systematically biased, and the construct model has proportional representation, as $\mu_1=\mu_2=0$ and $\gamma_1 = \gamma_2 = 0$. This corresponds to the WAE worldview with no systematically biased errors.
In this case, we can find a $\chi_4^2$ asymptotic pivotal quantity.
Our analysis in Section~\ref{sec:theory}, in particular Figure~\ref{fig:sbmlimits}(c), verified that $\gamma_1=\gamma_2=0$ alone does not guarantee no systematic bias will occur.
Even when $\gamma_1=\gamma_2=0$, observation error can affect the representation profile of a network due to its influence on edge density. 

Our final test is defined in Theorem~\ref{thm:hyptest}, based on the intersection-union principle. 
Detailed derivations, and a proof of the following Theorem~\ref{thm:hyptest} are given in Section 1.3 of the Supplementary Material.
We confirm its empirical performance in simulation in Section~\ref{subsubsec:detection_sims}.

\begin{theorem} \label{thm:hyptest}
    Suppose $q \in (0,1)$, $\beta \in (0,1)$ and $\bar{\beta} \in (0,1)$.
    If $Y \neq Y^*$, define statistics
    \begin{align*}
        Q_{n,\bar{\beta}}(\beta) &= (n(\hat{\beta}_{\mathrm{b}} - \beta),\sqrt{n}\hat{\gamma}_1,\sqrt{n}\hat{\gamma}_2)^{\tp} \widehat{\Theta}_{\beta} (n(\hat{\beta}_{\mathrm{b}} - \beta),\sqrt{n}\hat{\gamma}_1,\sqrt{n}\hat{\gamma}_2), \\
        Q_{n,\mu} &= n (\hat{\gamma}_1,\hat{\gamma}_2,\hat{\mu}_1,\hat{\mu}_2)^{\tp} \widehat{\Theta}_{\mu} (\hat{\gamma}_1,\hat{\gamma}_2,\hat{\mu}_1,\hat{\mu}_2),
    \end{align*}
    where $(\hat{\beta}_{\mathrm{b}},\hat{\gamma}_1,\hat{\gamma}_2,\hat{\mu}_1,\hat{\mu}_2)$ are defined in Section~\ref{subsec:estimation}, and $\widehat{\Theta}_{\beta}$ and $\widehat{\Theta}_{\mu}$ are consistent estimators of asymptotic precision matrices (detailed definitions are given in Section 1.3 of the Supplementary Material).
    Define a test function
    $$
        \Psi_{\alpha}(Y,Y^*,\bm{c}) = \mathbb{I}(Y \neq Y^*) \cdot \mathbb{I}\left\{ \min_{\beta \in [0,\bar{\beta}]} Q_{n,\bar{\beta}}(\beta) > c^{(3)}_{\alpha} \right\} \cdot \mathbb{I}\left( Q_{n,\mu} > c^{(4)}_{\alpha} \right),
    $$
    where $c^{(\nu)}_{\alpha}$ is the $1-\alpha$ quantile of a $\chi^2_{\nu}$ distribution.
    Then, under $H_{0,\bar{\beta}}$, the test defined by $\Psi_{\alpha}$ controls type I errors asymptotically at level $\alpha$.
\end{theorem}

\begin{remark}
Note $Q_{n,\bar{\beta}}(\beta)$ is strictly convex, as it is a quadratic function of $\beta$, and $\widehat{\Theta}_{\beta}$ is positive definite. Thus, $Q_{n,\bar{\beta}}(\beta)$ will have a unique minimum on the closed interval $[0,\bar{\beta}]$.
\end{remark}


\subsection{Ranking Correction} \label{subsec:correction}

Given that we reject the hypothesis in Section~\ref{subsec:detection}, and conclude there is a need for ranking correction, it is natural use the estimators of the construct model SBM parameters and error rates to correct the degree ranking, which we summarize in Algorithm~\ref{alg:plugin_correction}.

\begin{algorithm}
  \SetAlgoLined
  \caption{Plug-in corrected degree ranking  \label{alg:plugin_correction}}
  \KwResult{Corrected top-$K$ sets based on estimated construct model}
  \KwIn{Network replicates $Y$, $Y^*$, group labels $\bm{c}$}
  
  Compute sufficient statistics $(n_1,\hat{u}_{1,1},\hat{u}_{2,1},\hat{u}_{1,2},\hat{u}_{2,2},\hat{u}_{1,\mathrm{b}},\hat{u}_{2,\mathrm{b}})$ from $Y$, $Y^*$, and $\bm{c}$ \; 
  $\hat{\beta}_{\mathrm{b}} \gets \hat{u}_{2,\mathrm{b}} ~/~ \hat{u}_{1,\mathrm{b}} $ ; $\hat{\beta}_{1} \gets \hat{u}_{2,1} ~/~ \hat{u}_{1,1} $ ; $\hat{\beta}_{2} \gets \hat{u}_{2,2} ~/~ \hat{u}_{1,2} $ \;
  $\hat{\kappa} \gets n_1 ~/~ n$ ; $\hat{q} \gets \hat{u}_{1,\mathrm{b}} ~/~ (1- \hat{\beta}_{\mathrm{b}})$ \;
  $\hat{\mu}_1 \gets \sqrt{n}[ \{ \hat{u}_{1,1} ~/~ (1 - \hat{\beta}_1)\} - \hat{q} ] $ ; $\hat{\mu}_2 \gets \sqrt{n}[ \{ \hat{u}_{1,2} ~/~ (1 - \hat{\beta}_2)\} - \hat{q} ] $\;
  \For{$K \gets 1$ \KwTo $n$}{
    $\hat{\rho}_K^* = \rho^*(K/n ~;~ \hat{\kappa},\hat{q},\hat{\mu}_1,\hat{\mu}_2)$ (see Theorem~\ref{thm:SBM_limits}) \;
  }
  Compute $\hat{\rho}_K$-corrected degree ranking from $Y$, $\bm{c}$, $\{\hat{\rho}_K^*\}_{K=1}^n$ (Algorithm~\ref{alg:correction})
\end{algorithm}


This correction algorithm shows the utility of a good parametric approximation to the minority representation profile: if we can estimate the model parameters, we can estimate a target representation profile, and adaptively correct the degree ranking.
We evaluate the performance of this correction algorithm on simulated networks in Section~\ref{subsubsec:correction_sims}.

\subsection{Simulations}

\subsubsection{Detecting systematic bias} \label{subsubsec:detection_sims}

Under an errorfully observed SBM with type II observation errors, we perform simulations to evaluate the operating characteristics of the detection approach described in Section~\ref{subsec:detection}.
We evaluate our approach for two noisy network replicates with $n \in \{20,50,100\}$ nodes and $\kappa=1/4$. 
In order to ignore variability in $n_1$, we fix $n_1$ to be exactly equal to the integer floor of its expectation, $\lfloor \kappa n \rfloor$.
We evaluate empirical power with nominal type I error rate $\alpha=1/10$ for $400$ independent replications.
The remaining free model parameters are $q$, $\mu_1$, $\mu_2$, $\beta$, $\gamma_1$, and $\gamma_2$. For simplicity, we will set $\mu_1=\mu_2$ and $\gamma_1=\gamma_2$ throughout.

In scenario (A), we fix $q=1/5$, $\gamma_1=\gamma_2=0$ and vary $\beta$. 
As increasing $\beta$ attenuates the signal strength, we set $\mu_1 = \mu_2 = 1 / 2(1-\beta)$. 
For $\beta \leq \bar{\beta} = 0.1$, the null hypothesis ($H^{(\beta)}_{0,\bar{\beta}}$) is true; the observation error rate is low and balanced.
In scenario (B), we fix $q=1/4$, $\beta=1/5$, $\gamma_1=\gamma_2=0$ and vary $\mu$, the common value of $\mu_1$ and $\mu_2$.
For $\mu=0$, the null hypothesis ($H^{(\mu)}_{0}$) is true; the construct and observation models are both ER.
In scenario (C), we fix $q=1/4$, $\mu_1=\mu_2=0$, $\beta=1/4$, and vary $\gamma$, the common value of $\gamma_1$ and $\gamma_2$. 
For $\gamma=0$, the null hypothesis ($H^{(\mu)}_{0}$) is true; the construct and observation models are both ER.
Finally, in scenario (D) we fix $\mu_1=\mu_2=1/10$, $\beta=1/5$, $\gamma_1=\gamma_2=1/2$ and vary $q$.
None of the scenario (D) settings are null.
All results are shown in Figure~\ref{fig:detection}. 
Power $1$ is denoted by a dotted line, and the nominal level $\alpha=1/10$ is denoted by a dashed line.

\begin{figure}[!ht]
  \centering
  \begin{subfigure}[b]{0.48\linewidth}
    \centering
    \includegraphics[width=\linewidth]{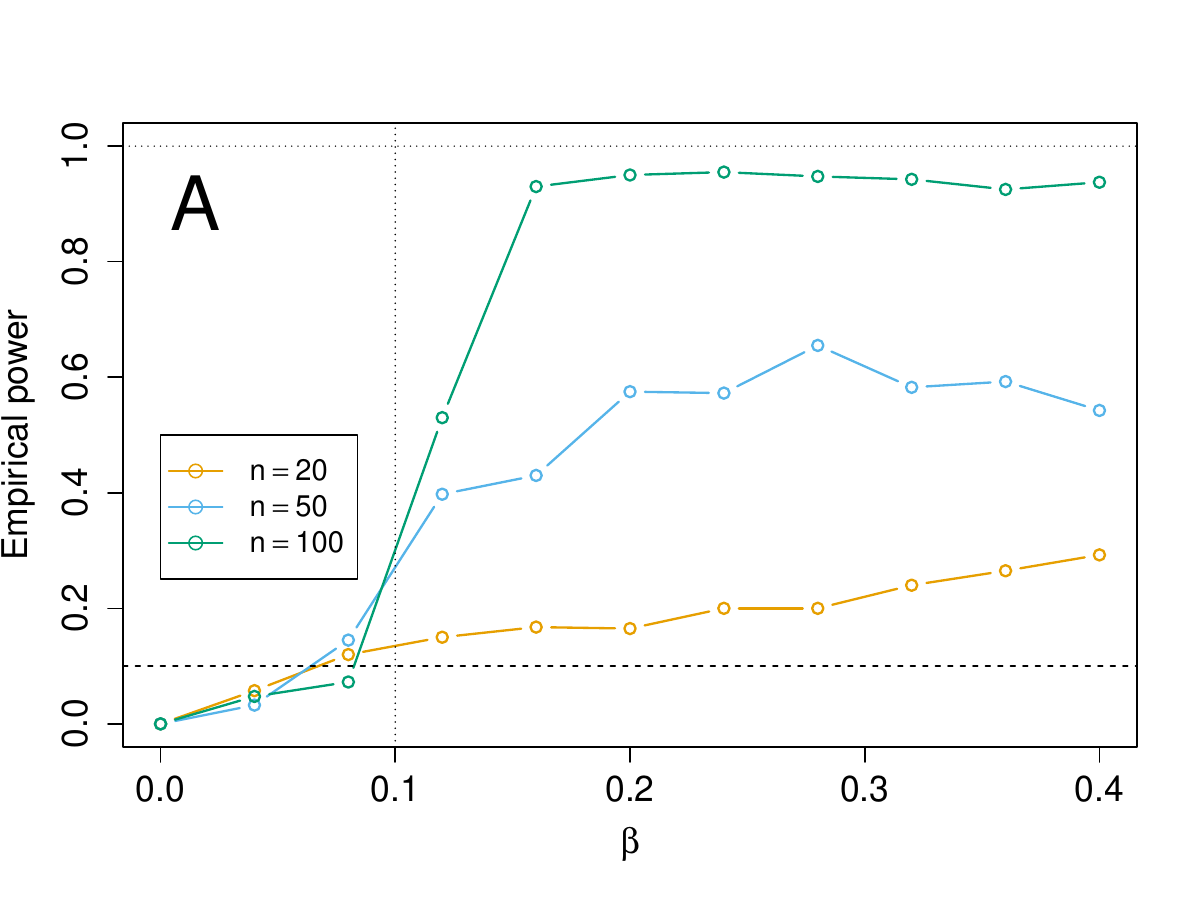}
    \caption{$q=\tfrac{1}{5}$, $\mu_1=\mu_2=\tfrac{1}{2}(1-\beta)$, $\gamma_1=\gamma_2=0$, varying $\beta$.}
  \end{subfigure}
  \hfill
  \begin{subfigure}[b]{0.48\linewidth}
    \centering
    \includegraphics[width=\linewidth]{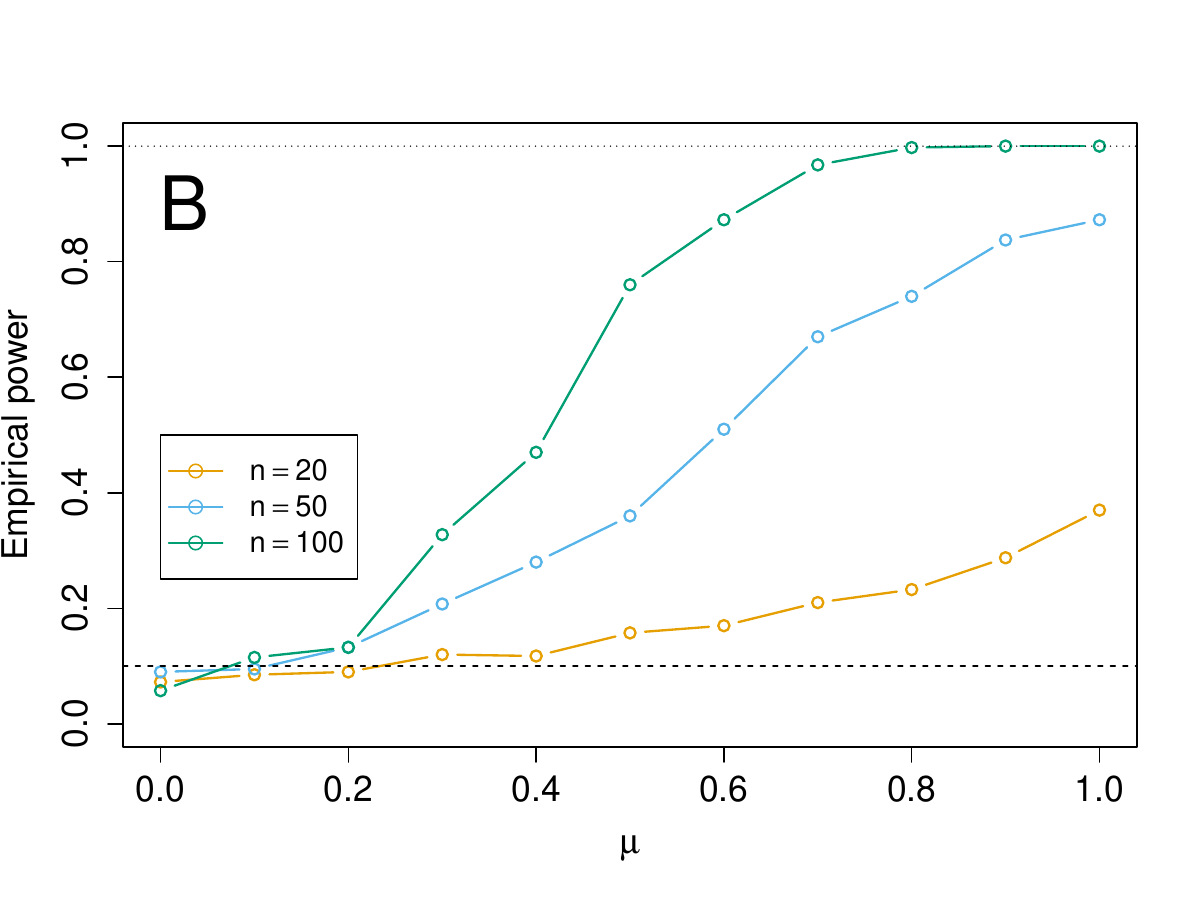}
    \caption{$q=\tfrac{1}{4}$, $\beta=\tfrac{1}{5}$, $\gamma_1=\gamma_2=0$, varying $\mu$, where $\mu_1=\mu_2=\mu$.}
  \end{subfigure}

  \medskip 

  \begin{subfigure}[b]{0.48\linewidth}
    \centering
    \includegraphics[width=\linewidth]{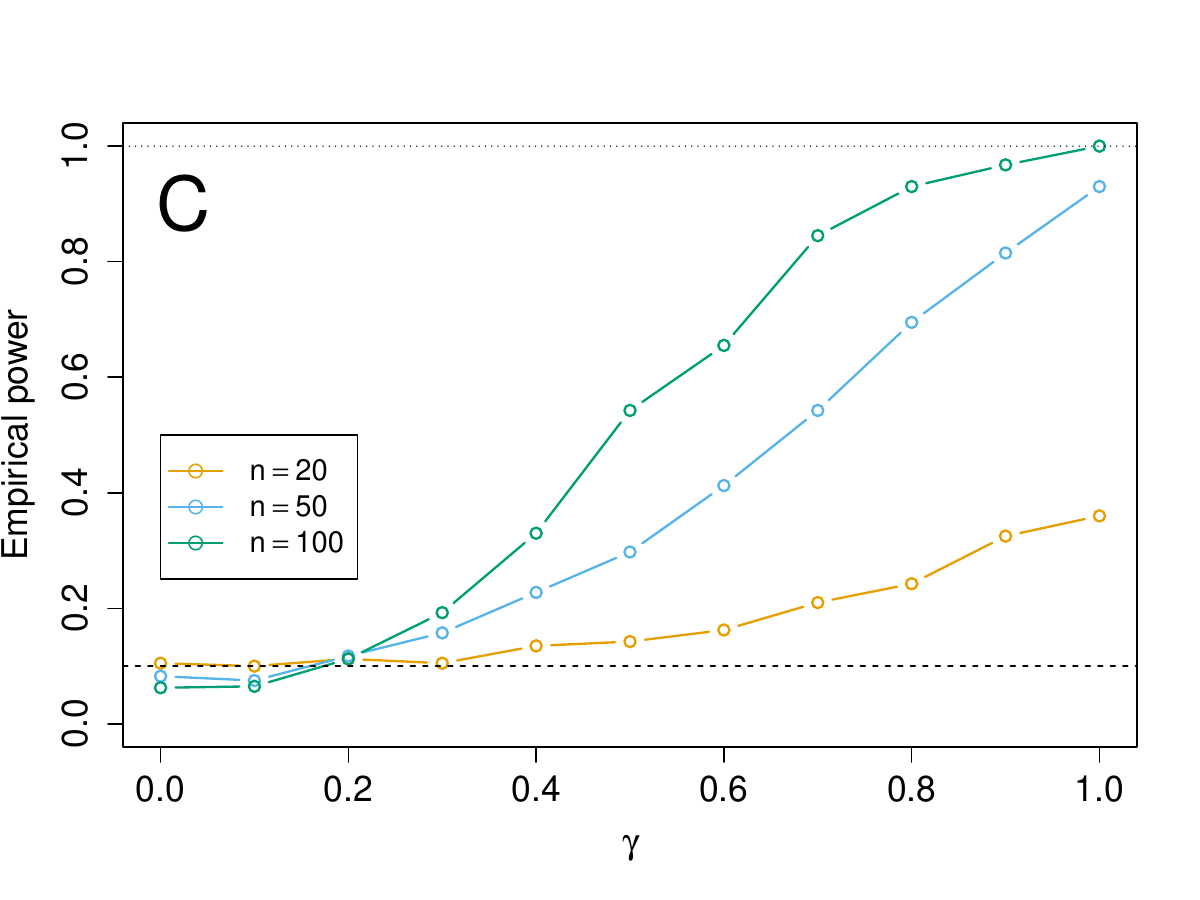}
    \caption{$q=\tfrac{1}{4}$, $\mu_1=\mu_2=0$, $\beta=\tfrac{1}{4}$, varying $\gamma$, where $\gamma_1=\gamma_2=\gamma$.}
  \end{subfigure}
  \hfill
  \begin{subfigure}[b]{0.48\linewidth}
    \centering
    \includegraphics[width=\linewidth]{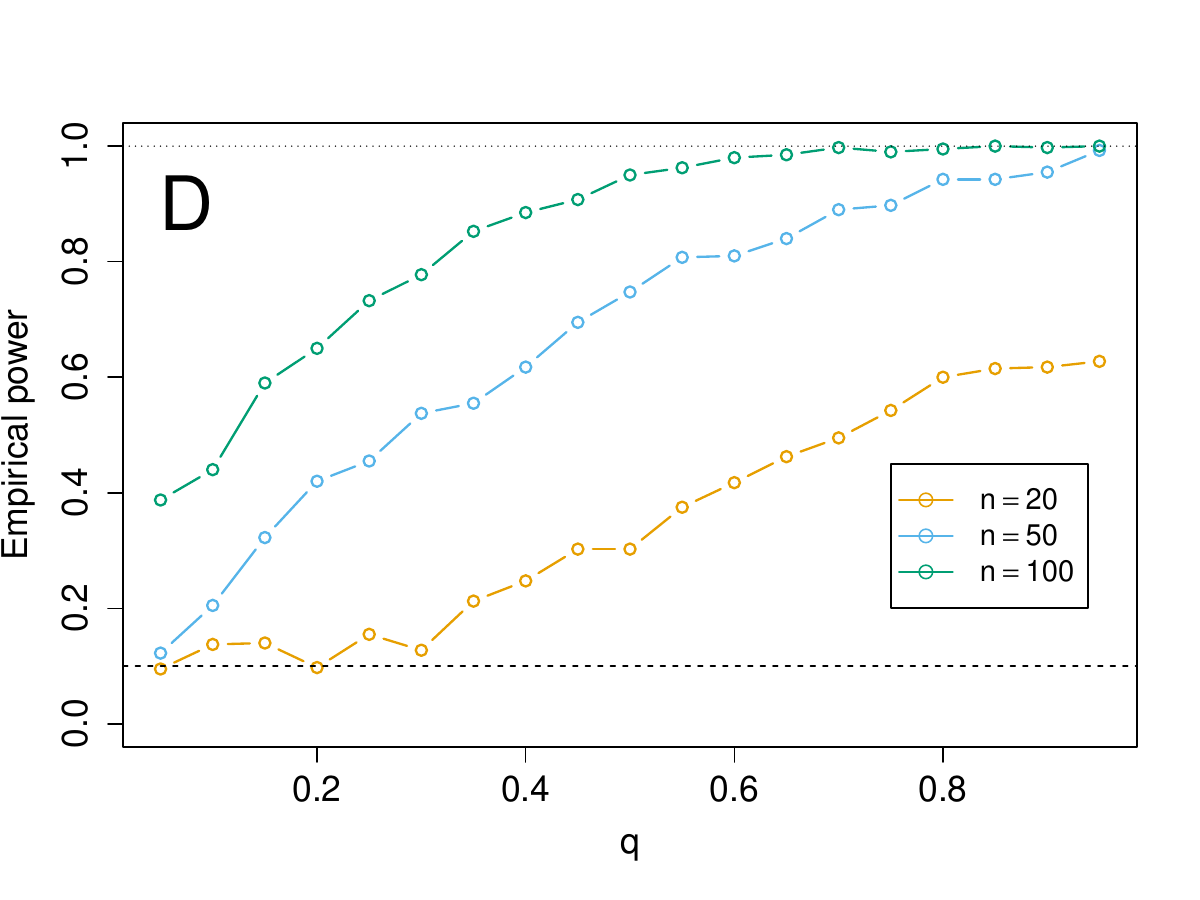}
    \caption{$\mu_1=\mu_2=\tfrac{1}{10}$, $\beta=\tfrac{1}{5}$, $\gamma_1=\gamma_2=0.5$, varying $q$.}
  \end{subfigure}

  \caption{Empirical rejection rate of systematic bias detection for various parameter settings, 2-group SBM with $\kappa=0.25$.}
  \label{fig:detection}
\end{figure}

We see that the type I error rate is controlled in the null settings in scenarios (A) (left of the vertical dotted line), (B) and (C).
In alternative settings
in scenarios (A), (B) and (C), power gets very close to $1$ as $n$ and any of the three signal parameters $\beta$, $\mu$, or $\gamma$ increase.
In scenario (D), none of these signal parameters are varied.
However, the effective sample size to estimate the type II error rate parameters $\beta_w$ and $\beta_b$ depends on the number of edges present, not the number of potential edges, thus test power increases with overall edge density and $n$.

\subsubsection{Correcting systematic bias} \label{subsubsec:correction_sims}

Under an errorfully observed SBM with type II observation errors, we perform simulations to compare our parametric correction approach to an uncorrected ranking and a proportionally representative ranking.
Ties in rankings are resolved uniformly at random.

\ifjrss
\else
In all settings, we generate a construct adjacency matrix $A$ on $n$ nodes with parameters $\kappa$, $p$, $\mu_1$ and $\mu_2$.
We then generate two conditionally independent observed adjacency matrices with type II error rates parameterized by $\beta$, $\gamma_1$, and $\gamma_2$.
\fi

\ifjrss
\else
In Figure~\ref{fig:correction_varyb} we show the results of the three ranking approaches in terms of their Spearman rank correlation coefficient, compared to a construct model node ranking based on the degrees in the construct network $A$. 
We set $n=200$, $\kappa=2/5$, $q=1/2$, $\mu_1=\mu_2=-2$, and vary $\gamma$, the common value of $\gamma_1$ and $\gamma_2$.
We set $\beta$ as a function of $\gamma$ such that $\beta - \gamma/\sqrt{n} = 1/10$, so that the within-group type II error rates are fixed.
Note that when $\mu < 0$ and $\gamma > 0$, the construct network model favors between-group connections, while for sufficiently large $\gamma$ (greater than about 3), the observed network favors within-group connections.

\begin{figure}
        \centering
        \includegraphics[width=0.7\linewidth]{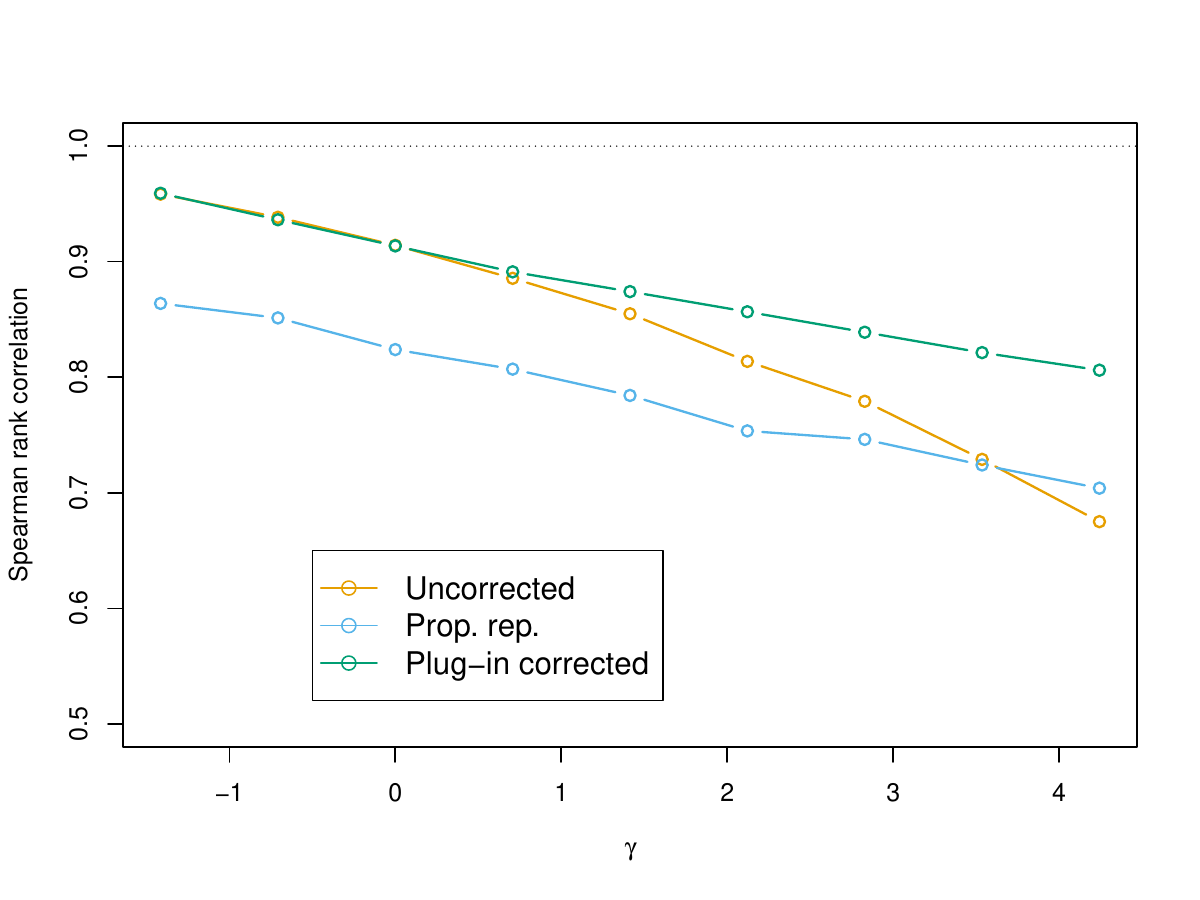}
        \caption{Performance of ranking correction varying $\gamma$, where $\gamma = \gamma_1 = \gamma_2$; $n=200$, $\kappa=2/5$, $q=1/2$, $\mu_1=\mu_2=-2$, $\beta = 1/10 + \gamma/\sqrt{n}$.} 
        \label{fig:correction_varyb}
\end{figure}

For $\gamma > 0$, our plug-in correction performs the best at approximating the ranking based on $A$.
None of the methods achieve a perfect correlation due to the intrinsic variability of the observed degrees.
The performance of the uncorrected ranking is comparable to that of the plug-in correction when the error rates are equal.
\fi

\begin{figure}[!ht]
  \centering
  \begin{subfigure}[b]{0.475\linewidth}
    \centering
    \includegraphics[width=\linewidth]{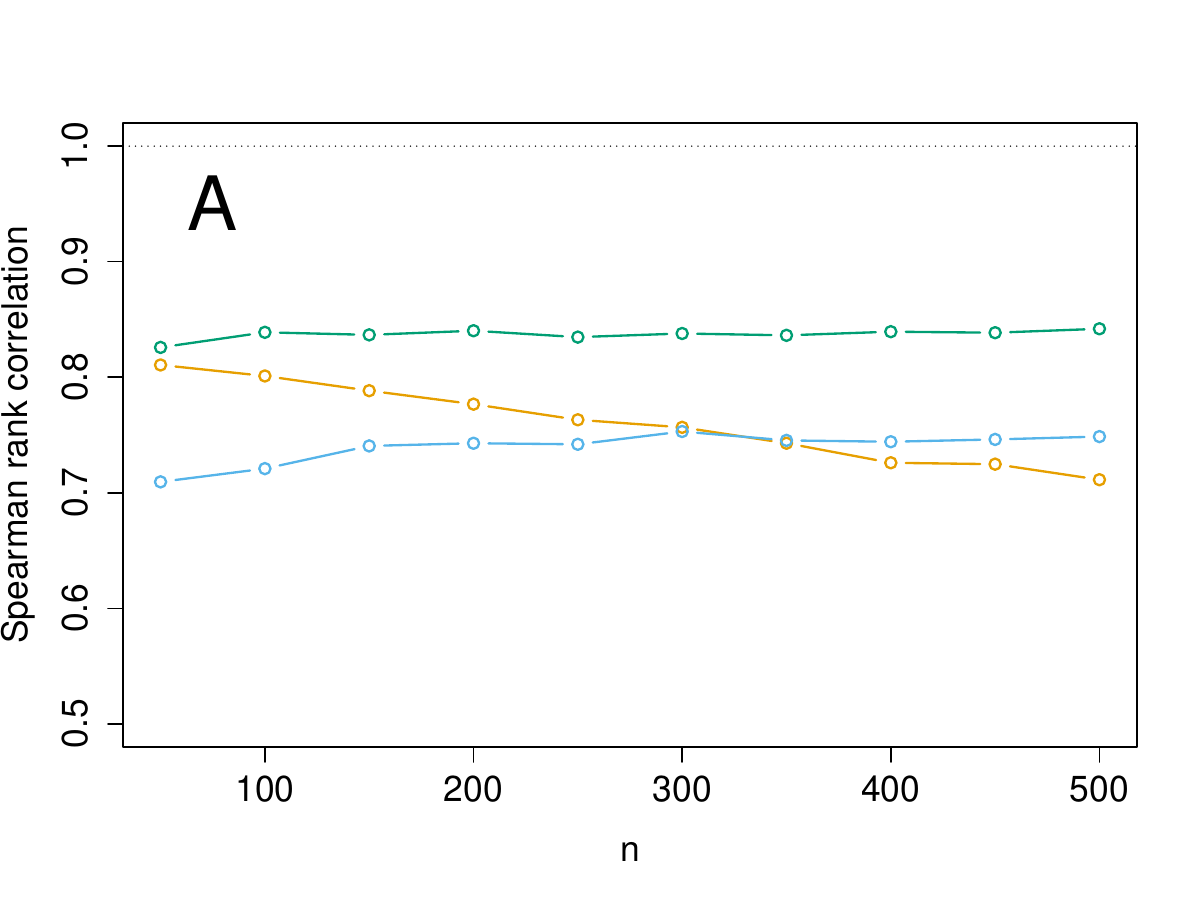}
    \caption{SBM with observation errors; $\mu_1=\mu_2=-2$, $\beta=\tfrac{3}{10}$, $\gamma_1=\gamma_2 = \tfrac{\sqrt{n}}{5}$, varying $n$.}
  \end{subfigure}
  \hfill
  \begin{subfigure}[b]{0.475\linewidth}
    \centering
    \includegraphics[width=\linewidth]{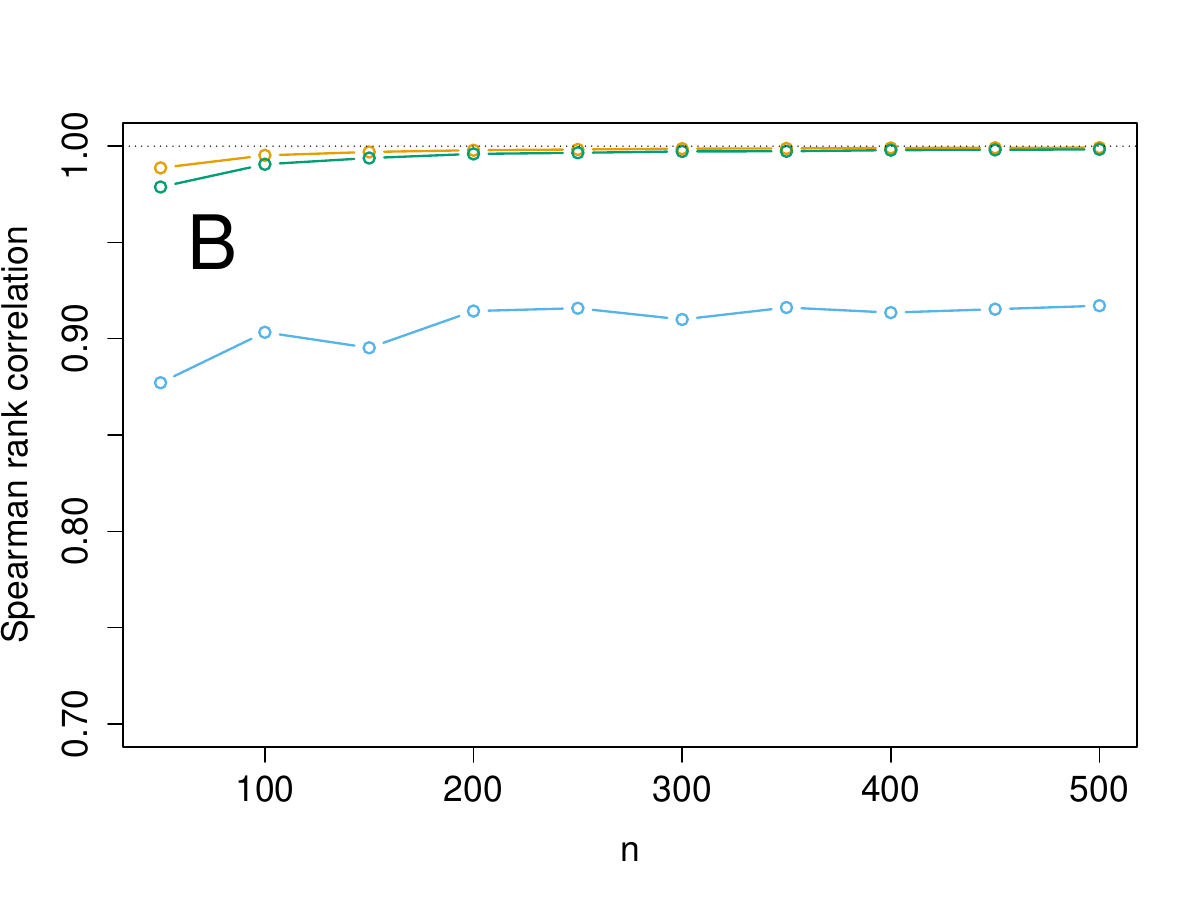}
    \caption{SBM with no observation errors; $\mu_1=\mu_2=-2$, $\beta=0$, $\gamma_1=\gamma_2=0$, varying $n$.}
  \end{subfigure}

  \medskip 

  \begin{subfigure}[b]{0.475\linewidth}
    \centering
    \includegraphics[width=\linewidth]{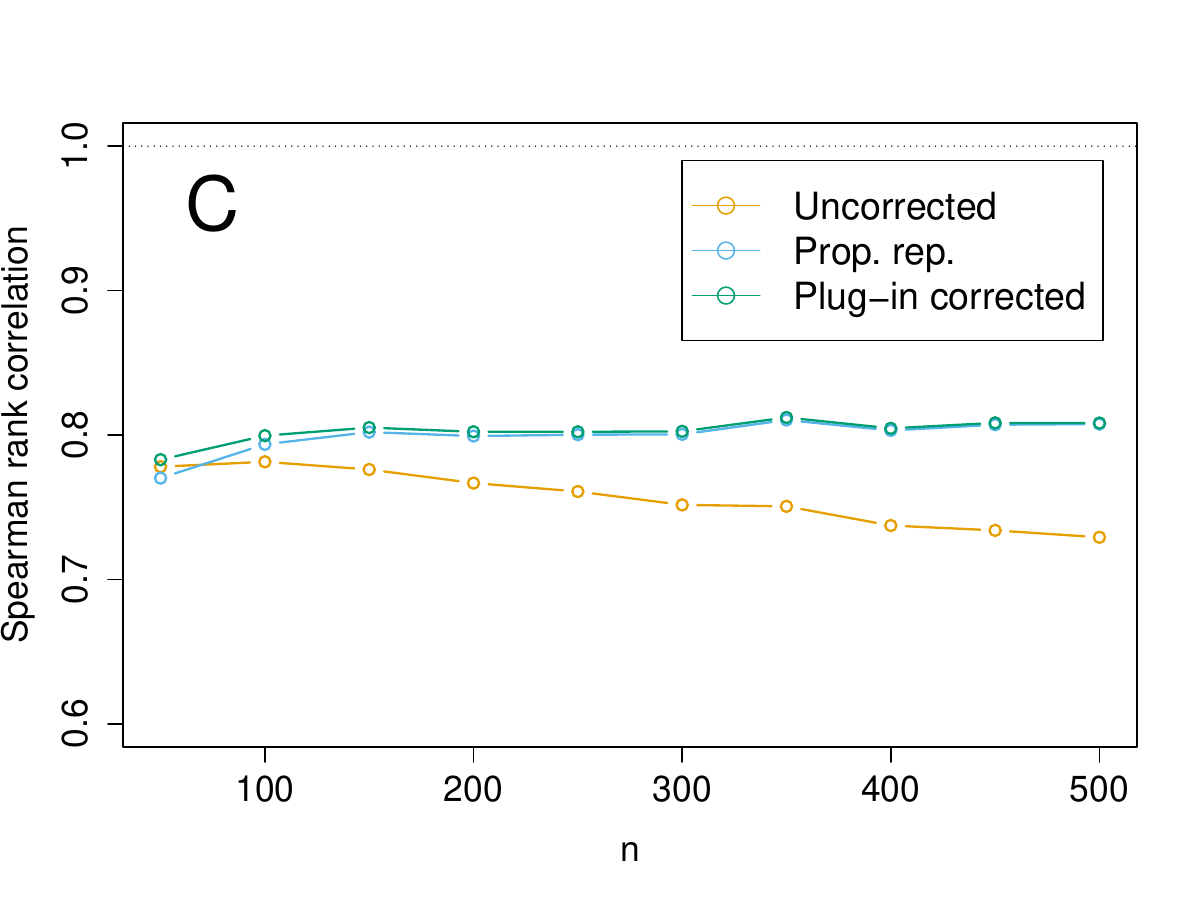}
    \caption{ER model with observation errors; $\mu_1=\mu_2=0$, $\beta=\tfrac{3}{10}$, $\gamma_1=\gamma_2=\tfrac{\sqrt{n}}{5}$, varying $n$.}
  \end{subfigure}
  \hfill
  \begin{subfigure}[b]{0.475\linewidth}
    \centering
    \includegraphics[width=\linewidth]{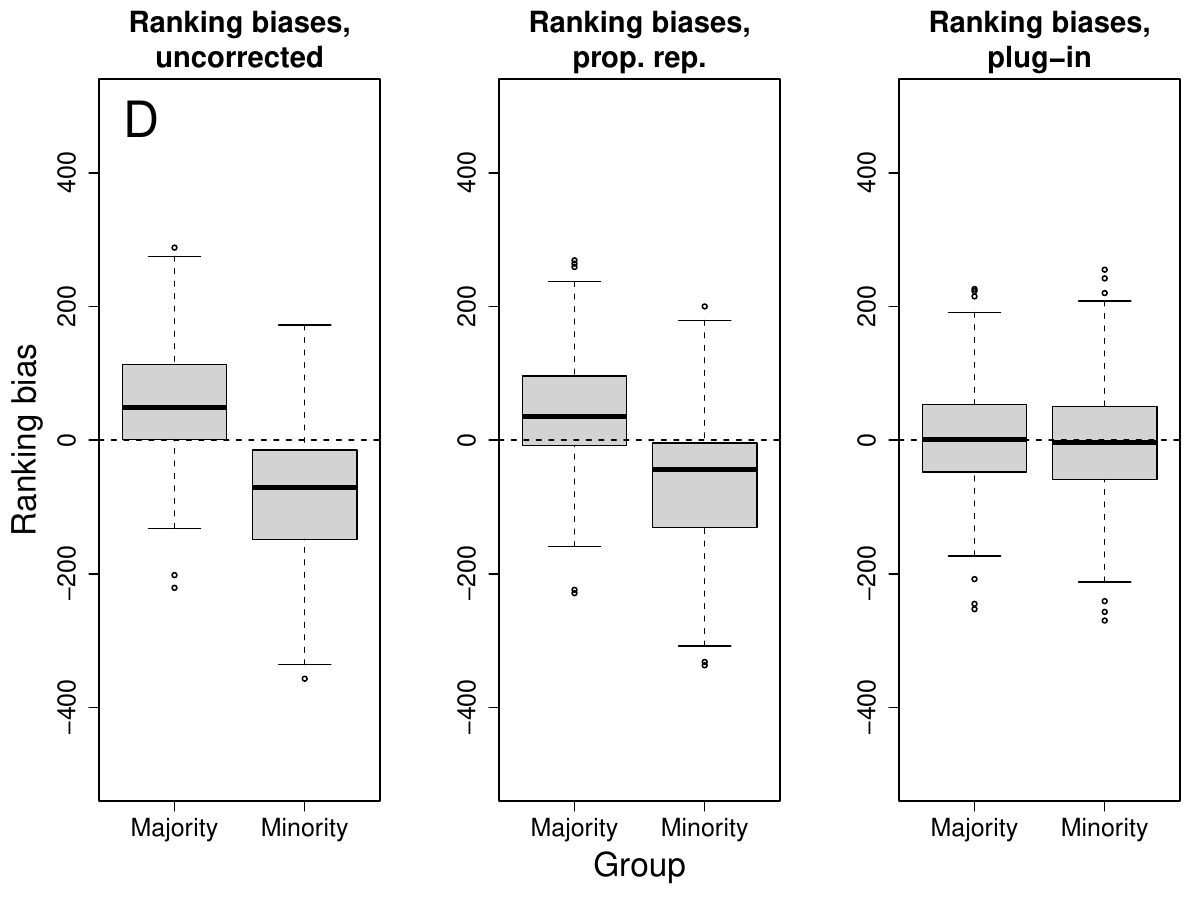}
    \caption{Box plots of ranking bias; $n=500$, $\mu_1=\mu_2=-2$, $\beta=\tfrac{3}{10}$, $\gamma_1=\gamma_2=2\sqrt{5}$. Ranking bias: ground truth ranking minus estimated/corrected ranking. Positive bias implies nodes are ranked too high.}
  \end{subfigure}

  \caption{(A)–(C): Performance of ranking correction varying $n$ under SBM, $\kappa=\tfrac{2}{5}$, $q=\tfrac{1}{2}$. (D): visual inspection of ranking bias.}
  \label{fig:correction_varyn}
\end{figure}

In Figure~\ref{fig:correction_varyn}, we show the performance of 3 parameter settings as $n$ increases.
\ifjrss
Performance is in terms of Spearman rank correlation coefficient, compared to a construct model node ranking based on the degrees in the construct network $A$. 
In all cases, although our approach is based on an asymptotic approximation, its performance is good even for $n = 50$.
\else
\fi
In scenario (A), the construct model ranking favors the minority class, but the observation errors lead to minority underrepresentation among highly ranked nodes.
In scenario (B), there are no observation errors, 
so there is no need for ranking correction.
Moreover, as the construct model does not satisfy proportional representation, the proportionally representative ranking correction performs poorly relative to the other two approaches.
In scenario (C), the construct model does satisfy proportional representation, so the proportionally representative ranking correction is a good choice.
In panel (D), we show box plots of the average bias in the rankings for the different groups, for one replicate, under scenario (A) with $n=500$. 
We can see clearly that the uncorrected ranking favors the majority class. While the proportionally representative ranking does some correction to reduce this systematic bias, our plug-in ranking correction fully removes it, as it detects that in the construct network, between-group connections occur at a higher rate than within-group connections.

\ifjrss
\else
Although our method relies on an asymptotic approximation, its performance is good even for small $n$, and in panel (A) it outperforms the uncorrected and proportional approaches for $n$ as small as $50$.
In panel (B), there are no observation errors and the performance of the plug-in correction is comparable to that of the uncorrected ranking, and nearly identical for $n \geq 300$.
In panel (C), the construct ranking is already proportionally representative on average, the plug-in correction is nearly identical to the proportionally representative correction for $n \geq 100$.
\fi

\section{Analysis of high school contact data}\label{sec:realdata}

We analyze the contact diary data introduced in Section~\ref{sec:intro} using our proposed estimation, testing, and correction procedures. The dataset consists of directed reports of contacts among 120 students at Lycée Thiers in Marseilles, France, collected in December 2013 \citep{mastrandrea2015contact}. In addition to reported contacts, the data record each student’s gender (male/female) and academic specialization—MP (mathematics and physics), PC (physics and chemistry), and BIO (biology)—with each specialization comprising multiple classes (e.g., BIO1, BIO2, BIO3 for biology).  

For ranking, we focus on in-degrees, which better reflect how often a student is recognized by peers as a contact. This choice is supported by \citet{mastrandrea2015contact}, who found that in-degree is positively correlated with degree in the corresponding sensor-based network, whereas the
correlation for out-degree was insignificant.

For data preprocessing, we first excluded class MP2, which had only one student and no interactions. Figure~\ref{fig:combined_connection_probability} shows the estimated connection probability matrix across classes and genders. Most classes exhibit a relatively homogeneous connection pattern, with more interactions within the same gender, while BIO1 and BIO3 display more cross-gender connections. 

Although our methods can handle both assortative and disassortative structures, we restrict attention to subsets with coherent mixing patterns to ensure clearer interpretation. For this reason, we exclude BIO1 and BIO3 from the analysis. The remaining classes BIO2, MP1, MP3, and PC1 all display assortative structures and are thus more suitable for applying our proposed methods. 

The selected subset consists of 83 individuals, including 37 females, a female proportion of 0.45, indicating a slight minority presence. As shown in Figure~\ref{fig:proportion_females_rho}, the blue curve represents the proportion of females among the top-$K$ nodes by naive degree ranking, revealing disproportional representation for smaller $K$. 
\begin{figure}
    \centering
    \includegraphics[width=0.5\linewidth]{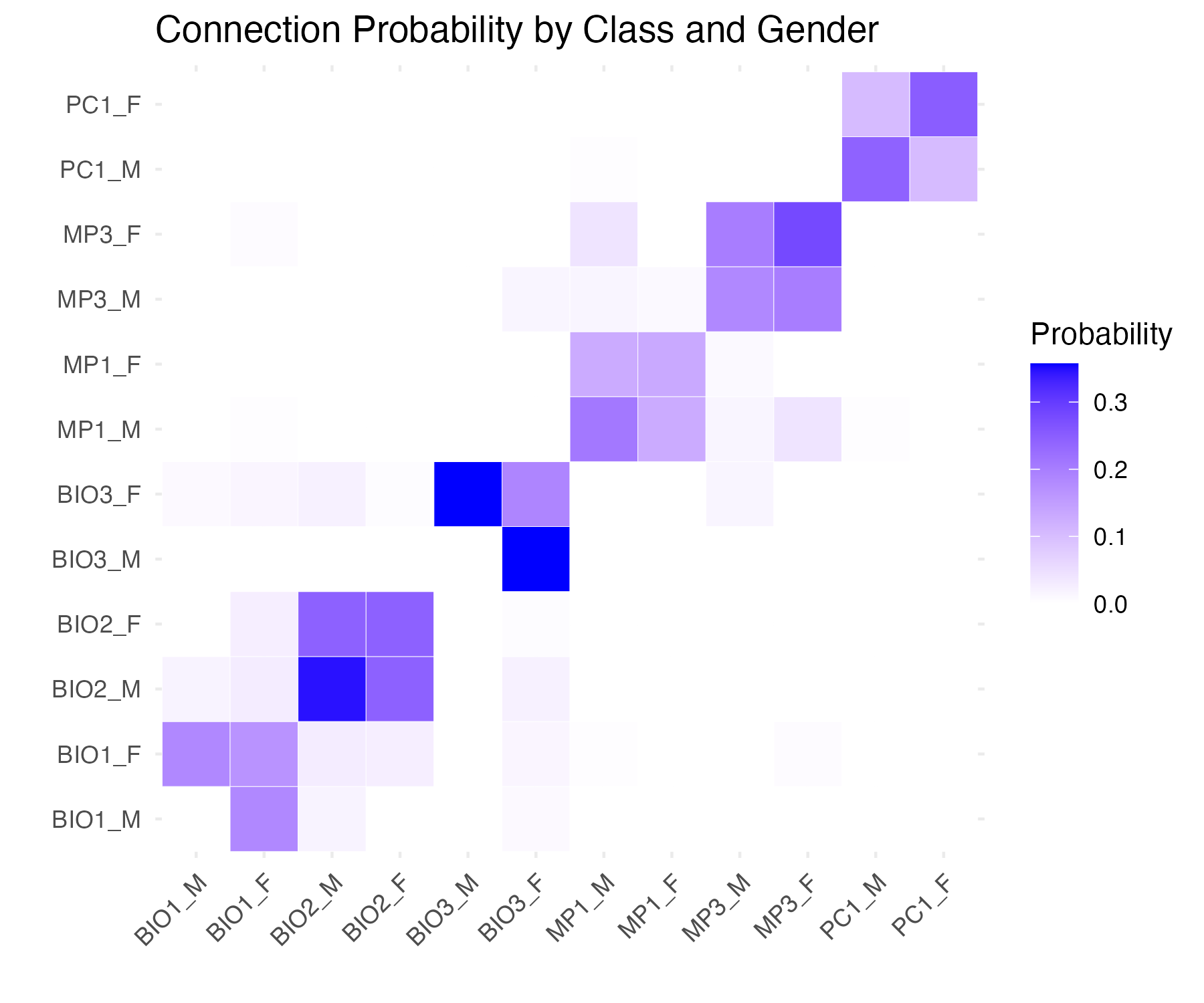}
       \caption{Average contact probabilities by classes and genders, contact diary data. Each cell shows the estimated probability between student groups.}
    \label{fig:combined_connection_probability}
\end{figure}

As noted in Section~\ref{sec:correction}, a noisy directed network can be treated analogously to two independent replicates, and the relevant methodology is detailed in 
Section 1.4 of the Supplementary Material. Previous research has identified systematic recall biases in contact diary data, particularly those driven by gender \citep{quintane2025gender}. Figure~\ref{fig:intro_bias_overview}, panel (A) shows a strong asymmetry in cross-gender reporting: 
male students appear less likely to record contacts with female peers than the reverse. This asymmetry is more important than overall assortativity, as it reflects a systematic recall bias that can distort network-based measures of influence. To formally test for gender-related bias, we define the hypothesis as
\[
H_{1,0}: \beta_{12} = \beta_{21}, \qquad H_{1,1}: \beta_{12} < \beta_{21},
\]
where $\beta_{12}$ is the probability that a female forgets a connection with a male, and 
$\beta_{21}$ is the probability that a male forgets a connection with a female.

Our proposed testing procedure yields evidence against the null hypothesis 
($p = 0.021$), with estimates $\hat\beta_{12} = 0.39$ and $\hat\beta_{21} = 0.53$. 
This result highlights the presence of systematic bias and underscores the need for correction 
to ensure equitable analysis. The remaining estimates are 
$\hat\beta_{11} = 0.36$ and $\hat\beta_{22} = 0.37$. First, the difference between $\hat\beta_{21}$ and $\hat\beta_{22}$ shows that the imbalance cannot be explained by overall recall tendencies of males, but instead reflects a directional bias against female students. Second, the $\beta$ estimates are consistent with the findings of \cite{mastrandrea2015contact}.  In their comparison of diary-based network with 
the sensor-based network recorded on the same day, they 
observed that ``only 41.4\% of the contacts registered by the sensors find a match in the 
contact diary.'' This comparison illustrates the usefulness of our model: by formally modelling and estimating heterogeneous recall error rates, it recovers more accurate degree-based rankings even when the construct network is observed with bias, and provides more reliable inputs for downstream analyses.

\citet{mastrandrea2015contact} found that the probability of forgetting a connection is negatively correlated with the duration of the contact: longer interactions are more likely to be remembered. To further investigate this effect, we exploit the fact that the original contact diary data encode contact duration through the weight variable $w$: (i) at most 5 minutes if $w=1$; (ii) 5--15 minutes if $w=2$; (iii) 15 minutes--1 hour if $w=3$; and (iv) more than 1 hour if $w=4$. Based on this categorization, we construct two directed diary-based networks restricted to short contacts ($w \in \{1,2\}$) and long contacts ($w \in \{3,4\}$). In both networks we continue to observe cross-gender recall imbalance, as illustrated in Figure~\ref{fig:prob_matrices}. However, the asymmetry is more pronounced for short contacts and appears to diminish as duration increases. Applying our estimation and testing procedure to the two networks yields $p$-values of 0.0046 and 0.293 for testing $H_{0}: \beta_{12} = \beta_{21}$ against $H_{1}: \beta_{12} < \beta_{21}$. These results indicate that interaction duration is strongly associated with cross-gender recall bias: shorter male--female interactions are particularly underreported, whereas longer interactions are recalled more symmetrically. This pattern is consistent with the observation in \citet{mastrandrea2015contact} that longer interactions are less likely to be forgotten.

\begin{figure}[tbp]
    \centering
    \begin{subfigure}{0.48\linewidth}
        \centering
        \includegraphics[width=\linewidth]{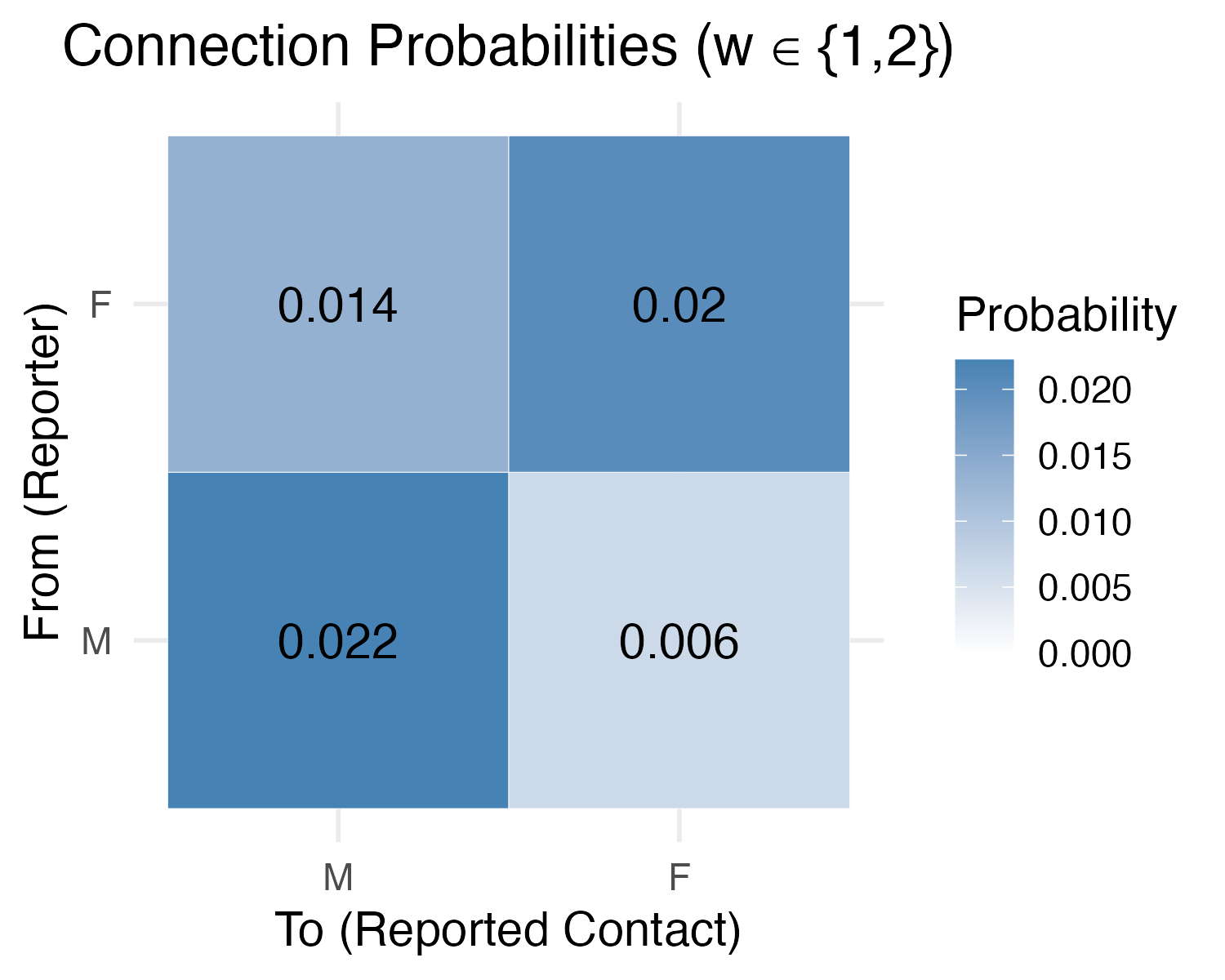}
        \caption{Connection probabilities for short contacts, with nodes subsetted by gender.}
    \end{subfigure}
    \hfill
    \begin{subfigure}{0.48\linewidth}
        \centering
        \includegraphics[width=\linewidth]{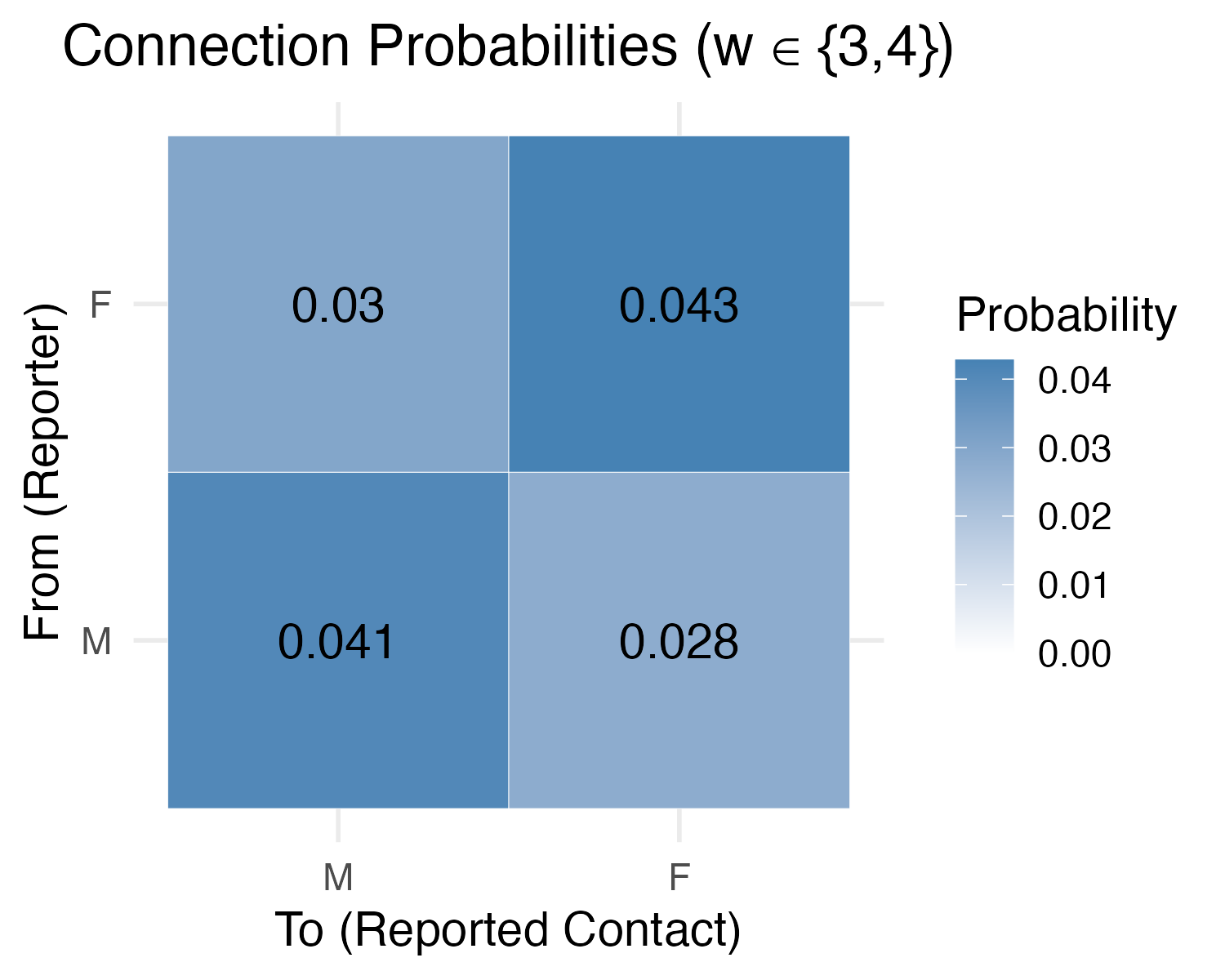}
        \caption{Connection probabilities for long contacts, with nodes subsetted by gender.}
    \end{subfigure}
    \caption{Estimated connection probability matrices by duration, contact diary data. Short contacts exhibit stronger cross-gender recall asymmetry, whereas long contacts are reported more symmetrically.}
    \label{fig:prob_matrices}
\end{figure}

To correct for systematic bias, we apply our parametric approach to compute a corrected node ranking based on in-degrees. 
Figure \ref{fig:proportion_females_rho} presents the proportion of females under different methods, where the red curve represents the target representation profile of females. 
Leveraging replication from the directed diary-based network, we can perform a principled ranking correction without assuming WYSIWYG or WAE. For instance, in the empirical ranking, only one of the top-10 students is female. In contrast, our corrected ranking includes four female students among the top-10, demonstrating a notable shift in representation after accounting for recall bias.  A detailed comparison of the top-10 students before and after correction, including their IDs and genders, is provided in Table~\ref{tab:top10_comparison}. Note that some reordering among male students arises from ties in in-degree, which are broken randomly under both the empirical and corrected rankings. 

To evaluate the effectiveness of our re-ranking algorithm, we compare to the sensor-based network collected on the same day as the diary-based network. Although the diary-based and sensor-based networks capture different aspects of student interaction, the latter provides a reasonable proxy for true contact intensity, as it is based on directly measured physical proximity.
Moreover, it is reasonable to assume that the sensor-based network is free of the directional bias which we detect in the diary-based network.
In Table~\ref{tab:rerank_summary}, we report the Spearman correlations between the sensor-based degree ranking and three diary-based rankings: the unadjusted ranking (corresponding to WYSIWYG), the proportionally representative re-ranking (corresponding to WAE), and our proposed algorithmic re-ranking. 
The re-ranking procedures maintain a similar level of concordance with the sensor-based ranking (Spearman $\approx$ 0.5), indicating that fairness improvements do not come at the cost of overall accuracy. At the same time, both the proportional and plug-in corrections substantially reduce group-specific ranking bias, yielding more balanced gender representation among top-ranked students. The unadjusted ranking exhibits a pronounced underrepresentation of female students, whereas our re-ranking algorithm substantially reduces this bias. Moreover, it slightly outperforms the proportional adjustment on the same metrics. Overall, the proposed re-ranking algorithm achieves a significant gain in fairness while preserving the fidelity of the original ranking structure.

Taken together, these results demonstrate that our methods effectively capture systematic biases, yield interpretable parameter estimates consistent with prior findings, and introduce a re-ranking algorithm that enhances fairness, thereby offering a stronger foundation for subsequent analyses.

\begin{figure}[tbp]
\centering
\begin{minipage}{0.485\linewidth}
    \centering
    \includegraphics[width=\linewidth]{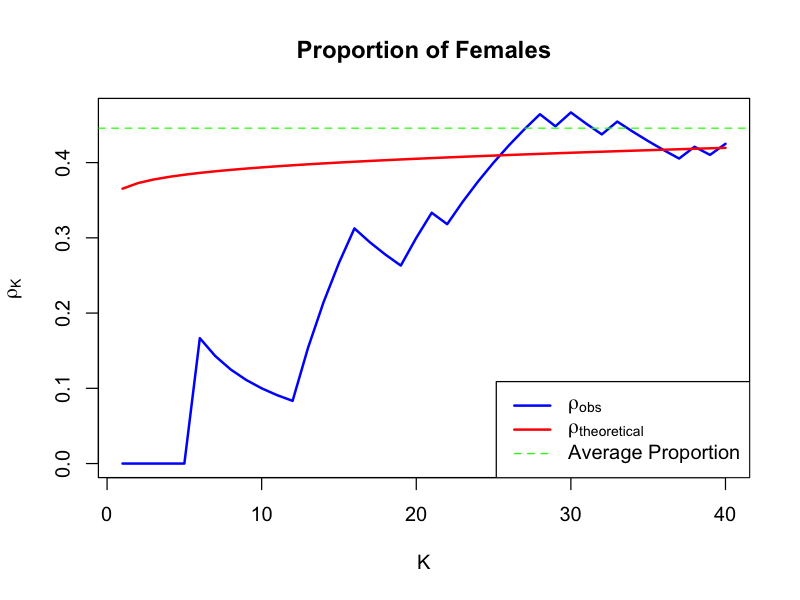}
    \captionsetup{font=footnotesize}
    \caption{Proportion of females among top-$K$ nodes. Blue: observed $\rho_K$, red: parametric plug-in estimate of $\rho_K$, green: overall female proportion.}
    \label{fig:proportion_females_rho}
\end{minipage}
\hspace{0.01\linewidth}
\begin{minipage}{0.475\linewidth}
    \scriptsize
    \centering
    \begin{tabular}{rllrl}
        \toprule
        \textbf{Rank} & \textbf{Orig ID} & \textbf{Gen} & \textbf{Corr ID} & \textbf{Gen} \\
        \midrule
        1 & 1295 & M & 1295 & M \\ 
        2 & 1423 & M & 441 & F \\ 
        3 &  21 & M & 1423 & M \\ 
        4 &  28 & M & 295 & F \\ 
        5 & 200 & M & 1214 & M \\ 
        6 & 441 & F & 200 & M \\ 
        7 & 488 & M & 232 & F \\ 
        8 & 826 & M &  28 & M \\ 
        9 & 1214 & M & 388 & F \\ 
        10 & 1401 & M & 826 & M \\ 
        \bottomrule
    \end{tabular}
    \captionsetup{font=footnotesize}
    \captionof{table}{Top-10 students by in-degree before and after bias correction. ``Gen" denotes gender.}
    \label{tab:top10_comparison}
\end{minipage}
\end{figure}

\begin{table}[t!]
\centering
\caption{Comparison of re-ranking methods based on 100 random tie-breaking replicates.
Reported values are averages across replicates.
Minority and majority biases are mean rank differences relative to the sensor-based ranking
(positive = ranked too high, negative = ranked too low).}
\label{tab:rerank_summary}
\begin{tabular}{l|c|cc}
\toprule
Method & Spearman vs.\ Sensor rank & Minority Bias & Majority Bias \\
\midrule
Unadjusted & 0.49 & $-3.46$ & $2.69$ \\
Proportion & 0.50 & $-0.26$ & $0.20$ \\
Plug-in    & 0.50 & $-0.17$ & $0.13$ \\
\bottomrule
\end{tabular}
\end{table}

\section{Discussion and conclusions}\label{sec:conclusion}

Motivated by contact diary data, and the broader challenge of fairness in network-based ranking across different social groups,  this paper introduces a mathematically rigorous framework to define and analyze systematic bias in network degree-based rankings. 
We define systematic bias as the disparity between construct and observed spaces, and introduce the minority representation profile as a tractable measure to evaluate bias and proportional representation among top-$K$ ranked nodes. 
To model systematic bias, we employ a group-dependent observation error model. 
Due to non-identifiability, estimating bias parameters requires either two independent replicates or a directed network. We propose methods to detect and correct for systematic bias when it is significant, leveraging the derived asymptotic limits to adjust minority group representations and refine ranking results. Our approach provides a formal statistical foundation for identifying and mitigating bias in network rankings, contributing to broader efforts in algorithmic fairness and equitable representation in labelled networks.

\ifjrss
\else
To validate our theoretical results, we conducted simulations that demonstrate the utility of our derived asymptotic limits, particularly for small node sizes ($n$) and small values of $K$ in selection. 
For detection, our method demonstrated great power, and nominal type I error control across different network sizes. 
The ranking correction, implemented using the derived theoretical plug-in estimator,  outperformed both the uncorrected rankings and those based solely on proportional representation in all cases except when no bias is present, achieving improved fairness and enhanced minority representation. 

We applied our methods to contact diary data, and detected systematic bias between female and male nodes, which aligns with the recall discriminability phenomenon in egocentric networks.
Our correction approach provided a more balanced degree-based ranking, without a loss of accuracy relative to an uncorrected ranking, or the need for strong identifying assumptions (WYSIWYG or WAE).
\fi

There are several limitations and potential extensions of our work. First, our analysis is restricted to degree centrality. An extension to eigenvector centrality could be pursued using modern spectral theory \citep{chen2021spectral}. However, unlike the degree case, such an extension would not yield a sharp asymptotic limit, since spectral theory provides only upper bounds on the eigenvector perturbation introduced by systematic bias.

Second, while our proposed label-dependent missing edge model falls within the class of systematic bias models, it would be valuable to explore alternative modeling approaches for systematic bias, as well as to develop new statistics for measuring fairness in network data.
For instance, \cite{braithwaite2020automated} identified underreporting by highly connected individuals in contact networks during the COVID-19 pandemic, suggesting a model in which observation error rates are related to node degree in the construct network. 

Third, particular applications motivate a hierarchical SBM (HSBM) formulation, in which edge error rates still depend on only protected group membership, but the construct model edge probabilites depend on additional categorical node covariates. That is,
\begin{align*}
    \prob(Y_{ij}=0 ~\vert~ A_{ij}=1, c_i,c_j,\bm{x}_i,\bm{x}_j) &= \beta_{c_ic_j}, \\
    \prob(A_{ij}=1 ~\vert~ c_i,c_j,\bm{x}_i,\bm{x}_j) &= B(c_i,c_j,\bm{x}_i,\bm{x}_j) \in [0,1]
\end{align*}
for $i < j$, where $\{ \bm{x}_i \}_{i=1}^n$ are additional (non-protected) categorical node covariates.
As this HSBM is a special case of the labeled graphon model, our estimation and testing framework remains applicable. In principle, this model should better capture the structure of our contact diary data, using students’ class membership as the second covariate. 
However, for our particular dataset, sparse cross-class connections led to high estimator variance and limited the reliability of testing. Hence we do not present these results here. 

Lastly, if the labels contain more than two groups ($k > 2$), designing tractable statistics to measure the fairness of the observed network and extending current detection and correction methods will significantly advance the statistical analysis of fairness in networks.

\noindent {\bf Acknowledgements} \\
Eric D. Kolaczyk and Hui Shen were supported in part by the Natural Sciences and Engineering Research Council of Canada (NSERC) through grants RGPIN-2023-03566 and DGDND-2023-03566. Peter W. MacDonald was supported in part by NSERC through a Postdoctoral Fellowship, and grant RGPIN-2025-02892.

\bibliography{bibfile}

\appendix

\section{Technical proofs}

\subsection{Supporting lemma}
We first introduce and prove a lemma that will be useful for proving the main theorems in the manuscipt. 
Consider the general setting where $\frac{n_1}{n} \rightarrow \kappa \in (0,1) $ and the scaled degrees have non-degenerate limiting distributions:
$$
\phi_n(d_i) \convd F_{c_i},
$$
for some monotone scaling function $\phi_n(\cdot)$ and $F_1(x)$ and $F_2(x)$ denote the limiting distributions for nodes in group 1 and group 2, respectively. The functions $\phi_n$ form a sequence of increasing transformations. With slight abuse of notation, we use $F_1$ and $F_2$ interchangeably to denote both the random variables and their distributions. Then an arbitrary node's degree has a mixture limiting distribution with
    $$
    F^*(x) := \kappa F_1(x) + (1-\kappa) F_2(x). 
    $$

We establish the following Glivenko–Cantelli-type result for the empirical distribution of normalized degrees.

    \begin{lemma}\label{lem:Glivenko-Cantelli}
    Under the graphon model and $K/n \rightarrow z$, we have: 
    $$
    \operatorname{sup}_{x \in \mathbb{R}} \left\lvert \frac{1}{n} \sum_{i=1}^n \mathbb{I}\left( \frac{d_i - np}{\sqrt{np(1-p)}} \leq x \right) - F^*(x) \right\rvert \convp 0,
    $$
    and 
    $$
    \phi_n(d_{(K)}) \convp c^*. 
    $$
    where $d_{(K)}$ is the $K$-th largest degree and $c^* = (F^{*})^{-1}(1-z)$ is the $(1 - z)$-th quantile of the mixture distribution $F^*$.  
    \end{lemma}

\begin{proof}
    The proof consists of two steps:
    \begin{enumerate}
         \item
        Assmuing        \begin{equation}\label{eqn:GC_result}
            \operatorname{sup}_{x \in \mathbb{R}}  \left\lvert \frac{1}{n} \sum_{i=1}^n \mathbb{I}\left( \frac{d_i - np}{\sqrt{np(1-p)}} \leq x \right) - F^*(x) \right\rvert \convp 0,
        \end{equation}
show that $\phi_n(d_{(K_n)}) \convp c^*$.
        \item
        Show that $$
	\operatorname{sup}_{x \in \mathbb{R}}  \left\lvert \frac{1}{n} \sum_{i=1}^n \mathbb{I}\left( \frac{d_i - np}{\sqrt{np(1-p)}} \leq x \right) - F^*(x) \right\rvert \convp 0.
    $$
    \end{enumerate}

{\bf Proof of Step 1}.  
We follow the proof strategy of Lemma 21.2 in \cite{van2000asymptotic}. 

Define 
$$F_n(x) = \frac{1}{n} \sum_{i=1}^n \mathbb{I}\left( \frac{d_i - np}{\sqrt{np(1-p)}} \leq x \right).$$ 
Let $Z$ be a standard normally distributed random variable. Assuming~\eqref{eqn:GC_result}, we have
$$
F_n(Z) \xrightarrow{p} F^*(Z),
$$
since
$$
F_n(Z) - F(Z) = \bbE\left( F_n(Z) - F^*(Z) \mid Z \right) 
\leq \bbE\left( \sup_{x \in \bbR} \left|F_n(x) - F^*(x)\right| \mid Z \right) 
\xrightarrow{p} 0.
$$

Thus,
$$
\Phi\left(F_n^{-1}(t)\right) = \bbP_Z\left(F_n(Z) < t\right)
\xrightarrow{p} \bbP_Z\left(F^*(Z) < t\right)
= \Phi\left((F^*)^{-1}(t)\right),
$$
for every $t$, since both $F^*$ and $\Phi$ are continuous. By the continuity of $\Phi^{-1}$, it follows from the continuous mapping theorem that
$$
F_n^{-1}(t) \xrightarrow{p} (F^*)^{-1}(t).
$$

In particular, for quantiles where $K_n/n \to q$, and using the continuity of $F$, we have
$$
F_n^{-1}\left(\frac{K_n}{n}\right) \xrightarrow{p} (F^*)^{-1}(q).
$$

{\bf Proof of Step 2}.

By the proof of basic Glivenko-Cantelli Theorem \citep{kosorok2008introduction},
it should be sufficient to show pointwise convergence for any $x \in \mathbb{R}$.

Following the proof of the the basic Glivenko-Cantelli Theorem, let $0<\delta<1$ be arbitrary and take $x_0<x_1<\cdots<x_N$ such that $F^*\left(x_j\right)-F^*\left(x_{j-1}\right)=\delta$. Note that $N(\delta)\leq 1+1 / \delta$. If $x \in\left(x_{j-1}, x_j\right]$, then clearly
$$
\left(-\infty, x_{j-1}\right] \subset(-\infty, x] \subset\left(-\infty, x_j\right]. 
$$

Now, fix $x \in \mathbb{R}$. Define $X_i := \mathbb{I} \left( \dfrac{d_i - np}{\sqrt{np(1-p)}} \leq x \right)$,
and let
$$F_n(x) = \frac{1}{n} \sum_{i=1}^n X_i, \quad F^{(n)}_{\text{mix}}(x) := \mathbb{E}[F_n(x)]. $$

We bound:
\begin{equation}\label{eqn:2251}
    \mathbb{P} \left( |F_n(x) - F^*(x)| \geq \epsilon \right)
    \leq \mathbb{P} \left( |F_n(x) - F^{(n)}_{\text{mix}}(x)| \geq \epsilon/2 \right) + \mathbb{P} \left( |F^{(n)}_{\text{mix}}(x) - F^*(x)| \geq \epsilon/2 \right).
\end{equation}

We now control both terms on the right-hand side.

For the first term in \eqref{eqn:2251}, by Chebyshev’s inequality:
$$\mathbb{P} \left( |F_n(x) - F^{(n)}_{\text{mix}}(x)| \geq \epsilon/2 \right)
\leq \frac{4}{\epsilon^2} \operatorname{Var}(F_n(x)).$$

We control the variance via:
$$\operatorname{Var}(F_n(x)) = \frac{1}{n^2} \sum_{i=1}^n \operatorname{Var}(X_i) + \frac{1}{n^2} \sum_{i \neq j} \operatorname{Cov}(X_i, X_j).$$

Since $X_i \in \{0,1\}$, we have $\operatorname{Var}(X_i) \leq \frac{1}{4}$, and so:
$$\frac{1}{n^2} \sum_{i=1}^n \operatorname{Var}(X_i) = O(n^{-1}).$$

Now consider the covariance term. For $i \neq j$,
\begin{align*}
\operatorname{Cov}(X_i, X_j)
&= \mathbb{E}[X_i X_j] - \mathbb{E}[X_i] \mathbb{E}[X_j] \\
&= \mathbb{P}(X_i = 1, X_j = 1) - \mathbb{P}(X_i = 1)\mathbb{P}(X_j = 1).
\end{align*}

Write them as:
\begin{align*}
    & \mathbb{P}(X_i = 1, X_j = 1) = \mathbb{P}\left( \frac{d_i - np}{\sqrt{np(1-p)}} \leq x, \frac{d_j - np}{\sqrt{np(1-p)}} \leq x \right) = \mathbb{P}(d_i \leq \tau, d_j \leq \tau), \\    
    & \mathbb{P}(X_i = 1)\mathbb{P}(X_j = 1) = \mathbb{P}(d_i \leq \tau) \mathbb{P}(d_j \leq \tau),
\end{align*}
where $\tau := np + \sqrt{np(1-p)}x$. 

To isolate dependence, we condition on the edge $A_{ij}$:
$$
\mathbb{P}(d_i, d_j \leq \tau)
= \mathbb{P}(d_i, d_j \leq \tau \mid A_{ij} = 1)\mathbb{P}(A_{ij} = 1) + \mathbb{P}(d_i, d_j \leq \tau \mid A_{ij} = 0)\mathbb{P}(A_{ij} = 0).
$$

Recall that $d_i = d_i^{(-j)} + A_{ij}$, so:
\begin{align*}
\mathbb{P}(d_i, d_j \leq \tau \mid A_{ij} = 1)
&= \mathbb{P}(d_i^{(-j)} \leq \tau - 1, d_j^{(-i)} \leq \tau - 1), \\
\mathbb{P}(d_i, d_j \leq \tau \mid A_{ij} = 0)
&= \mathbb{P}(d_i^{(-j)} \leq \tau, d_j^{(-i)} \leq \tau).
\end{align*}

We now compute the conditional expectation and variance of degrees under the graphon model. For any node $i$,
\begin{align*}
\mathbb{E}(d_i \mid U_i) &= \mathbb{E}\left[ \sum_{j=1}^n \omega(U_i, U_j) \right] = n \int_0^1 \omega(U_i, v) dv \\
&= np + n \int_0^1 \left\{ \omega(U_i, v) - p \right\} dv \\
&= np + \sqrt{np(1-p)} \left\{ \mu(U_i) + o(1) \right\} \\
&= np + \mu(U_i) \sqrt{np(1-p)} + o(\sqrt{n}),
\end{align*}
where the second last step uses the scaling assumption in Assumption~2 of the manuscript.

Similarly, the conditional variance is:
\begin{align*}
    \operatorname{Var}(d_i \vert U_i) &= \mathbb{E}\left( \sum_j \omega(U_i,U_j)\{ 1 - \omega(U_i,U_j) \} \right) + \operatorname{Var} \left\{ \sum_j \omega(U_i,U_j) \right\} \\
    &= n \left\{ \int_0^1 \omega(U_i,v)dv -  \int_0^1 \omega^2(U_i,v)dv  \right\} + n  \int_0^1 \omega^2(U_i,v)dv - n \left( \int_0^1 \omega(U_i,v)dv \right)^2 \\
    &= n \left(  \int_0^1 \omega(U_i,v)dv \right) \left( 1 -  \int_0^1 \omega(U_i,v)dv \right).
\end{align*}
Then $ \int_0^1 \omega(U_i,v)dv = p + O(1/\sqrt{n})$ since $\omega(U_i, v) = p + O(n^{-1/2})$ uniformly in $v$ by Assumption~1 of the manuscript.

As a result, for any node $i\in [n]$, we have 
\begin{align*}
    \mathbb{E}(d_i \vert U_i) = np + \mu(U_i)\sqrt{np(1-p)} + o(\sqrt{n}), \quad
    \operatorname{Var}(d_i \vert U_i) = np(1-p) + O(\sqrt{n}).
\end{align*}

Applying the Berry–Esseen theorem conditionally on $U_i$, we obtain:
$$\mathbb{P}(d_i \leq \tau \mid U_i) = \Phi\left( \frac{\tau - \mathbb{E}(d_i \mid U_i)}{\sqrt{\operatorname{Var}(d_i \mid U_i)}} \right) + o(1).$$
Now marginalizing over $U_i \sim \text{Uniform}[0,1]$, we obtain:
$$\mathbb{P}(d_i \leq \tau) = \int_0^1 \Phi\left( \frac{\tau - \mathbb{E}(d_i \mid u)}{\sqrt{\operatorname{Var}(d_i \mid u)}} \right) du + o(1).$$
Using the expansion:
$$    
\mathbb{E}(d_i \vert U_i) = np + \mu(U_i)\sqrt{np(1-p)} +     o(\sqrt{n}), \quad 
\operatorname{Var}(d_i \vert U_i) = np(1-p) + O(\sqrt{n}),$$
we have
\begin{align*}
\frac{\tau - \mathbb{E}(d_i \mid u)}{\sqrt{\operatorname{Var}(d_i \mid u)}}
&= \frac{np + x \sqrt{np(1 - p)} - \left[ np + \mu(u) \sqrt{np(1 - p)} + o(\sqrt{n}) \right]}{\sqrt{np(1 - p)} + O(1)} \\
&= \frac{(x - \mu(u)) \sqrt{np(1 - p)} + o(\sqrt{n})}{\sqrt{np(1 - p)} + O(1)} \\
&= x - \mu(u) + o(1).
\end{align*}

Therefore:
$$\mathbb{P}(d_i \leq \tau) = \int_0^1 \Phi(x - \mu(u)) du + o(1).$$
Similarly, for the leave-one-out degrees $d_i^{(-j)} := d_i - A_{ij}$, we apply the same Berry–Esseen argument. Since removing a single Bernoulli term changes the degree by at most one, we have
$$\mathbb{P}(d_i^{(-j)} \leq \tau - 1 \mid U_i) = \Phi\left(x - \mu(U_i)\right) + o(1), \quad
\mathbb{P}(d_i^{(-j)} \leq \tau \mid U_i) = \Phi\left(x - \mu(U_i)\right) + o(1).$$
Marginalizing over $U_i \sim \operatorname{Uniform}[0,1]$, it follows that
$$\mathbb{P}(d_i^{(-j)} \leq \tau - 1) = \int_0^1 \Phi(x - \mu(u)) \, du + o(1), \quad
\mathbb{P}(d_i^{(-j)} \leq \tau) = \int_0^1 \Phi(x - \mu(u)) \, du + o(1).$$

Therefore, 
$$
\operatorname{Cov}(X_i, X_j) = o(1).
$$
Finally, summing over all $i \neq j$:
$$\frac{1}{n^2} \sum_{i \neq j} \operatorname{Cov}(X_i, X_j) = o(1),$$
which establishes the overall variance bound:
$$\operatorname{Var}\left( \frac{1}{n} \sum_{i=1}^n X_i \right) = O(n^{-1}) + o(1) = o(1).$$

For the second term in \eqref{eqn:2251}, note that 
$$ \bbP\left( \left| F^{(n)}_{\operatorname{mix}}(x) - F^*(x) \right| \geq  \epsilon/2 \right) = 0$$ 
for large $n$ since $F^{(n)}_{\operatorname{mix}}(x) - F^*(x) = o\left( 1\right)$ by the central limit theorem.

Combining both components, we conclude that
\begin{align*}
    \bbP\left( \left| \frac{1}{n} \sum_{i=1}^n \mathbb{I}\left( \frac{d_i - np}{\sqrt{np(1-p)}} \leq x \right) - F^*(x) \right| \geq \epsilon \right) = o(1) .
\end{align*}
It follows that
$$
\sup _{x \in \mathbb{R}}\left|F_n(x)-F^*(x)\right| \leq \max _{j=1, \ldots, N}\left|F_n\left(x_j\right)-F^*\left(x_j\right)\right|+\delta, 
$$
and
\begin{align*}
     \bbP\left( \sup _{x \in \mathbb{R}}\left|F_n(x)-F^*(x)\right| \geq 2\delta \right) \leq \bbP\left( \max_{j=1, \ldots, N}\left|F_n\left(x_j\right)-F^*\left(x_j\right)\right| \geq \delta \right) \leq No(1). 
\end{align*}
In conclusion,
$$
\sup _{x \in \mathbb{R}}\left|F_n(x)-F^*(x)\right| \convp 0.
$$
This completes the entire proof. 

\end{proof}

\subsection{Technical proofs for Section~3 of the manuscript}

\subsubsection{Proof of Theorem~1}
\begin{proof}

Under the graphon model, to generate a graph with minority and majority group memberships, we assign $c_i = 1$ with probability $\kappa < 1/2$ and $c_i=2$ with probability $1-\kappa$, independently.
For $i=1,\ldots,n$,
\begin{align}
    U_i \vert c_i=1 &\sim \operatorname{Uniform}[0,\kappa],\label{eqn:graphon_U1} \\
    U_i \vert c_i=2 &\sim \operatorname{Uniform}[\kappa,1]\label{eqn:graphon_U2}
\end{align}
independently. 
For $i < j$,
$$
    A_{ij} \sim \operatorname{Bernoulli}\left\{ \omega_n(U_i,U_j) \right\}
$$
independently.
Marginally, $U_i \sim  \operatorname{Uniform}[0,1]$. 
For simplicity of the proof, we omit the subscript $n$ on $\omega$.

We can show that, conditional on $U_i$ (and $c_i$), i.e.~taking probability over the other $U$'s and the edges,
$$
    \frac{d_i - np}{\sqrt{np(1-p)}} \convd \mathcal{N}(\mu(U_i),1).
$$
It follows that the marginal of $d_i$ is asymptotically a (continuous) mixture over $U_i$, as is the conditional distribution on only group membership. The proof proceeds as follows: 

\begin{enumerate}
    \item By Lyapunov CLT and since
    $$
        \frac{\sum_{j \neq i} \omega(U_i,U_j) \{ 1 - \omega(U_i,U_j) \}}{n p (1-p)} \inprob 1
    $$
    we have
    $$
        \frac{d_i - \sum_{j \neq i} \omega(U_i,U_j)}{\sqrt{np(1-p)}} = \frac{\sum_{j \neq i} \{A_{ij} - \omega(U_i,U_j)\}}{\sqrt{np(1-p)}} \convd \mathcal{N}(0,1)
    $$
    conditional on all the $U$'s.
    \item By Chebyshev inequality,
    $$
        \frac{\sum_{j \neq i} \omega(U_i,U_j) - n \int_0^1 \omega(U_i,v)dv}{\sqrt{np(1-p)}} \inprob 0,
    $$
    conditional on $U_i$, since we can show that the variance of the numerator is asymptotically bounded.
    \item By assumption
    $$
        \frac{n \int_0^1 \omega(U_i,v)dv - np}{\sqrt{np(1-p)}} \rightarrow \mu(U_i).
    $$
    \item Thus,
    $$
    \frac{d_i - np}{\sqrt{np(1-p)}} \convd \mathcal{N}(\mu(U_i),1) 
    $$
    conditional on $U_i$.
\end{enumerate}

Then unconditional on $U_i$, we have for $x\in \bbR $
\begin{align*}
    & \bbP\left( \frac{d_i - np}{\sqrt{np(1-p)}} \leq x \right) 
    = \bbE_{U_i}\left[ \bbP\left( \frac{d_i - np}{\sqrt{np(1-p)}} \leq x | U_i \right) \right] \\
    \rightarrow & \bbE_{U_i}\left[
    \bbP\left( \mathcal{N}(\mu(U_i),1) \leq x|U_i \right)
    \right] \\
    =& \bbE_{U_i}\left[
    \bbP\left( \mathcal{N}(0,1) + \mu(U_i) \leq x|U_i \right)
    \right] \\
    = & 
    \bbP\left( \mathcal{N}(0,1) + \mu(U_i) \leq x \right)
\end{align*}
where the convergence is guaranteed by the dominanted convergence theorem. 


Applying Lemma~\ref{lem:Glivenko-Cantelli} with $F_1 = \mathcal{N}\left(0, 1\right) + \mu(U^{(1)})$ and $F_2 = \mathcal{N}\left(0, 1\right) + \mu(U^{(2)})$  with $U^{(1)}\sim \operatorname{Uniform}[0,\kappa]$, $U^{(2)}\sim \operatorname{Uniform}[\kappa,1]$ gives the final result.

\end{proof}

\subsubsection{Proof of Proposition~1}
\begin{proof}
We begin by proving the first conclusion.

    {\bf Proof of Conclusion (1)}. 
    
    Assume $p_i-q \ll \frac{1}{\sqrt{n}}$ for $i\in [2]$. The proof consists of two steps:
    \begin{enumerate}
        \item 
        Show that for any node $i$, regardless of its group, the normalized degree satisfies:
        $$
        \frac{d_i - nq}{\sqrt{nq(1-q)}} \convd \mathcal{N}(0,1). 
        $$
        \item 
        Apply Lemma \ref{lem:Glivenko-Cantelli} to conclude $R_K(\bm{c},A) \convp \kappa$. 
    \end{enumerate}
    Denote $\delta_i = p_i - q$ for $i\in [2]$.  For an arbitrary node $i\in [n]$, assume $c_i = 1$. Then 
    $$
    d_i = \sum_{j\in \cC_1, j\neq i} d_{ij} + \sum_{j\in \cC_2} d_{ij},
    $$
    whose expectation and variance are 
    \begin{align*}
    \bbE\left(d_i\right) & = (n_1-1) p_1 + n_2q = (n_1-1) (q+\delta_1) + n_2 q = (n-1)q + (n_1 - 1)\delta_1 \\     
    \operatorname{Var}\left(d_i\right) & = (n_1-1) p_1 (1-p_1)+ n_2q(1-q) = (n_1-1) (q+\delta_1) (1-(q+\delta_1))+ n_2q(1-q) \\
    & = (n-1)q(1-q) + (n_1 - 1)\left[ \delta_1 (1-2q) - \delta_1^2 \right]  
    \end{align*}
    The central limit theorem gives 
    \begin{equation}\label{eqn:di_CLT}
        \frac{d_i-\bbE\left(d_i\right)}{\sqrt{\operatorname{Var}\left(d_i\right)}} \convd \cN(0,1). 
    \end{equation}
    and 
    $$
    n_1 = \kappa n + O_{\prob}\left( \sqrt{n} \right), \quad n_2 = (1 - \kappa) n + O_{\prob}\left( \sqrt{n} \right). 
    $$
    Now consider 
    \begin{align*}
    \frac{d_i - nq}{\sqrt{nq(1-q)}} 
    = \left(\frac{d_i-\bbE\left(d_i\right)}{\sqrt{\operatorname{Var}\left(d_i\right)}} + \frac{(n_1-1)\delta_1 - q}{\sqrt{\operatorname{Var}(d_i)}} \right)\frac{\sqrt{\operatorname{Var}(d_i)}}{\sqrt{nq(1-q)}} .
    \end{align*}
    
    Under the assumption $\delta_1 = p_1 -q =o\left( \frac{1}{\sqrt{n}}\right)$, we have 
    $$
    (n_1-1)\delta_1 - q = o_{\prob}(\sqrt{n}), \quad \operatorname{Var}(d_i) = nq(1-q) + o_{\prob}(\sqrt{n}), 
    $$
    and thus 
    $$
    \frac{(n_1-1)\delta_1 - q}{\sqrt{\operatorname{Var}(d_i)}} = o_{\prob}(1), \quad \frac{\sqrt{\operatorname{Var}(d_i)}}{\sqrt{nq(1-q)}} \convp 1. 
    $$
    By Slutsky’s theorem, conditional on $c_i=1$, 
    $$
    \frac{d_i - nq}{\sqrt{nq(1-q)}} \convd \cN(0,1). 
    $$
    A similar argument applies when $c_i = 2$, where
    $$
    d_i = \sum_{j\in \cC_1} d_{ij} + \sum_{j\in \cC_2, j\neq i} d_{ij},
    $$
    whose expectation and variance are 
    \begin{align*}
    \bbE\left(d_i\right) & = n_1 q + (n_2-1)p_2 = n_1 q  + (n_2-1)(q + \delta_2 )= (n-1)q + (n_2-1)\delta_2 \\    
    \operatorname{Var}\left(d_i\right) & = n_1 q (1-q)+ (n_2-1)p_2(1-p_2) = n_1 q(1 - q) + (n_2 - 1)(q + \delta_2)(1 - q - \delta_2) \\
    &= (n - 1) q(1 - q) + (n_2 - 1) \left[ \delta_2 (1 - 2q) - \delta_2^2 \right].
    \end{align*}
    Therefore, regardless of group membership, we have
    $$
    \frac{d_i - np}{\sqrt{np(1 - p)}} \xrightarrow{d} \mathcal{N}(0,1).
    $$
    To apply Lemma~\ref{lem:Glivenko-Cantelli}, define 
    $$
    \phi_{n}(x) = \frac{x-nq}{\sqrt{nq(1-q)}} , \quad F_1 = F_2 = F^* = \Phi. 
    $$
    where $\Phi$ is the CDF of the standard normal distribution. 
    Hence, 
    $$
    \phi_n(d_{(K)}) \convp c^* = (\Phi)^{-1}(1-z).
    $$
    Recall the proportion of group 1 (minority) nodes as
    $$
    R_K(\bm{c},A) = \frac{1}{K} \left\lvert \left\{ i \in D_K(A) : c_i = 1 \right\} \right\rvert = \frac{1}{K}\sum_{i=1}^n \mathbb{I}(c_i = 1 , d_i \geq d_{(K)})
    $$
    and we define a population version of the proportion of group 1 (minority) nodes as 
    $$
    R_K^*(\bm{c},A) = \frac{1}{K}\sum_{i=1}^n \mathbb{I}(c_i = 1 , \phi(d_i) \geq c^*).
    $$    
    To show the convergence of this quantity, we apply Chebyshev’s inequality and follow a similar argument to the proof of Lemma~\ref{lem:Glivenko-Cantelli}. Specifically, for any $\epsilon > 0$,
\begin{align*}
    & \bbP\left( \left| R_K^*(\bm{c},A) -\bbE\left(  R_K^*(\bm{c},A) \right) \right| > \epsilon \right) \\
    = & \bbP\left( \left| \frac{1}{K}\sum_{i=1}^n \mathbb{I}(c_i = 1 , \phi(d_i) \geq c^*) - \bbE\left( \frac{1}{K}\sum_{i=1}^n \mathbb{I}(c_i = 1 , \phi(d_i)\geq c^*) \right) \right| > \epsilon \right) \\
    \leq & \frac{1}{\epsilon^2}\operatorname{Var}\left( \frac{1}{K}\sum_{i=1}^n \mathbb{I}(c_i = 1 , \phi(d_i) \geq c^*) \right)\\
    = & \frac{1}{\epsilon^2}\frac{1}{K^2} \left[  \sum_{i=1}^n \operatorname{Var}\left( \mathbb{I}(c_i = 1 , \phi(d_i) \geq c^*)  \right) + \sum_{i\neq j} \cov\left(\mathbb{I}(c_i = 1 , \phi(d_i) \geq c^*), \mathbb{I}(c_j = 1 , \phi(d_j) \geq c^*)\right) \right] \\
    = & O(n^{-1/4}).
\end{align*}
Moreover, since
    $$
    \bbE\left( \mathbb{I} \left( c_i = 1 , \phi(d_i) \geq c^*\right) \right)
    = \bbP\left( 
    \phi(d_i) \geq c^*\right)\bbP\left( c_i = 1|\phi(d_i) \geq c^*\right) \rightarrow z\kappa 
    $$
   we conclude that
    \begin{equation} \label{eqn:R_star}
    R_K(\bm{c}, A) \xrightarrow{p} \kappa, \quad \text{as } \frac{n}{K} \to \frac{1}{z}.
    \end{equation}
Now consider bounding the difference between $R_K$ and $R_K^*$:
    \begin{align*}
    & \lvert R_K(\bm{c},A) - R_K^*(\bm{c},A) 
    \rvert \\
    \leq & \frac{1}{K}\sum_{i=1}^n \left\{\mathbb{I}\left(\phi_n(d_{(K)}) \leq \phi_n(d_i) < c^*, c_i = 1\right) + \mathbb{I}\left(c^* \leq \phi_n(d_i) < \phi_n(d_{(K)}), c_i = 1\right) \right\}\\    
    \leq & \frac{1}{K}\sum_{i=1}^n \left\{\mathbb{I}\left(\phi_n(d_{(K)}) \leq \phi_n(d_i) < c^*\right) + \mathbb{I}\left(c^* \leq \phi_n(d_i) < \phi_n(d_{(K)})\right) \right\}\\
    \leq & \frac{1}{K}\sum_{i=1}^n \mathbb{I}(c^* - \epsilon_n \leq \phi_n(d_i) \leq c^* + \epsilon_n) + \frac{n}{K}\mathbb{I}(B_n^c)  
    \end{align*}
    where we define the event
    $$
    B_n = \{ c^* - \epsilon_n \leq \phi_n(d_{(K)}) \leq c^* + \epsilon_n \},
    $$
    for some sequence $\epsilon_n > 0 $. 
    By Lemma~\ref{lem:Glivenko-Cantelli}, we have
    $$
    \phi_n(d_{(K)}) \convp c^*. 
    $$
    Therefore, there exist $\epsilon_n \rightarrow 0 $ such that 
    $$
    \bbP(B_n) \to 1 \quad \text{as} \quad n \to \infty.
    $$
    We can bound the expectation as
    $$
    \bbE\lvert R_K(\bm{c},A) - R_K^*(\bm{c},A) 
    \rvert \leq \frac{n}{K}\bbP\left( c^* - \epsilon_n \leq \phi_n(d_1) \leq c^* + \epsilon_n \right) + \frac{n}{K}\bbP(B_n^c) \rightarrow 0, 
    $$
    because $\phi_n(d_1) \convd \cN(0,1)$, $\epsilon_n\rightarrow 0$ and $\frac{n}{K} \rightarrow \frac{1}{z} \in (1,\infty)$. 
    Since  $L^1$  convergence implies convergence in probability, we have 
    \begin{equation}\label{eqn:R_and_R_star}
    R_K(\bm{c},A) - R_K^*(\bm{c},A) \convp 0.         
    \end{equation}
    Combing \eqref{eqn:R_star} and \eqref{eqn:R_and_R_star}, we conclude  
    $$
    R_K(\bm{c},A) \convp \kappa,
    $$
    which completes the proof of Conclusion (1). 
   
    {\bf Proof of Conclusion (2)}. \\
    Using results derived above,
    \begin{align*}
    \bbE\left( d_i|i\in \cC_1 \right) - \bbE\left( d_i|i\in \cC_2 \right) &= [(n-1)q + (n_1 - 1)\delta_1] - [(n-1)q + (n_2-1)\delta_2] \\
    &=  (n_1 - 1)\delta_1 - (n_2-1)\delta_2 := \Delta_n (\delta). 
    \end{align*}
    Using Bernstein inequality and the union bound argument,
    \begin{align*}
        & \bbP\left( \max_{i\in \cC_1}|d_i - \bbE(d_i)| \geq \frac{1}{3} |\Delta_n (\delta)| \right) \\
        \leq & \sum_{i\in \cC_1} \bbP\left( |d_i - \bbE(d_i)| \geq \frac{1}{3} |\Delta_n (\delta) | \right) \\
        \leq & 2n_1 \exp\left( - \frac{\frac{1}{18}\Delta_n ^2(\delta)}{\operatorname{Var}(d_i) + \frac{1}{9}| \Delta_n (\delta)| } \right) \\
        \leq &  2n_1 \exp\left( - 
        \frac{\frac{1}{18}\Delta_n^2 (\delta) }{(n-1)q(1-q) + (n_1 - 1)\left[ \delta_1 (1-2q) - \delta_1^2 \right] + \frac{1}{9} |\Delta_n (\delta)| } \right) \\
        \leq & 2n_1  \exp\left( - c\log n \right)
        \leq 2 n^{-(c-1)}
    \end{align*}
    for some constant $c$ under the assumption that $|\Delta_n (\delta)| = |\kappa\delta_1 - (1-\kappa)\delta_2| \gtrsim \sqrt{\frac{\log n}{n}}$. 
   
    Similarly, we can show that 
    $$
    \bbP\left( \max_{i\in \cC_2}|d_i - \bbE(d_i)| \geq \frac{1}{3} |\Delta_n(\delta)| \right) \leq 2 n^{-(c-1)}. 
    $$
    Therefore, in the case where $\Delta_n(\delta)<0$, with probability $1 - 4n^{-(c-1)}$, 
    $$\max_{i\in \cC_1} d_i \leq [(n-1)q + (n_1 - 1)\delta_1] - \frac{1}{3}\Delta_n(\delta) := B_1, 
    $$
    and 
    $$
    \max_{i\in \cC_2} d_i \geq [(n-1)q + (n_2-1)\delta_2] + \frac{1}{3}\Delta_n(\delta) := B_2,
    $$
    where $B_1 < B_2$ under the condition $\Delta_n(\delta)<0$. 
    From the central limit theorem, 
    $$
    n_2 = (1-\kappa)n + O_{\prob}(\sqrt{n}). 
    $$
    Under the assumption that
    $K/n \rightarrow z$ with $z \leq 1-\kappa$, 
    we have 
    $$d_{(K)} >  [(n-1)q + (n_1 - 1)\delta_1] - \frac{1}{3}\Delta_n(\delta),$$
    and all top $K$ nodes are from group 2 with high probability. 
    Therefore, 
    $$
    R_K(\bm{c},A) \convp 0.
    $$    
    If $\Delta_n(\delta) > 0$, we can similarly show that
    $$
    R_K(\bm{c},A) \convp 1.
    $$
    This finishes the entire proof. 
\end{proof}

\subsubsection{Proof of Theorem~2}
\begin{proof}
    The SBM can be viewed as a special case of the graphon model. In this case, Theorem~2 follows as a corollary of Theorem~1, where the graphon $\omega_n$ is block-constant:
    \begin{equation*}
        \omega_n(u,v) = \begin{cases} q + \mu_1/\sqrt{n}, &\tabby u \in [0,\kappa], v \in [0,\kappa], \\
        q, &\tabby u \in [0,\kappa], v \in (\kappa,1], \\
        q, &\tabby u \in (\kappa,1], v \in [0,\kappa], \\
        q + \mu_2/\sqrt{n}, &\tabby u \in (\kappa,1], v \in (\kappa,1]. \end{cases}
    \end{equation*}
\end{proof}

\subsection{Technical proofs for Section 4.2 of the manuscript} \label{app:detection}

We first prove a supporting lemma about the joint asymptotic normality of the empirical moments.
\begin{lemma}
    Suppose $q \in (0,1)$ and $\beta \in (0,1)$. Then 
    $$
     n \begin{pmatrix}
        \hat{u}_{1,1} - p_1^* \\
        \hat{u}_{2,1} - \beta_{11} p_1^* \\ 
        \hat{u}_{1,2} - p_2^* \\
        \hat{u}_{2,2} - \beta_{22} p_2^* \\ 
        \hat{u}_{1,\mathrm{b}} - q^* \\
        \hat{u}_{2,\mathrm{b}} - \beta q^*
    \end{pmatrix} \convd \mathcal{N}\left( \bm{0}, \begin{pmatrix}
        a_{\kappa}^{-1} \Sigma_u & \bm{0} & \bm{0} \\ \bm{0} & b_{\kappa}^{-1} \Sigma_u & \bm{0} \\ 
        \bm{0} & \bm{0} & c_{\kappa}^{-1} \Sigma_u
    \end{pmatrix}\right),
    $$
    where
    \begin{align*}
        \beta_{11} &= \beta - \gamma_1/\sqrt{n}, \\
        \beta_{22} &= \beta - \gamma_2/\sqrt{n}, \\
        q^* &= (1-\beta) q, \\
        p^*_g &=  (1 - \beta + \gamma_g/\sqrt{n})(q + \mu_g/\sqrt{n}) ~ \text{for $g=1,2$},\\
        a_{\kappa} &= \frac{1}{2} \kappa^2, \\
        b_{\kappa} &= \frac{1}{2}(1-\kappa)^2, \\
        c_{\kappa} &= \kappa(1-\kappa), \\
        \Sigma_u &= \begin{pmatrix}
        q^*(1-q^*) & (1/2 - q^*)\beta q^* \\ (1/2 - q^*)\beta q^* & (1/2 -  \beta q^*)\beta q^*
    \end{pmatrix}.
    \end{align*}
    Moreover, since $q \in (0,1)$ and $\beta \in (0,1)$, $\Sigma_u$ is strictly positive definite.
    
    As a consequence, we also have
    \begin{equation*} \label{u_consistent}
\begin{pmatrix}
        \hat{u}_{1,1} &
        \hat{u}_{2,1} & 
        \hat{u}_{1,2} &
        \hat{u}_{2,2} & 
        \hat{u}_{1,\mathrm{b}} &
        \hat{u}_{2,\mathrm{b}}
    \end{pmatrix}^{\tp} \convp \begin{pmatrix}
        q^* &
        \beta q^* &
        q^* &
        \beta q^* &
        q^* &
        \beta q^*
\end{pmatrix}^{\tp}.
\end{equation*}
\end{lemma}

\begin{proof}
    For $g=1,2$, consider edges within group $g$. 
    Consider the $N_g$ independent random vectors
$$
    \left( Y_{ij}, \frac{1}{2}\lvert Y_{ij} - Y^*_{ij} \rvert \right)
$$
with $i < j$ and $c_i=c_j=g$.
\begin{align*}
    \mathbb{E}Y_{ij} &= p_g^* \\
    \mathbb{E}(\lvert Y_{ij} - Y^*_{ij} \rvert /2) &= \beta_{gg} p_g^* \\
    \sfvar(Y_{ij}) &= p_g^*(1-p_g^*) = q^*(1-q^*) + o(1) \\
    \sfvar(\lvert Y_{ij} - Y^*_{ij} \rvert /2) &= (1/2 -  \beta p_g^*)\beta_{gg} p_g^* = (1/2 - q^*)\beta q^* + o(1) \\
    \cov (Y_{ij},\lvert Y_{ij} - Y^*_{ij} \rvert /2) &= (1/2 - p_g^*)\beta_{gg} p_g^* = (1/2 -  \beta q^*)\beta q^* + o(1).
\end{align*}

since we can view $|Y_{ij} - Y^*_{ij}|\sim \operatorname{Bern}(2\beta_{gg} p_g^*)$.

For between-group edges we have the same replacing $\beta_{gg}$ by $\beta$ and $p^*$ by $q^*$.
Note that $\Sigma_u$ is positive definite, since $\sfvar(Y_{ij}) > 0$, $\sfvar(\lvert Y_{ij} - Y^*_{ij}\rvert) > 0$, and they are not linearly related.
Moreover, $\lambda_{\min}(\Sigma_u)$ will depend only on $p$ and $\beta$.

Since each random vector is bounded we have by Lindeberg-Feller CLT for triangular arrays,
$$
     \begin{pmatrix}
        \sqrt{N_1}(\hat{u}_{1,1} - p_1^*) \\
        \sqrt{N_1}(\hat{u}_{2,1} - \beta_{11} p_1^*) \\ 
        \sqrt{N_2}(\hat{u}_{1,2} - p_2^*) \\
        \sqrt{N_2}(\hat{u}_{2,2} - \beta_{22} p_2^*) \\ 
        \sqrt{N_{\mathrm{b}}}(\hat{u}_{1,\mathrm{b}} - q^*) \\
        \sqrt{N_{\mathrm{b}}}(\hat{u}_{2,\mathrm{b}} - \beta q^*) 
    \end{pmatrix} \convd \mathcal{N}\left( \bm{0}, \begin{pmatrix}
        \Sigma_u & \bm{0} & \bm{0} \\ \bm{0} & \Sigma_u & \bm{0} \\
        \bm{0} & \bm{0} & \Sigma_u
    \end{pmatrix}\right),
$$
where $p_g^*$ and $\beta_{gg}$ for $g=1,2$ depend on $n$.
Using the central limit theorem, 
$$
    \frac{1}{n^2}N_{1} = \frac{1}{n^2}\left\{\frac{n_1(n_1 - 1)}{2} \right\} \convp \frac{1}{2} \kappa^2 = a_{\kappa},
$$
$$
    \frac{1}{n^2}N_{2} = \frac{1}{n^2}\left\{\frac{n_2(n_2 - 1)}{2} \right\} \convp \frac{1}{2} (1-\kappa)^2 = b_{\kappa},
$$
and 
$$
     \frac{1}{n^2}N_{\mathrm{b}} = \frac{n_1n_2 }{n^2} \convp \kappa(1-\kappa) = c_{\kappa}
$$
Thus, by Slutsky's Theorem, we have
$$
     n \begin{pmatrix}
        \hat{u}_{1,1} - p_1^* \\
        \hat{u}_{2,1} - \beta_{11} p_1^* \\ 
        \hat{u}_{1,2} - p_2^* \\
        \hat{u}_{2,2} - \beta_{22} p_2^* \\ 
        \hat{u}_{1,\mathrm{b}} - q^* \\
        \hat{u}_{2,\mathrm{b}} - \beta q^*
    \end{pmatrix} \convd \mathcal{N}\left( \bm{0}, \begin{pmatrix}
        a_{\kappa}^{-1} \Sigma_u & \bm{0} & \bm{0} \\ \bm{0} & b_{\kappa}^{-1} \Sigma_u & \bm{0} \\ 
        \bm{0} & \bm{0} & c_{\kappa}^{-1} \Sigma_u
    \end{pmatrix}\right),
    $$
    as desired.
\end{proof}

We now use the delta method to prove Theorems~\ref{thm:pivot_beta} and \ref{thm:pivot}, asymptotic normality of our localized estimators

\begin{theorem} \label{thm:pivot_beta}
    Suppose $q \in (0,1)$, $\beta \in (0,1)$, and $Y \neq Y^*$.
    Then, using notation from Section~4.1 of the manuscript, 
    $$
        \begin{pmatrix}
        n(\hat{\beta}_{\mathrm{b}} - \beta) \\
        \sqrt{n}(\hat{\gamma}_1 - \gamma_1) \\
        \sqrt{n}(\hat{\gamma}_2 - \gamma_2)
    \end{pmatrix} \convd \mathcal{N}\left( \bm{0}, \Sigma_{\beta} \right)
    $$
    for a $3 \times 3$ positive definite covariance $\Sigma_{\beta}$.
\end{theorem}

\begin{proof}[Proof of Theorem~\ref{thm:pivot_beta}]
    Define
$$
    h(u,v,w,x,y,z) = \begin{pmatrix} \frac{z}{y} \\ \frac{v}{u} - \frac{z}{y} \\ \frac{x}{w} - \frac{z}{y} \end{pmatrix},
$$
then
$$
    \bm{D}h(u,v,w,x,y,z) = \begin{pmatrix}
        0 & 0 & 0 & 0 & - \frac{z}{y^2} & \frac{1}{y} \\
        -\frac{v}{u^2} & \frac{1}{u} & 0 & 0 & \frac{z}{y^2} & - \frac{1}{y} \\
         0 & 0 & - \frac{x}{w^2} & \frac{1}{w} & \frac{z}{y^2} & - \frac{1}{y} \\
    \end{pmatrix}
$$
Then by the delta method,
$$
    n \begin{pmatrix} \hat{\beta}_{\mathrm{b}} - \beta \\ \hat{\beta}_{\mathrm{b}} - \hat{\beta}_{g} - (\beta - \beta_{gg}) \end{pmatrix} \convd \mathcal{N}(\bm{0},\Sigma_{\beta}),
$$
where
$$
    \Sigma_{\beta} = \bm{D}h(\bm{\xi}^*) \begin{pmatrix}
        a_{\kappa}^{-1} \Sigma_u & \bm{0} & \bm{0} \\ \bm{0} & b_{\kappa}^{-1} \Sigma_u & \bm{0} \\ 
        \bm{0} & \bm{0} & c_{\kappa}^{-1} \Sigma_u
    \end{pmatrix} \bm{D}h(\bm{\xi}^*)^{\tp}
$$
and $\bm{\xi}^* = (q^*,\beta q^*, q^*, \beta q^*, q^*, \beta q^*)$.
Note that $\hat{\beta}_{\mathrm{b}} - \hat{\beta}_{g} - (\beta - \beta_{gg}) = n^{-1/2}(\hat{\gamma}_g - \gamma_g)$, which completes the proof.
\end{proof}

\begin{theorem} \label{thm:pivot}
    Suppose $q \in (0,1)$, $\beta \in (0,1)$, and $Y \neq Y^*$.
    Then, using notation from Section~4.1 of the manuscript,
    $$
        \sqrt{n} \begin{pmatrix}
        \hat{\mu}_1 - \mu_1 \\
        \hat{\mu}_2 - \mu_2 \\
        \hat{\gamma}_1 - \gamma_1 \\
        \hat{\gamma}_2 - \gamma_2 \\
    \end{pmatrix} \convd \mathcal{N}\left( \bm{0}, \Sigma_{\mu} \right)
    $$
    for a $4 \times 4$ positive definite covariance $\Sigma_{\mu}$.
\end{theorem}

\begin{proof}[Proof of Theorem~\ref{thm:pivot}]
    This proof proceeds similarly to the proof of Theorem~\ref{thm:pivot_beta}, for a smooth transformation
    $$
    g(u,v,w,x,y,z) = \begin{pmatrix}
        \frac{u^2}{u-v} - \frac{y^2}{y -z} \\
        \frac{w^2}{w-x} - \frac{y^2}{y -z} \\
        \frac{v}{u} - \frac{z}{y} \\
        \frac{x}{w} - \frac{z}{y}
    \end{pmatrix}.
    $$
    We have that
    $$
    g(\hat{u}_{1,1},\hat{u}_{2,1},\hat{u}_{1,2},\hat{u}_{2,2},\hat{u}_{1,\mathrm{b}},\hat{u}_{2,\mathrm{b}}) 
    = \frac{1}{\sqrt{n}} \begin{pmatrix}
        \hat{\mu}_1 \\
        \hat{\mu}_2 \\
        \hat{\gamma}_1 \\
        \hat{\gamma}_2
    \end{pmatrix}
    $$
    and
    $$
    g(p_1^*,\beta_{11} p_1^*,p_2^*,\beta_{22} p_2^*,q^*,\beta_{\mathrm{b}}q^*) = \frac{1}{\sqrt{n}} \begin{pmatrix}
        \mu_1 \\
        \mu_2 \\
        \gamma_1 \\
        \gamma_2
    \end{pmatrix}.
    $$
    So by delta method,
    $$
    \sqrt{n} \begin{pmatrix}
        \hat{\mu}_1 - \mu_1 \\
        \hat{\mu}_2 - \mu_2 \\
        \hat{\gamma}_1 - \gamma_1 \\
        \hat{\gamma}_2 - \gamma_2
    \end{pmatrix} \convd \mathcal{N}\left( \bm{0}, \Sigma_{\mu} \right)
$$
where
$$
   \Sigma_{\mu} =  \bm{D}g(\bm{\xi}^*) \begin{pmatrix}
        a_{\kappa}^{-1} \Sigma_u & \bm{0} & \bm{0} \\ \bm{0} & b_{\kappa}^{-1} \Sigma_u & \bm{0} \\ 
        \bm{0} & \bm{0} & c_{\kappa}^{-1} \Sigma_u
    \end{pmatrix} \bm{D}g(\bm{\xi}^*)^{\tp},
$$
and
$$
    \bm{D}g(x,y,z,w) = \begin{pmatrix}
        \frac{u(u-2v)}{(u-v)^2} & \frac{u^2}{(u-v)^2} & 0 & 0 & - \frac{y(y-2z)}{(y-z)^2} & - \frac{y^2}{(y-z)^2} \\
        0 & 0 & \frac{w(w-2x)}{(w-x)^2} & \frac{w^2}{(w-x)^2} & - \frac{y(y-2z)}{(y-z)^2} & - \frac{y^2}{(y-z)^2} \\
        -\frac{v}{u^2} & \frac{1}{u} & 0 & 0 & \frac{z}{y^2} & - \frac{1}{y} \\
         0 & 0 & - \frac{x}{w^2} & \frac{1}{w} & \frac{z}{y^2} & - \frac{1}{y} \\
    \end{pmatrix}.
$$
\end{proof}

We are now ready to state and prove Corollaries~\ref{cor:beta_test} and \ref{cor:mu_test}, which give the asymptotic distributions of the $\chi^2$ pivots used in Theorem~3 of the manuscript.

\begin{corollary} \label{cor:beta_test}
    Under the assumptions of Theorem~\ref{thm:pivot_beta}, suppose $\widehat{\Theta}_{\beta}$ is a consistent estimator of $(\Sigma_{\beta})^{-1}$. Then under $\gamma_1=\gamma_2=0$, we have
    $$
        Q_{n,\bar{\beta}}(\beta) := (n(\hat{\beta}_{\mathrm{b}} - \beta), \sqrt{n}\hat{\gamma}_1, \sqrt{n}\hat{\gamma}_2 )^{\tp} \widehat{\Theta}_{\beta} (n(\hat{\beta}_{\mathrm{b}} - \beta), \sqrt{n}\hat{\gamma}_1, \sqrt{n}\hat{\gamma}_2) \convd \chi^2_3.
    $$
\end{corollary}

\begin{corollary} \label{cor:mu_test}
    Under the assumptions of Theorem~\ref{thm:pivot}, suppose $\widehat{\Theta}_{\mu}$ is a consistent estimator of $(\Sigma_{\mu})^{-1}$. Then under $H_0^{(\mu)}$, we have
    $$
        Q_{n,\mu} := n (\hat{\mu}_1, \hat{\mu}_2, \hat{\gamma}_1,\hat{\gamma}_2)^{\tp} \widehat{\Theta}_{\mu} (\hat{\mu}_1, \hat{\mu}_2, \hat{\gamma}_1,\hat{\gamma}_2) \convd \chi^2_4.
    $$
\end{corollary}
In practice, a consistent estimator for the asymptotic covariance matrix can be found easily by plugging in the MLE's and error rate estimators developed in Section~4.1 of the manuscript.

\begin{proof}[Proof of Corollary~\ref{cor:beta_test}]
    This result follows directly from Theorem~\ref{thm:pivot_beta}, if we can find a consistent (in operator norm) estimator of $\Sigma_{\beta}^{-1}$.

    We propose the natural plug-in estimate
    $$
        \widehat{\Theta}_{\beta} = \left\{ \bm{D}h(\hat{\bm{\xi}}) \begin{pmatrix}
        \frac{n^2}{N_{1}} \Sigma_u(\hat{u}_{1,1},\hat{u}_{2,1}) & \bm{0}  & \bm{0} \\ 
        \bm{0} & \frac{n^2}{N_{2}} \Sigma_u(\hat{u}_{1,2},\hat{u}_{2,2}) & \bm{0}  \\ \bm{0} & \bm{0} & \frac{n^2}{N_{\mathrm{b}}} \Sigma_u(\hat{u}_{1,\mathrm{b}},\hat{u}_{2,\mathrm{b}})
    \end{pmatrix} \bm{D}h(\hat{\bm{\xi}})^{\tp} \right\}^{-1}
    $$
    where
    $$
        \hat{\bm{\xi}} = \begin{pmatrix}
        \hat{u}_{1,1} &
        \hat{u}_{2,1} & 
        \hat{u}_{1,2} &
        \hat{u}_{2,2} & 
        \hat{u}_{1,\mathrm{b}} &
        \hat{u}_{2,\mathrm{b}}
    \end{pmatrix}^{\tp}
    $$
    and
    $$
        \Sigma_u(x,y) = \begin{pmatrix} x(1-x) & (1/2 - x) y \\
        (1/2 - x) y & (1/2 - y) y \end{pmatrix}.
    $$
    By convergence of 
    \begin{equation*}
        \begin{pmatrix}
        \hat{u}_{1,1} &
        \hat{u}_{2,1} & 
        \hat{u}_{1,2} &
        \hat{u}_{2,2} & 
        \hat{u}_{1,\mathrm{b}} &
        \hat{u}_{2,\mathrm{b}}
    \end{pmatrix}^{\tp} \convp \begin{pmatrix}
        q^* &
        \beta q^* &
        q^* &
        \beta q^* &
        q^* &
        \beta q^*
\end{pmatrix}^{\tp}, 
    \end{equation*}
    \begin{equation*}
        \frac{n^2}{N_{1}} \inprob a_{\kappa}^{-1}, \quad \frac{n^2}{N_{2}} \inprob b_{\kappa}^{-1}, \quad \frac{n^2}{N_{\mathrm{b}}} \inprob c_{\kappa}^{-1}
    \end{equation*}
    and the continuous mapping theorem, we have
    $$
        \lVert \widehat{\Theta}_{\beta}^{-1} - \Sigma_{\beta} \rVert_2 \inprob 0.
    $$
    Note that for square symmetric matrices, we have
    $$
        \lVert A^{-1} - B^{-1} \rVert_2 \leq \lVert A^{-1} (B - A) B^{-1} \rVert_2 \leq \frac{\lVert A - B \rVert_2}{\lambda_{\min}(A) \lambda_{\min}(B)} \leq \frac{\lVert A - B \rVert_2}{\lambda_{\min}(A) \left\{ \lambda_{\min}(A) - \lVert A - B \rVert_2 \right\} }.
    $$
    Thus, a sufficient condition for consistency of $\widehat{\Theta}_{\beta}$ is $\lVert \widehat{\Theta}_{\beta}^{-1} - \Sigma_{\beta} \rVert_2 = o_P(1)$, and $\lambda_{\min}(\Sigma_{\beta}) > 0$.
    
    First, recall that under our conditions $q \in (0,1)$ and $\beta \in (0,1)$, we have $\Sigma_u$ is positive definite.
    We also have
    $$
    \bm{D}h(\bm{\xi}^*) = \begin{pmatrix}
       0 & 0 & 0 & 0 & - \frac{\beta}{q^*} & \frac{1}{q^*} \\
       - \frac{\beta}{q^*} & \frac{1}{q^*} & 0 & 0 & \frac{\beta}{q^*} & -\frac{1}{q^*} \\
        0 & 0 & - \frac{\beta}{q^*} & \frac{1}{q^*} & \frac{\beta}{q^*} & -\frac{1}{q^*} 
    \end{pmatrix}
$$
has linearly independent rows, thus under $q \in (0,1)$ and $\beta \in (0,1)$, 
$$
    \lambda_{\min}\left\{ \bm{D}h(\bm{\xi}^*) \bm{D}h(\bm{\xi}^*)^{\tp} \right\} > 0
$$
Then
$$
    \lambda_{\min}(\Sigma_{\beta}) > \lambda_{\min}(\Sigma_{u}) \cdot \lambda_{\min}\left\{ \bm{D}h(\bm{\xi}^*) \bm{D}h(\bm{\xi}^*)^{\tp} \right\} \cdot \min\{a_{\kappa}^{-1},b_{\kappa}^{-1},c_{\kappa}^{-1}\} > 0,
$$
which is sufficient to show consistency of $\widehat{\Theta}_{\beta}$ in operator norm. 

Thus, under $\gamma=0$,
\begin{align*}
        Q_{n,\bar{\beta}}(\beta) &=(n(\hat{\beta}_{\mathrm{w}} - \beta), \sqrt{n}\hat{\gamma}_1, \sqrt{n}\hat{\gamma}_2)^{\tp} (\widehat{\Theta}_{\beta}  - \Sigma_{\beta}^{-1} )(n(\hat{\beta}_{\mathrm{w}} - \beta), \sqrt{n}\hat{\gamma}_1, \sqrt{n}\hat{\gamma}_2) + \cdots\\
        &\quad \cdots + (n(\hat{\beta}_{\mathrm{w}} - \beta), \sqrt{n}\hat{\gamma}_1, \sqrt{n}\hat{\gamma}_2)^{\tp} \Sigma^{-1}_{\beta} (n(\hat{\beta}_{\mathrm{w}} - \beta), \sqrt{n}\hat{\gamma}_1, \sqrt{n}\hat{\gamma}_2) \\
        &\convd \chi^2_2,
\end{align*}
by Slutsky's Theorem, since $\lVert  (n(\hat{\beta}_{\mathrm{w}} - \beta), \sqrt{n}\hat{\gamma}_1, \sqrt{n}\hat{\gamma}_2) \rVert_2 = O_{\prob}(1)$.
\end{proof}

\begin{proof}[Proof of Corollary~\ref{cor:mu_test}]
    We propose a similar plug in estimator for $\Sigma_{\mu}^{-1}$,
    $$
        \widehat{\Theta}_{\mu} =\left\{ \bm{D}g(\hat{\bm{\xi}}) \begin{pmatrix}
        \frac{n^2}{N_{1}} \Sigma_u(\hat{u}_{1,1},\hat{u}_{2,1}) & \bm{0}  & \bm{0} \\ 
        \bm{0} & \frac{n^2}{N_{2}} \Sigma_u(\hat{u}_{1,2},\hat{u}_{2,2}) & \bm{0}  \\ \bm{0} & \bm{0} & \frac{n^2}{N_{\mathrm{b}}} \Sigma_u(\hat{u}_{1,\mathrm{b}},\hat{u}_{2,\mathrm{b}})
    \end{pmatrix} \bm{D}g(\hat{\bm{\xi}})^{\tp} \right\}^{-1}
    $$
    Then the proof is identical to the proof of Corollary~\ref{cor:beta_test}, noting that  $\bm{D}g(\bm{\xi}^*)$ has linearly independent rows for $q \in (0,1)$ and $\beta \in (0,1)$, so that
    $$
        \lambda_{\min}\left\{ \bm{D}g(\bm{\xi}^*) \bm{D}g(\bm{\xi}^*)^{\tp} \right\} > 0.
    $$
\end{proof}

These supporting results are sufficient to give a simple proof of Theorem~3 in the manuscript.

\begin{proof}[Proof of Theorem~3 of the manuscript.]
    Write the null parameter space as
    $$
        H_{0,\bar{\beta}} 
        = \{\theta :\gamma_1=\gamma_2=\beta=0\} \cup \{\theta : \gamma_1=\gamma_2=\mu_1=\mu_2=0\} \cup \left( \cup_{b \in (0,\bar{\beta}]} \{\theta: \beta = b, \gamma_1=\gamma_2=0 \} \right).
    $$
    By the intersection-union principle, we construct a test of this null hypothesis by finding the intersection of the individual rejection regions.

    For $\{\theta :\gamma_1=\gamma_2=\beta=0\}$, the test with test function $\mathbb{I}(Y \neq Y^*)$ is a level $\alpha$ test for any $\alpha$. Moreover, in the remainder of the proof we may assume $Y \neq Y^*$, so the pivotal quantities $Q_{n,\mu}$ and $Q_{n,\beta}$ are well defined.

    For $\{\theta : \gamma_1=\gamma_2=\mu_1=\mu_2=0\}$, by Corollary~\ref{cor:mu_test}, the test with test function
    $$
        \mathbb{I}\left( Q_{n,\mu} > c^{(4)}_{\alpha} \right) 
    $$
    has asymptotic level $\alpha$.

    For an individual element $\{\theta: \beta = b, \gamma_1=\gamma_2=0 \}$ for $b \in (0,\bar{\beta}]$, 
    by Corollary~\ref{cor:mu_test}, the test with test function
    $$
        \mathbb{I}\left( Q_{n,\bar{\beta}}(b) > c^{(3)}_{\alpha} \right) 
    $$
    has asymptotic level $\alpha$. 

    Combining all these rejection regions exactly gives the test defined by $\Psi_{\alpha}$ in the statement of Theorem~3 of the manuscript, and establishes that this test controls type I errors asymptotically at level $\alpha$, by the intersection-union principle.
\end{proof}

\subsection{Technical details for directed networks} \label{app:directed}

Suppose the construct network $A$ is undirected and we observe $Y$ with directed type II (recall) errors, following a group structure
$$
    \mathbb{P}(Y_{ij}=1 \vert A_{ij}=1) = 1 - \beta_{c_ic_j}
$$
for 2 groups ($c_i \in \{1,2\}$). There are $4$ error rate parameters in $(0,1)$. 

Consider the unobserved network generated from SBM$(n,\kappa,B)$ with 
$$
B = 
\begin{pmatrix}
    p_{11} & p_{12} \\
    p_{21} & p_{22}
\end{pmatrix}
$$
with $p_{12} = p_{21}$. 
Thefore, in total, we have seven parameters to estimated under the directed noisy network setting. 

Under this directed noisy network setting, we aim to estimate seven unknown parameters. We define the following estimators:
\begin{enumerate}
\item $\hat{u}_{11,d} = \frac{1}{n_1(n_1-1)} \sum_{c_i=c_j=1}  Y_{ij}$ — observed density within group 1.
\item $\hat{u}_{11,a} = \binom{n_1}{2}^{-1} \sum_{c_i=c_j=1, i < j}  \lvert Y_{ij} - Y_{ji} \rvert$ — asymmetry within group 1.
\item $\hat{u}_{22,d} = \frac{1}{n_2(n_2-1)} \sum_{c_i=c_j=2}  Y_{ij}$ — density within group 2.
\item $\hat{u}_{22,a} = \binom{n_2}{2}^{-1} \sum_{c_i=c_j=2, i < j}  \lvert Y_{ij} - Y_{ji} \rvert$ — asymmetry within group 2.
\item $\hat{u}_{12,d} = \frac{1}{n_1n_2} \sum_{c_i=1,c_j=2} Y_{ij}$ — density from group 1 to group 2.
\item $\hat{u}_{21,d} = \frac{1}{n_1n_2} \sum_{c_i=2,c_j=1} Y_{ij}$ — density from group 2 to group 1.
\item $\hat{u}_{12,a} = \frac{1}{n_1n_2} \sum_{c_i=1,c_j=2} \lvert Y_{ij} - Y_{ji} \rvert$ — asymmetry between groups.
\end{enumerate}

These statistics have expectations:
\begin{enumerate}
    \item $\mathbb{E}\hat{u}_{11,d} = (1 - \beta_{11})p_{11}$
    \item $\mathbb{E}\hat{u}_{11,a} =2\beta_{11}(1-\beta_{11})p_{11}$
    \item $\mathbb{E}\hat{u}_{22,d} = (1 - \beta_{22})p_{22}$
    \item $\mathbb{E}\hat{u}_{22,a} = 2\beta_{22}(1-\beta_{22})p_{22}$
    \item $\mathbb{E}\hat{u}_{12,d} = (1 - \beta_{12})p_{12}$
    \item $\mathbb{E}\hat{u}_{21,d} = (1 - \beta_{21})p_{12}$ 
    \item $\mathbb{E}\hat{u}_{12,a} = [\beta_{12}(1-\beta_{21}) + \beta_{21}(1-\beta_{12})]p_{12}$
\end{enumerate}

By method of moments, the estimators for the unknown parameters are:
\begin{align*}
    \hat\beta_{11} &= \frac{\hat{u}_{11,a}}{2\hat{u}_{11,d}}, \quad
    \hat{p}_{11} = \frac{2\hat{u}_{11,d}^2}{2\hat{u}_{11,d} - \hat{u}_{11,a}}, \\ 
    \hat\beta_{22} &= \frac{\hat{u}_{22,a}}{2\hat{u}_{22,d}}, \quad
    \hat{p}_{22} = \frac{2\hat{u}_{22,d}^2}{2\hat{u}_{22,d} - \hat{u}_{22,a}}, \\
    \hat\beta_{12} &= \frac{\hat{u}_{21,d} - \hat{u}_{12,d} + \hat{u}_{12,a}}{2\hat{u}_{21,d}}, \quad  \hat\beta_{21} = \frac{\hat{u}_{12,d} - \hat{u}_{21,d} + \hat{u}_{12,a}}{2\hat{u}_{12,d}}, \\
    \hat{p}_{12} &= \frac{2\hat{u}_{12,d}\hat{u}_{21,d}}{\hat{u}_{12,d} + \hat{u}_{21,d} - \hat{u}_{12,a}}.
\end{align*}

We now consider testing for systematic bias in edge recall rates across groups. The hypotheses are:
$$
H_{1,0}: \beta_{12} = \beta_{21}; \qquad H_{1,1}: \beta_{12} < \beta_{21}.
$$
The test statistic is based on the asymptotic distribution of $\hat{\beta}_{12} - \hat{\beta}_{21}$:
\begin{theorem} \label{thm:directed_beta}
    Under the directed noisy network model, we have 
    $$
        \sqrt{n_1n_2} 
        \left[(\hat{\beta}_{12} - \hat{\beta}_{21}) - (\beta_{12} - \beta_{21})\right] 
        \convd \mathcal{N}\left( 0, \sigma_{\operatorname{d}}^2\right),
    $$
    for a postive variance $\sigma_{\operatorname{d}}^2$.
\end{theorem}

\begin{proof}
    In the proof, we establish a stronger version of the joint asymptotic normality for the empirical moments of $\hat{u}_{11,d}, \hat{u}_{11,a}, \hat{u}_{22,d}, \hat{u}_{22,a}, \hat{u}_{12,d}, \hat{u}_{21,d}, \hat{u}_{12,a}$. This result may also be useful for more general hypothesis testing problems beyond the one considered here.

    First, if the unobserved network $A$ is generated from SBM$(n, B, \kappa)$ with 
    $$
    B = \begin{pmatrix}
        p_{11} & p_{12}\\
        p_{21} & p_{22}
    \end{pmatrix},
    $$
    where $p_{12} = p_{21}$. 

    According to our model assumption, we observed a noisy directed network $Y$ with 
    $$
    \mathbb{P}(Y_{ij}=1 \vert A_{ij}=1) = 1 - \beta_{c_ic_j}. 
    $$
    Marginally, we can think about $Y$ generated from directed SBM$(n,B^*,\kappa)$ with 
    $$
    B^* = \begin{pmatrix}
        p_{11}^* & p_{12}^* \\
        p_{21}^* & p_{22}^* \\
    \end{pmatrix}
    $$
    with 
    \begin{align*}
        p_{11}^* &= p_{11}(1-\beta_{11}), \quad
        p_{12}^* = p_{12}(1-\beta_{12}), \\ 
        p_{21}^* &= p_{21}(1-\beta_{21}), \quad
        p_{22}^* = p_{22}(1-\beta_{22}).
    \end{align*}
    Here, we abuse the notation and definition of SBM as $Y_{ij}$ and $Y_{ji}$ are dependent by nature and $Y$ is not exactly from a SBM. 


    Similar to Lemma 2 in the manuscript, we can show the following asymptotic results of MOM estiamtors: 
    $$
    \sqrt{n_1(n_1-1)}
    \begin{pmatrix}
        \hat{u}_{11,d} - p_{11}^* \\
        \hat{u}_{11,a} - 2p_{11}^*\beta_{11} 
    \end{pmatrix} 
    \convd \mathcal{N}\left( \bm{0}, \Sigma_{11} \right),
    $$
    with 
    $$
    \Sigma_{11} = 
    \begin{pmatrix}
       p_{11}^*(1-p_{11}^*) & 2p_{11}^*\beta_{11}(1-2p_{11}^*) \\  
       2p_{11}^*\beta_{11}(1-2p_{11}^*) & 4p_{11}^*\beta_{11}(1-2p_{11}^*\beta_{11})
    \end{pmatrix},
    $$
    $$
    \sqrt{n_2(n_2-1)}
    \begin{pmatrix}
        \hat{u}_{22,d} - p_{22}^* \\
        \hat{u}_{22,a} - 2p_{22}^*\beta_{22} 
    \end{pmatrix} 
    \convd \mathcal{N}\left( \bm{0}, \Sigma_{22} \right),
    $$
    with 
    $$
    \Sigma_{22} = 
    \begin{pmatrix}
       p_{22}^*(1-p_{22}^*) & 2p_{22}^*\beta_{22}(1-2p_{22}^*) \\  
       2p_{22}^*\beta_{22}(1-2p_{22}^*) & 4p_{22}^*\beta_{22}(1-2p_{22}^*\beta_{22})
    \end{pmatrix},
    $$
    and 
    $$
    \sqrt{n_1n_2}
    \begin{pmatrix}
        \hat{u}_{12,d} - p_{12}^* \\
        \hat{u}_{21,d} - p_{21}^* \\
        \hat{u}_{12,a} - p_{12}^*\beta_{21} - p_{21}^*\beta_{12} 
    \end{pmatrix}  
    \convd \mathcal{N}\left( \bm{0}, \Sigma_{33} \right),
    $$
    with 
    $$
    \hspace{-1.2cm}
    \Sigma_{33} = 
    \begin{pmatrix}
       p_{12}^*(1-p_{12}^*) & p_{12}^*p_{21}^*/p_{12} - p_{12}^*p_{21}^* & 
       p_{12}^*\beta_{21} - p_{12}^*\left(p_{12}^*\beta_{21} +  p_{21}^*\beta_{12}\right)
       \\  
       p_{12}^*p_{21}^*/p_{12} - p_{12}^*p_{21}^* & p_{21}^*(1-p_{21}^*) & 
       p_{21}^*\beta_{12} - p_{21}^*\left( p_{12}^*\beta_{21} + p_{21}^*\beta_{12} \right)
       \\
       p_{12}^*\beta_{21} - p_{12}^*\left(p_{12}^*\beta_{21} +  p_{21}^*\beta_{12}\right) & p_{21}^*\beta_{12} - p_{21}^*\left( p_{12}^*\beta_{21} + p_{21}^*\beta_{12} \right) & \left( p_{12}^*\beta_{21} + p_{21}^*\beta_{12}\right)\left( 1- p_{12}^*\beta_{21} - p_{21}^*\beta_{12} \right)
    \end{pmatrix}.
    $$    
    For testing the null $\beta_{12} = \beta_{21}$, we consider the asymptotic distribution of $\hat\beta_{12} - \hat\beta_{21}$
    where $\hat\beta_{12}$ and $\hat\beta_{21}$ defined above as 
    $$
        \hat\beta_{12} = \frac{\hat{u}_{21,d} - \hat{u}_{12,d} + \hat{u}_{12,a}}{2\hat{u}_{21,d}}, \quad  \hat\beta_{21} = \frac{\hat{u}_{12,d} - \hat{u}_{21,d} + \hat{u}_{12,a}}{2\hat{u}_{12,d}}.
    $$
    To apply the delta method, we consider the transformation function 
    $$
        h(x,y,z) = \frac{y-x+z}{2y} - \frac{x-y+z}{2x} 
    $$
    with first derivative 
    $$
        \bm{D}h(x,y,z) = \frac{1}{2}\left(
            -\frac{y}{x^2} + \frac{z}{x^2} - \frac{1}{y}, \quad 
            \frac{1}{x} - \frac{z}{y^2} + \frac{x}{y^2}, \quad 
            \frac{1}{y} - \frac{1}{x} \right).
    $$
    For the expectation, we have 
    $$
    h(p_{12}^*, p_{21}^*, p_{12}^*\beta_{21} + p_{21}^*\beta_{12}) = \beta_{12} - \beta_{21}. 
    $$
    Finally we have 
    $$
        \sqrt{n_1n_2} 
        \left[(\hat{\beta}_{12} - \hat{\beta}_{12}) - (\beta_{12} - \beta_{21})\right] 
        \convd \mathcal{N}\left( 0, \sigma_{\operatorname{d}}^2\right)
    $$
    with 
    $$
    \sigma_{\operatorname{d}}^2 = \bm{D}h(x,y,z)\Sigma_{33}\bm{D}h(x,y,z)^T. 
    $$
    
\end{proof}
Using Theorem~\ref{thm:directed_beta}, we can derive the following result for hypothesis testing.
\begin{corollary} \label{cor:directed_test}
    Under the assumptions of Theorem~\ref{thm:directed_beta}, suppose $\hat{\sigma}_{\operatorname{d}}$ is a consistent estimator of $\sigma_{\operatorname{d}}$. Then under $\beta_{12}=\beta_{21}$, we have
    $$
        Z = \frac{\sqrt{n_1n_2}\left( \hat{\beta}_{12} - \hat{\beta}_{21} \right)}{\hat{\sigma}_{\operatorname{d}}} \convd \cN(0,1). 
    $$
\end{corollary}

After achieving consistent estimators for the parameters in generating the construct network $A$, we can correct the representation of the minority using the asymptotic limits we derived in Section 4 of the manuscript. 

\section{Plug-in estimation of minority representation} \label{subsubsec:sbm_sims}

Here we provide some simulations to check the quality of our asymptotic approximation, and show how it can be used to produce a parametric plug-in estimator for the expected minority representation profile $\rho_K$. Each plot is based on 200 independent replicates.

Throughout we will refer to $R_K(\bm{c},A)$ as the ``empirical'' $\rho_K$ estimator, defined above as the proportion of minority group nodes among the top $K$ according to degree ranking.
By Theorem~2 of the manuscript, $R_{\lfloor nz \rfloor}(\bm{c},A)$ converges in probability for all $z$, and thus for a scaled sequence of $K$'s, the empirical $\rho_K$ estimator is consistent for $\rho^*(z)$. 
However, for a fixed $K$, we have no such convergence guarantee.

We will also refer to a ``parametric plug-in'' estimate for $\rho_K$, which uses the MLE under SBM:
\begin{align*}
    \hat{\kappa} &= \frac{n_1}{n} \\
    \hat{p}_1 &= \frac{1}{{n_1 \choose 2}} \sum_{i < j ~:~ c_i=c_j=1} A_{ij}  \\
    \hat{p}_2 &=  \frac{1}{{n_2 \choose 2}} \sum_{i < j ~:~ c_i=c_j=2} A_{ij}  \\
    \hat{q} &= \frac{1}{n_1n_2} \sum_{i < j ~:~ c_i=1,c_j=2} A_{ij} \\   
    \hat{\mu}_1 &= \sqrt{n}(\hat{p}_1 - \hat{q}) \\
    \hat{\mu}_2 &= \sqrt{n}(\hat{p}_2 - \hat{q}) \
\end{align*}
where
$$
    n_1 = \lvert \{1 \leq i \leq n : c_i=1\}, \quad n_2 = n - n_1.
$$
We then evaluate the asymptotic approximation
\begin{equation} \label{rho_plugin}
    \rho^*(K/n;\hat{\kappa},\hat{q},\hat{\mu}_1,\hat{\mu}_2)
\end{equation}
after plugging in these estimates.

It is easy to see that \eqref{rho_plugin} is a consistent estimator of the asymptotic minority representation profile for the related SBM.
From Theorem~2 of the manuscript, $\rho^*$ is continuous in both the quantile $z = \lim_{n \rightarrow \infty} K/n$ and the model parameters $(\kappa,q,\mu_1,\mu_2)$.
We also have that $\hat{\kappa}$ and $\hat{q}$ are consistent by the weak law of large numbers.
By the CLT, we have $\hat{p}_g - p_g = O(n^{-1})$ and $\hat{q} - q = O(n^{-1})$, thus $\hat{\mu}_g$ is also a consistent estimator of $\mu_g$ for $g=1,2$.

In Figure~\ref{fig:sbmrho}, we show the performance of these two estimators, and the asymptotic limits under SBM settings with $n=200$ and $n=1000$ nodes.
In Figure~\ref{fig:sbmrmse} we compare the two estimators in terms of RMSE for $\rho_K$.

\begin{figure}[ht]
        \centering
        \begin{subfigure}[b]{0.475\linewidth}
            \centering
            \includegraphics[width=\linewidth]{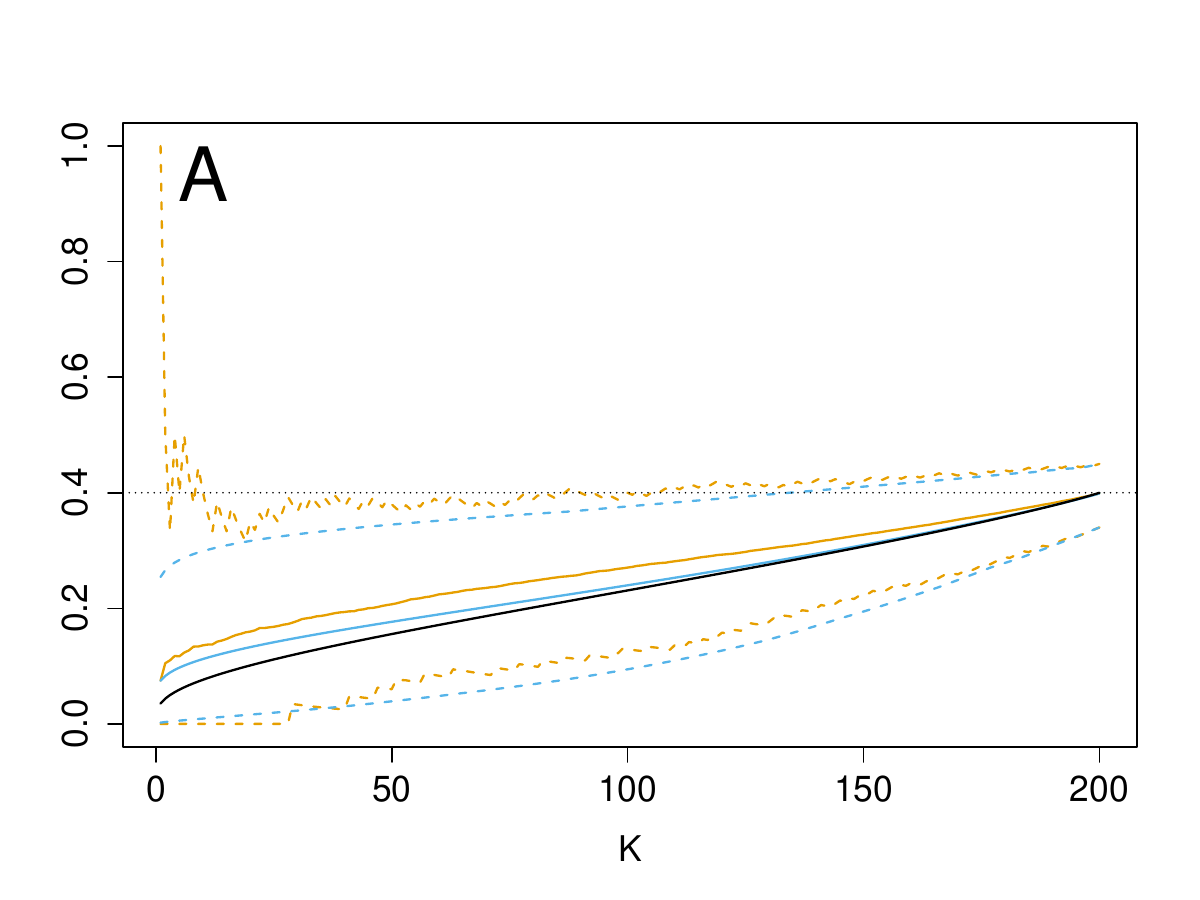}
            \caption{$n=200$, $\mu_1=\mu_2=1$.}
        \end{subfigure}
        \hfill
        \begin{subfigure}[b]{0.475\linewidth}  
            \centering 
            \includegraphics[width=\linewidth]{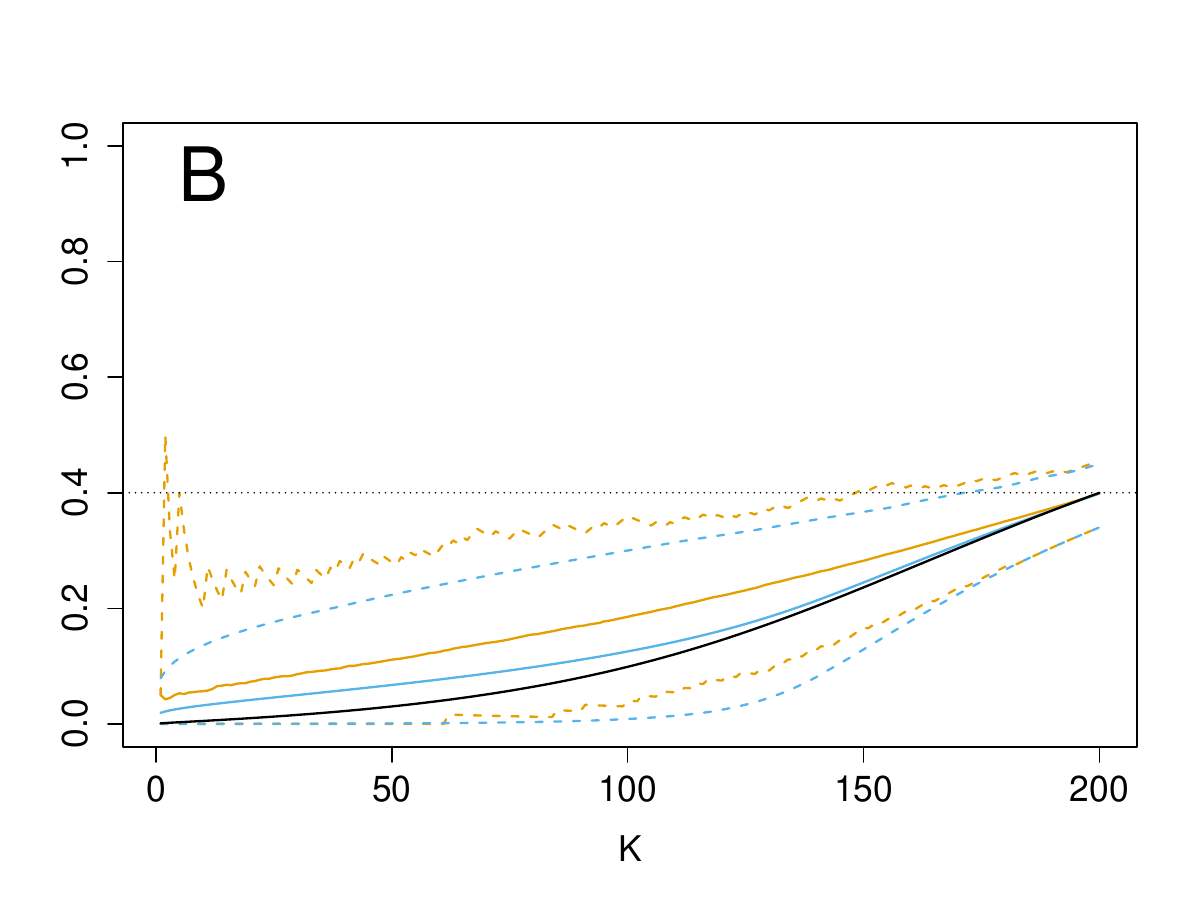}
            \caption{$n=200$, $\mu_1=\mu_2=2$.}  
        \end{subfigure}
        \medskip
        \begin{subfigure}[b]{0.475\linewidth}   
            \centering 
            \includegraphics[width=\linewidth]{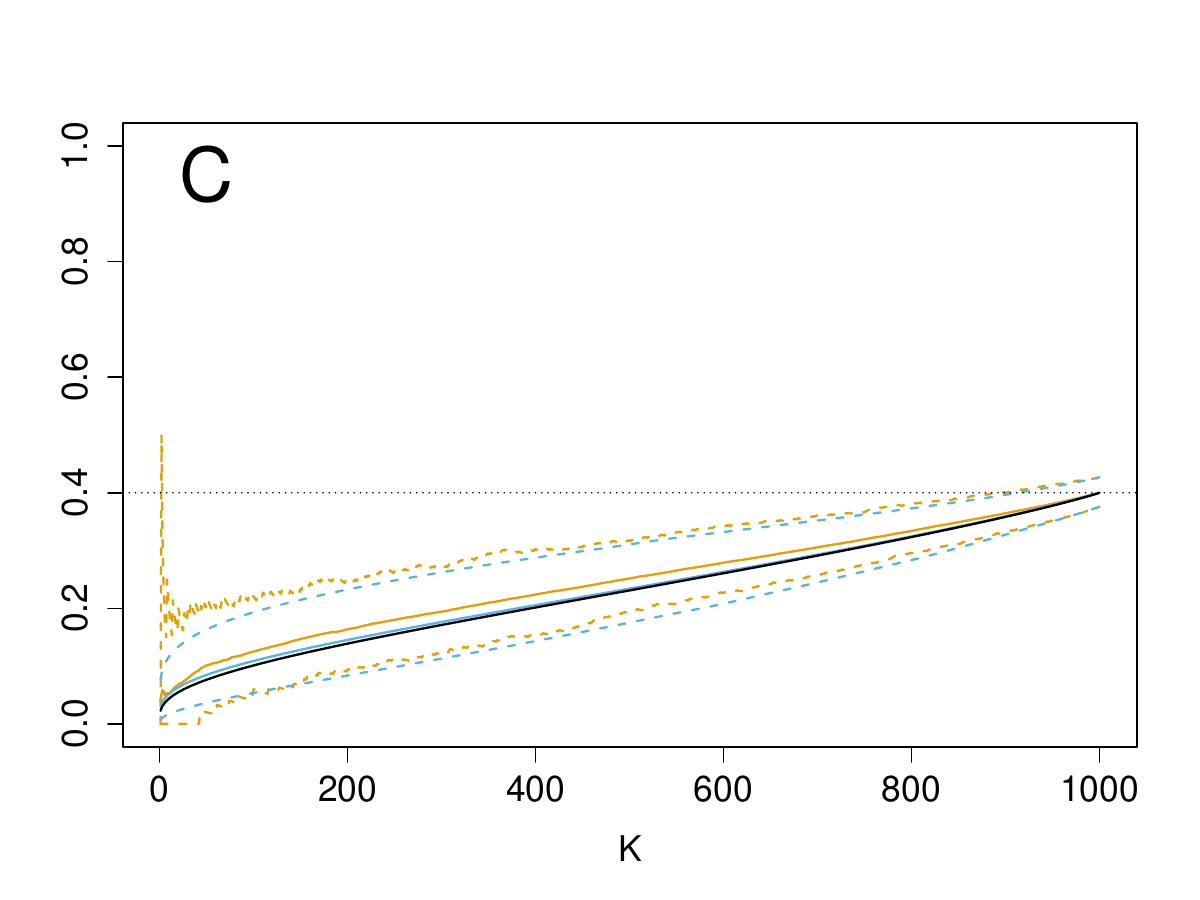}
            \caption{$n=1000$, $\mu_1=\mu_2=1$.}  
        \end{subfigure}
        \hfill
        \begin{subfigure}[b]{0.475\linewidth}   
            \centering 
            \includegraphics[width=\linewidth]{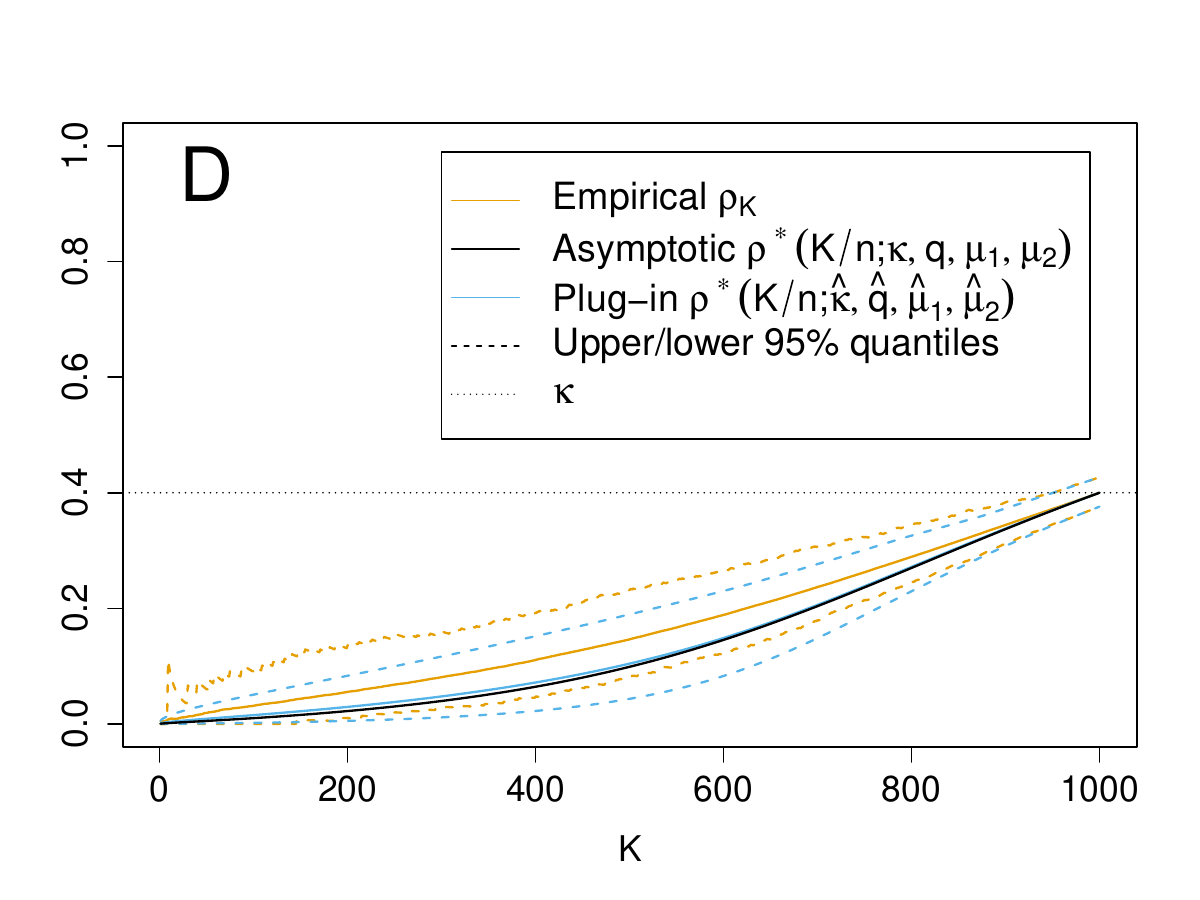}
            \caption{$n=1000$, $\mu_1=\mu_2=2$.} 
        \end{subfigure}
        \caption{Estimation and asymptotic approximation of minority representation profiles under SBM, $\kappa=0.4$, $q=0.05$.} 
        \label{fig:sbmrho}
\end{figure}

\begin{figure}[ht]
        \centering
        \begin{subfigure}[b]{0.475\linewidth}
            \centering
            \includegraphics[width=\linewidth]{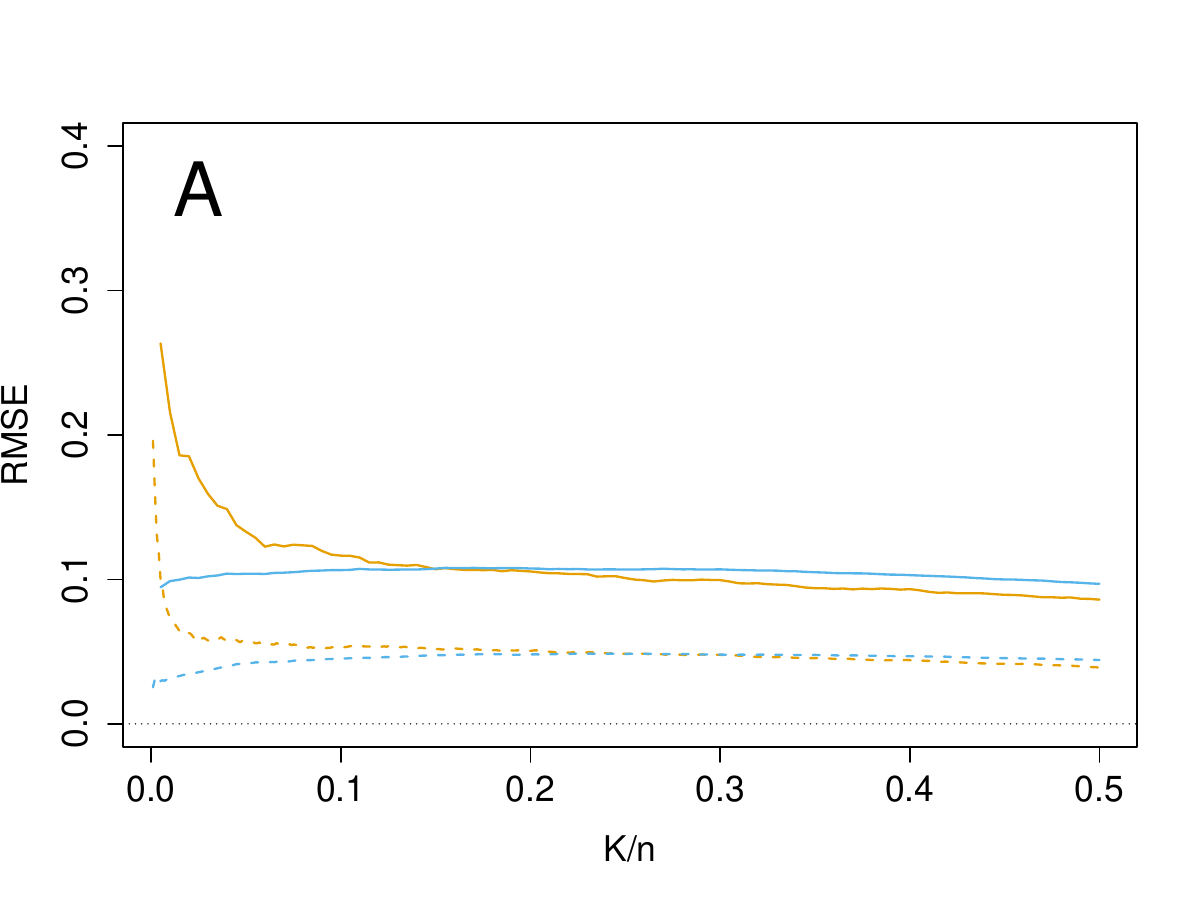}
            \caption{$\mu_1=\mu_2=1$}
        \end{subfigure}
        \hfill
        \begin{subfigure}[b]{0.475\linewidth}  
            \centering 
            \includegraphics[width=\linewidth]{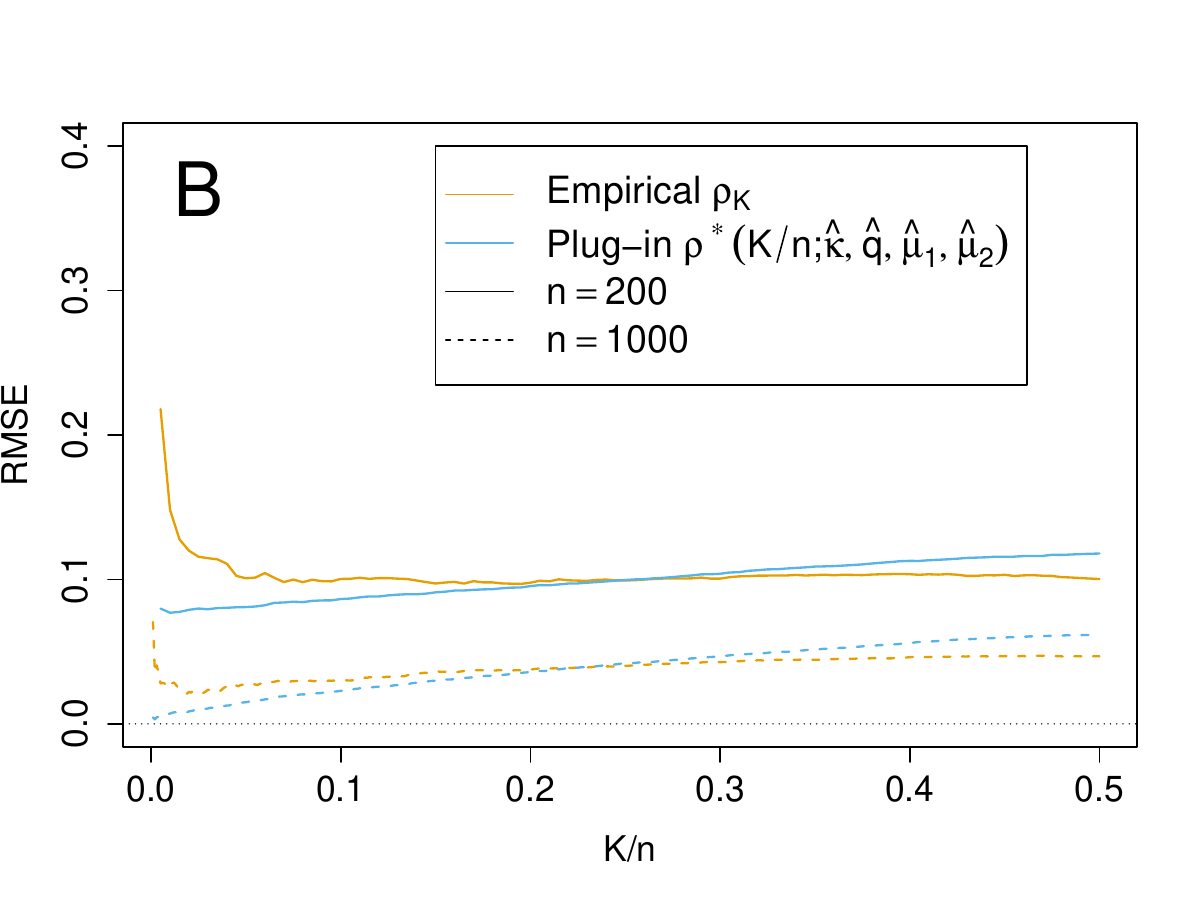}
            \caption{$\mu_1=\mu_2=2$}  
        \end{subfigure}
        \caption{RMSE performance of empirical and plug-in estimators of $\rho_K$ under SBM, $\kappa=0.4$, $q=0.05$, $n \in \{200,1000\}$.} 
        \label{fig:sbmrmse}
\end{figure}


For $\mu=1$, even for relatively small networks with $n=200$, the asymptotic approximation is satisfactory.
For $\mu=2$, the quality of the asymptotic approximation improves as $n$ increases.
Note that in all settings, the parametric plug-in agrees well with the limiting minority representation profile, as expected.
Estimation performance for $\rho_K$ is improved with larger numbers of nodes when $K$ is proportional to $n$, but the performance is comparable for small, fixed $K$.
From Figure~\ref{fig:sbmrmse}, we see that values of $K/n$ close to $0$ show the greatest advantage of estimation with the parametric plug-in compared to the empirical ranking, while the empirical $\rho_K$ profile tends to perform better for larger $K/n$ close to $1/2$.

\end{document}